\documentclass[11pt]{article}

\usepackage{float}
\usepackage{enumitem}
\usepackage{comment}
\usepackage[utf8]{inputenc}
\usepackage{lipsum}
\usepackage[table,xcdraw]{xcolor}
\usepackage{amsmath, amssymb,amsfonts,amsthm,mathtools,nicefrac}
\numberwithin{equation}{section}
\usepackage{xspace,graphicx,relsize,bm}
\usepackage{soul} 

\usepackage{parskip}  
\usepackage{comment}
\usepackage[most]{tcolorbox}

\usepackage{libertinus}  
\usepackage[zerostyle=a,scaled=.97]{newtxtt}  
\usepackage[T1]{fontenc}
\usepackage[libertine]{newtxmath}  
\usepackage{dsfont}  
\usepackage[cal=stix, calscaled=.97]{mathalpha}  
\usepackage{anyfontsize}
\usepackage{booktabs}
\usepackage{tabularx}

\usepackage{enumitem}
\setlist[enumerate,1]{
 label=\textup{(\roman*)},
 ref=\textup{(\roman*)}
}

\usepackage{url}
\usepackage[
    backend=biber,
    style=alphabetic,
    sorting=nyt,
    maxbibnames=99
    ]{biblatex}
\DefineBibliographyStrings{english}{%
  backrefpage = {page},
  backrefpages = {pages},
  mathesis = {MSc thesis},
}

\usepackage{tikz}
\usetikzlibrary{arrows.meta, positioning, decorations.pathreplacing}

\usepackage[linktocpage=true]{hyperref}
\usepackage[margin=1in]{geometry}
\usepackage[singlespacing]{setspace}
\hypersetup{
	colorlinks,
	linkcolor={purple},
	citecolor={cyan},
	urlcolor={purple}
}

\usepackage{subcaption}
\usepackage{mleftright}
\usepackage{tipa}
\usepackage{multirow}
\usepackage{physics}  
\usepackage{authblk}  

\usepackage{subfiles} 

\usepackage{algorithm}
\usepackage{algpseudocode}

\usepackage{multirow}

\usepackage{pifont}

\usepackage{float}

\usepackage{chngpage}

\usepackage[capitalize,noabbrev]{cleveref}  

\usepackage{longtable}
\usepackage{lscape}

\usepackage{algorithm}
\usepackage{algpseudocode}

\newcommand{\F}{\mathbb F}
\newcommand{\E}{\mathbb E}
\newcommand{\Prb}{\mathbb P}
\newcommand{\poly}{\operatorname{poly}}

\newcommand{\scor}{\operatorname{score}}

\newcommand{\ceil}[1]{\left\lceil #1\right\rceil}
\newcommand{\floor}[1]{\left\lfloor #1\right\rfloor}
\newcommand{\LR}{\operatorname{LR}}

\newcommand{\calO}{\mathcal O}

\definecolor{choiceone}{RGB}{100,160,220}
\definecolor{choicetwo}{RGB}{20,60,140}
\definecolor{classical}{HTML}{D55E00}
\definecolor{quantum}{HTML}{0072B2}
\definecolor{neutral}{HTML}{4D4D4D}
\definecolor{gap}{HTML}{009E73}

\providecommand{\myvec}[1]{\ensuremath{\boldsymbol{#1}}}

\providecommand{\yy}{\ensuremath{\myvec{y}}}

\newtheoremstyle{mydefinitionsty}
{10pt}
{10pt}
{\itshape}
{}
{}
{}
{.5em}
{\textbf{\thmname{#1}~\thmnumber{#2}:  }\thmnote{(#3)}}
\theoremstyle{mydefinitionsty}
\newtheorem{definition}{Definition}
\newtheorem{remark}{Remark}

\makeatother

\newtheorem{innercustomgeneric}{\customgenericname}
\providecommand{\customgenericname}{}
\newcommand{\newcustomtheorem}[2]{%
  \newenvironment{#1}[1]
  {%
   \renewcommand\customgenericname{#2}%
   \renewcommand\theinnercustomgeneric{##1}%
   \innercustomgeneric
  }
  {\endinnercustomgeneric}
}

\newcustomtheorem{customthm}{Theorem}
\newcustomtheorem{customlemma}{Lemma}
\newcustomtheorem{customproposition}{Proposition}

\newtheoremstyle{mythmsty}
{10pt}
{10pt}
{\itshape}
{}
{}
{}
{.5em}
{\textbf{\thmname{#1}~\thmnumber{#2}:  }\thmnote{(#3)}}
\theoremstyle{mythmsty}
\newtheorem{inftheorem}{Informal Theorem}
\newtheorem{theorem}{Theorem}
\newtheorem{lemma}{Lemma}

\newtheorem{corollary}{Corollary}

\newtheorem{proposition}{Proposition}

\definecolor{qbaccent}{HTML}{4C566A}   
\definecolor{qbfill}{HTML}{F4F5F7}     

\newtcolorbox{questionbox}[1][]{
    enhanced,
    breakable,
    colback=qbfill,
    colframe=qbfill,          
    boxrule=0pt,
    arc=3pt,
    outer arc=3pt,
    left=16pt, right=14pt, top=12pt, bottom=12pt,
    borderline west={2.5pt}{0pt}{qbaccent}, 
    fontupper=\itshape,
    #1
}

\definecolor{max}{RGB}{164,12,52}

\newcommand{\stkout}[1]{\ifmmode\text{\st{\ensuremath{#1}}}\else\st{#1}\fi}

\setlist[enumerate,1]{
 label=\textup{(\roman*)},
 ref=\textup{(\roman*)}
}

\newif\ifverbose

\definecolor{choiceone}{RGB}{100,160,220} 
\definecolor{choicetwo}{RGB}{20,60,140} 
\definecolor{neutral}{HTML}{4C566A} 
\definecolor{classical}{HTML}{4C78A8} 
\definecolor{quantum}{HTML}{B24C88} 
\definecolor{gap}{HTML}{6B5AA6} 

\title{A provable quantum advantage for approximate optimization\\ via decoded quantum interferometry}

\newcommand{\dccqs}{Dahlem Center for Complex Quantum Systems, Freie Universit{\"a}t Berlin, 14195 Berlin, Germany}
\newcommand{\hzb}{Helmholtz-Zentrum Berlin f{\"u}r Materialien und Energie, 14109 Berlin, Germany}
\newcommand{\hhi}{Fraunhofer Heinrich Hertz Institute, 10587 Berlin, Germany}

\author[1]{Maximilian~J.~Kramer\footnote{Corresponding author: \href{mailto:m.kramer@fu-berlin.de}{m.kramer@fu-berlin.de}}}
\author[1,2]{Elies~Gil-Fuster}
\author[1]{Benjamin~D.~M.~Jones}
\author[1,2,3]{\\Jens~Eisert}
\author[1]{Franz~J.~Schreiber}
\affil[1]{\dccqs}
\affil[2]{\hhi}
\affil[3]{\hzb}

\date{}

\begin{document}

\maketitle

\begin{abstract}
    Decoded quantum interferometry (DQI) is a novel paradigm for tackling approximate optimization problems on quantum computers. This framework comes with strong performance guarantees and exploits a well-established duality between optimization and coding theory. A central question, however, is whether DQI can actually provably outperform all polynomial-time classical algorithms. In this work, we establish such an advantage in an oracle setting: we consider an optimization task called folded optimal polynomial intersection (folded OPI), where the acceptance sets are chosen randomly and accessed through membership oracles. We establish a strict gap between the approximation ratio achievable by any polynomial-time classical algorithm and the approximation ratio achieved by the DQI algorithm. Our proof builds on Jordan et al.'s DQI framework for approximate optimization and extends the classical lower-bound method underlying Yamakawa and Zhandry's exact-search oracle separation to approximation. Building on recent developments by Sun and Wootters, Horinaga and Yamakawa, and Jo, we further show that a modified version of the DQI algorithm achieves a strictly higher score guarantee on typical sampled folded OPI instances. As a concrete example, for code rate $0.3$, DQI achieves expected scores of approximately $0.85$. In contrast, exceeding the classical threshold of $0.65$ by any fixed amount with constant probability on sampled instances requires super-polynomially many classical membership queries.
\end{abstract}

\newpage
\tableofcontents

\newpage

\section{Introduction}
\label{sec:introduction}

Combinatorial optimization is a prominent prospective application of quantum computing, yet rigorous quantum advantages remain difficult to establish. Optimization problems arise throughout science and industry, where even modest improvements in solution quality or computational cost can have substantial impact. The central challenge is to determine whether quantum algorithms can offer such improvements over arbitrary classical algorithms~\cite{Abbas_2024,MindTheGaps,VastWorld}.
\emph{Decoded quantum interferometry} (DQI) is a recently proposed quantum algorithmic framework that turns a coherent decoder for the dual of a linear code into a quantum algorithm for optimizing local predicates over the corresponding primal code~\cite{JordanEtAl2025DQI}. DQI connects these optimization and decoding tasks through a Fourier transform over the underlying finite field, building on the approach pioneered in Regev's quantum reductions~\cite{ATS03,Regev2004,AR2005,Regev2009}. Most prominent is its application to the \emph{optimal polynomial intersection} (OPI) problem, where the primal code is a \emph{Reed--Solomon code} and the dual is a \emph{generalized Reed--Solomon code}. In certain parameter regimes, DQI achieves better approximation guarantees than any \emph{known} polynomial time classical algorithm~\cite{JordanEtAl2025DQI}. Optimized quantum circuits and concrete resource estimates have further strengthened OPI as a candidate application for demonstrating quantum advantages via DQI~\cite{khattar2025verifiablequantumadvantageoptimized,rosmanis2026nearlylineartimedecodedquantum}.

Whether these approximation guarantees reflect a \emph{genuine quantum advantage}, in the sense of a \emph{provable separation} from \emph{arbitrary} classical algorithms, remains a central open question~\cite{Aaronson2024DQI,MarwahaEtAl2026} and is at the same time a highly contested challenge in the field. We address the following oracle formulation of this question:
\begin{quote}
    \emph{Can the approximation performance of DQI be separated in query complexity from that of arbitrary classical algorithms in a common oracle model?}
\end{quote}

We answer this question affirmatively by constructing a family of \emph{folded OPI instances} based on \emph{folded Reed--Solomon codes}. Each instance consists of a code from this family and a collection of acceptance sets, one per block, each containing exactly half of the block alphabet. These acceptance sets are accessed through separately callable block-membership oracles.
At every fixed rational code rate, this family exhibits a sharp classical threshold: a polynomial-time implementation of Prange's information-set strategy~\cite{Prange1962} reaches this threshold using linearly many membership queries, whereas achieving any fixed improvement with constant success probability requires super-polynomially many \emph{classical} membership queries.
DQI converges to an expected score strictly above the classical threshold at every fixed rational code rate $R\in (0,1)$.
This requires only conventional Reed--Solomon unique decoding based on the Berlekamp--Massey algorithm~\cite{Berlekamp1968,Massey1969}, correcting errors strictly below half the minimum distance of the unfolded dual code.
The score guarantee holds uniformly over all promised instances, and the algorithm uses one coherent query to each block-membership oracle.
Verification and a constant number of repetitions then yield a bounded-error quantum algorithm attaining a score above the classical threshold.
Together with the classical query lower bound, this establishes a query separation for approximate optimization in this oracle model via DQI.
The corresponding question of proving hardness for scalar OPI, based on standard computational assumptions, as studied in Refs.~\cite{JordanEtAl2025DQI,SunWootters2026,HorinagaYamakawa2026,Jo2026}, remains open.

Our classical query lower bound adapts the exact-search argument of Yamakawa and Zhandry~\cite{YamakawaZhandry2024} to approximation in our balanced membership-oracle model. Their original argument by itself does not rule out a candidate satisfying sufficiently many of the remaining blocks to exceed a certain score. We supplement their framework with a standard conditional binomial-tail estimate, which provides the additional control needed to extend their exact-search query lower bound to approximation.

\Cref{fig:advantage_region} schematically places the present result relative to prior scalar OPI analyses for DQI and beyond~\cite{JordanEtAl2025DQI,SunWootters2026,HorinagaYamakawa2026,Jo2026} and the Yamakawa--Zhandry exact-search separation~\cite{YamakawaZhandry2024}.

\begin{figure}[t]
\centering
\includegraphics[width=0.85\linewidth]{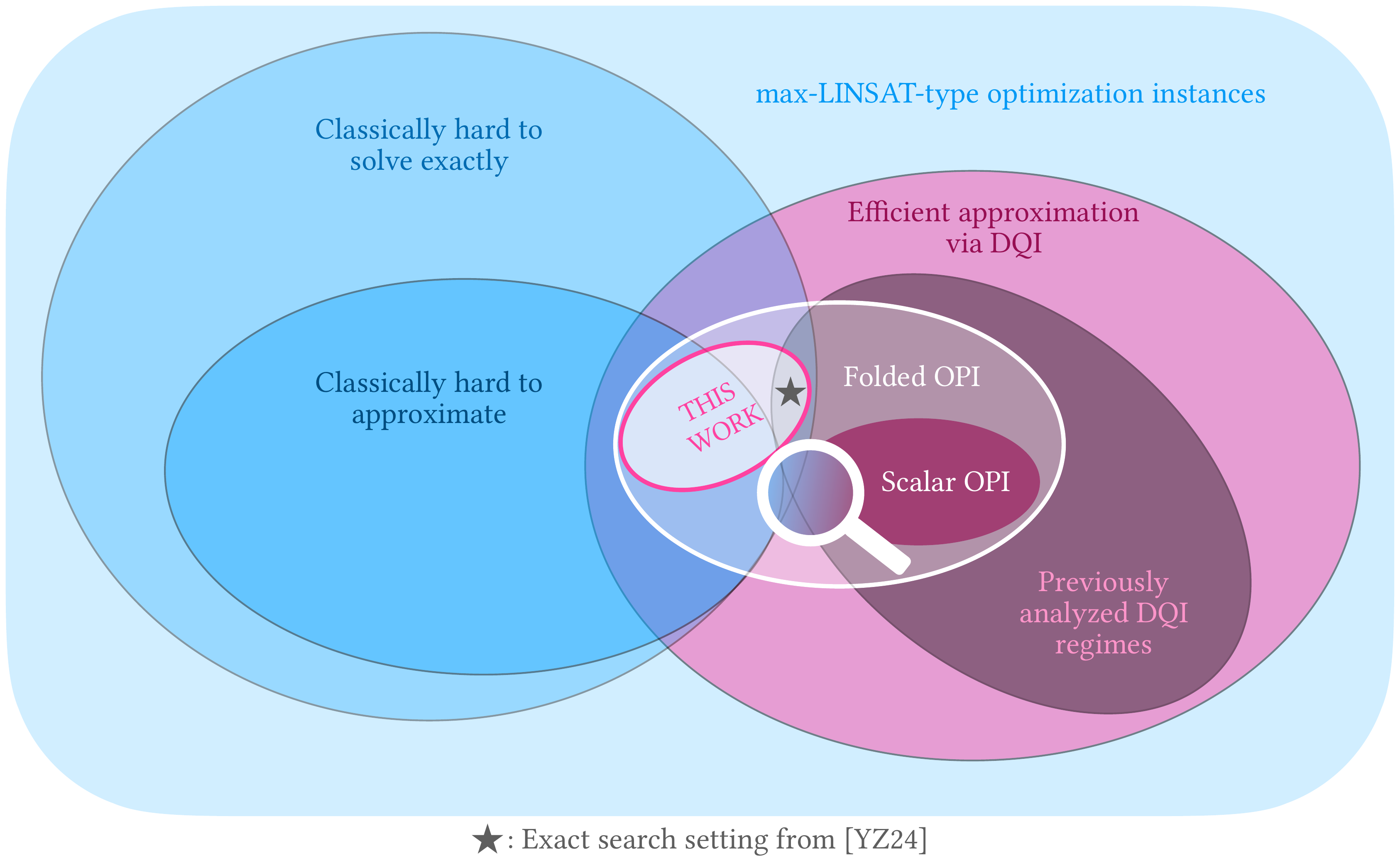}
\caption{Schematic context for the present results. The magnifying-glass lens marks the unresolved status of quantum advantage for scalar OPI: existing DQI analyses outperform the best \emph{known} classical algorithms, but do not establish a lower bound against \emph{arbitrary} classical algorithms. The star marks Yamakawa and Zhandry's exact-search oracle separation on folded Reed--Solomon instances. The present work establishes an approximation separation for DQI on folded OPI. The regions are conceptual and do not denote literal set inclusions across computational models.}
\label{fig:advantage_region}
\end{figure}

\subsection{Main result}

We prove the separation on an explicit family of balanced folded-OPI instances. Here, $M$ is the number of block constraints and the growing instance-size parameter, while $R\in(0,1)$ is the fixed target rate of the underlying folded Reed--Solomon code family. For each candidate polynomial, every block tests whether its folded evaluation symbol belongs to an oracle-defined acceptance set, and the score is the fraction of blocks that accept. In our concrete instantiation, every block accepts exactly half of the block alphabet. We focus on this balanced setting because it makes the separation particularly transparent and clean. For the distributional classical statements, the acceptance sets are sampled independently and uniformly from the half-size subsets of the block alphabet.

The two relevant score levels are
\begin{align}
  \alpha_{\rm C}(R)&\coloneq \frac{1+R}{2},
  &
  \alpha_{\rm Q}(R)&\coloneq \frac{1+\sqrt{R(2-R)}}2.
  \label{eq:intro-score-levels}
\end{align}
Prange's information-set strategy~\cite{Prange1962} attains the classical threshold $\alpha_{\rm C}(R)$ using linearly many membership queries in $M$. Exceeding this threshold by any fixed positive amount with constant probability requires super-polynomially many classical queries in $M$. DQI converges to the strictly higher expected score $\alpha_{\rm Q}(R)$~\cite{JordanEtAl2025DQI}.
Below we give an informal version of \cref{thm:frs-separation}.

\begin{inftheorem}[Balanced DQI separation]
\label{thm:informal-frs-separation}
For every fixed rational rate $R\in(0,1)$, there is an explicit balanced folded-OPI family indexed by block counts $M$, with $M+1$ a power of two, such that for all sufficiently large $M$:
\begin{enumerate}
  \item DQI runs in quantum polynomial time, uses $M$ coherent membership queries, and achieves an expected score $\alpha_{\rm Q}(R)+o(1)$ uniformly over every promised instance.
  \item Prange's information-set algorithm runs in classical polynomial time, uses $\calO(M)$ membership queries, and achieves a score converging in probability to $\alpha_{\rm C}(R)$ under the sampled-oracle distribution.
  \item Any classical algorithm that exceeds $\alpha_{\rm C}(R)$ by a fixed positive amount with constant probability on sampled instances requires super-polynomially many membership queries in $M$, even with unlimited computation between queries.
\end{enumerate}
\end{inftheorem}

Since $\alpha_{\rm Q}(R)>\alpha_{\rm C}(R)$ for every $R\in(0,1)$, verification and a constant number of repetitions of the DQI algorithm give a bounded-error quantum algorithm for every fixed rational threshold $r\in(\alpha_{\rm C}(R),\alpha_{\rm Q}(R))$, using $\calO_{R,r}(M)$ queries on every promised instance. For example, at rate $R=3/10$, $\alpha_{\rm C}(3/10)=0.65$ and $\alpha_{\rm Q}(3/10) \approx0.8571$.
Thus choosing for example $r=0.8$ gives a bounded-error quantum algorithm attaining score at least $0.8$. Since the optimum is at most one, this also gives a $0.8$-approximation.

The block alphabet has size $\exp(\Theta(M\log M))$, with each symbol represented using $\Theta(M\log M)$ bits. The acceptance sets are accessed through membership oracles, and all running-time bounds are measured in the number of blocks $M$.

\paragraph{Beyond DQI.}

Motivated by the beyond-DQI algorithms of Horinaga and Yamakawa~\cite{HorinagaYamakawa2026} and, independently, Jo~\cite{Jo2026}, we ask a secondary question: is the DQI score $\alpha_{\rm Q}(R)$ the best quantum guarantee efficiently attainable in our setting? We discover that it is not: We constructively exceed this score within the same Regev-style framework by replacing \emph{unique decoding} with \emph{complete list decoding} of the folded dual and adapting Jo's coherent fiber-summation method to combine all retained errors with the same syndrome.

For every fixed rational code rate $R\in(0,1)$, define the score envelope
\begin{align}
\alpha_{\rm fib}(R)
&\coloneq
\begin{cases}
\displaystyle \frac12+\sqrt{R(1-R)},&0<R\leq\frac12,\\[1ex]
1,&\frac12<R<1.
\end{cases}
\label{eq:intro-fiber-envelope}
\end{align}
For every fixed rational threshold $0<\beta<\alpha_{\rm fib}(R)$, our construction gives a quantum algorithm running in polynomial time and using linearly many membership queries. For a $1-o(1)$ fraction of sampled instances, this algorithm returns a polynomial with score at least $\beta$ with probability at least $2/3$. These stronger guarantees therefore hold on \emph{typical} sampled instances, whereas the DQI guarantee holds for \emph{every} promised instance.

For fixed $R$ and $\beta$, the complete decoder returns lists of constant size, independent of $M$. This allows all retained errors with the same syndrome to be combined coherently with polynomial overhead. The algorithm and its complexity bounds may depend on the fixed rate and threshold.

For $R>1/2$, we also obtain an \emph{exact-search} guarantee: a separate construction finds a polynomial accepted in every block, with bounded-error success on typical sampled instances and linearly many membership queries. This is stronger than attaining every fixed score threshold below one. The formal statements are given in \cref{thm:sampled-score-hierarchy}.

\Cref{fig:certified_advantage_regimes} summarizes the resulting certified score regimes.

\begin{figure}[!t]
\centering
    \includegraphics[width=0.7\linewidth]{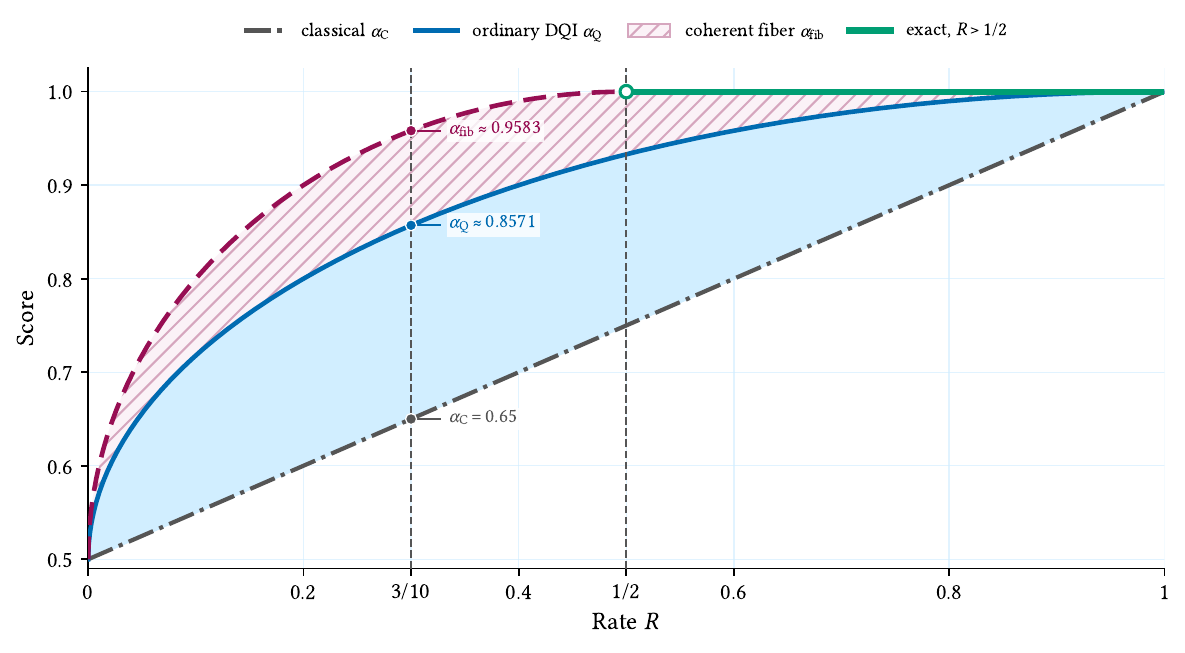}
    \caption{Provable score guarantees for our balanced folded OPI instances. DQI's guarantee holds for every promised oracle, whereas the stronger guarantees from coherent fiber summation hold only on typical sampled oracles. The gray dash-dotted curve is the classical threshold $\alpha_{\rm C}$ attained by Prange's algorithm under the distributional quantifiers of \cref{thm:informal-frs-separation}. The solid blue curve marks $\alpha_{\rm Q}$, and the light-blue region indicates the thresholds above $\alpha_{\rm C}$ and below $\alpha_{\rm Q}$ that DQI achieves on every promised oracle. The pale-pink hatched region marks the additional thresholds below $\alpha_{\rm fib}$ achieved by coherent fiber summation on typical sampled oracles. The solid green segment gives typical-oracle exact search for $R>1/2$, with an open endpoint at $R=1/2$; the vertical dashed line marks this rate. These quantum boundaries delimit provably achievable regimes, not upper bounds on arbitrary quantum algorithms. For $R\leq1/2$, attainment of the fiber endpoint is not asserted. At the marked rate $R=3/10$, the three levels are $\alpha_{\rm C}=0.65$, $\alpha_{\rm Q}\approx0.8571$, and $\alpha_{\rm fib}\approx0.9583$. Only the differences between $\alpha_{\mathrm{Q}}$ and $\alpha_{\mathrm{C}}$ constitutes a proven separation, because $\alpha_{\mathrm{fib}}$ is over typical instances only and not a worst-case statement.}
    \label{fig:certified_advantage_regimes}
\end{figure}

\subsection{Proof ideas}

The separation in \cref{thm:informal-frs-separation} rests on three ingredients: the DQI score guarantee~\cite{JordanEtAl2025DQI}, Prange's classical benchmark~\cite{Prange1962}, and a classical query lower bound. The novelty lies primarily in the resulting separation, rather than in a new coding-theoretic or quantum primitive. We do not claim a new code family, decoding paradigm, or DQI mechanism: these ingredients are inherited from earlier work~\cite{YamakawaZhandry2024,Berlekamp1968,Massey1969,JordanEtAl2025DQI}. We establish a common membership-oracle setting in which they yield a rigorous quantum advantage for approximation, together with a classical query lower bound identifying the sharp classical threshold. This requires verifying the folded DQI score guarantee and extending the Yamakawa--Zhandry exact-search query lower bound to approximation using a standard conditional binomial-tail estimate.

\paragraph{DQI on folded codes.}
DQI turns a coherent decoder for the dual code into an optimization algorithm over the primal code~\cite{JordanEtAl2025DQI}. Building on the folded construction of Jordan et al., we verify that the semicircle score guarantee extends to balanced folded predicates under membership-oracle access. For our folded Reed--Solomon family, the required decoder operates strictly below half the unfolded dual code's minimum distance, so conventional Reed--Solomon unique decoding suffices~\cite{Berlekamp1968,Massey1969}. This yields the claimed DQI score uniformly over every promised instance.

\paragraph{Prange's algorithm on folded codes.}
A block version of Prange's algorithm~\cite{Prange1962} constructs a solution by selecting an accepted symbol in each of a suitable number of blocks and interpolating a low-degree polynomial through their evaluations. These blocks are satisfied by construction. The remaining acceptance sets are independent of the interpolated polynomial, so each remaining block accepts with probability one half. This yields the classical score claimed above.

\paragraph{From exact search to approximation hardness.} 
Our classical query lower bound adapts the classical hardness proof of Yamakawa and Zhandry~\cite{YamakawaZhandry2024}. Their \emph{exact-search argument} exploits the requirement that a successful candidate must be accepted on every still-unqueried block. 
For \emph{approximation}, we instead need to control the probability that sufficiently many of these blocks accept.
We therefore augment their proof with a standard probabilistic estimate. List recoverability bounds the number of candidate codewords for which, after a polynomial number of queries, membership is known on a sufficiently large fraction of blocks. Then, a conditional binomial-tail bound makes it exponentially unlikely that any one of these candidates is accepted on sufficiently many of its unqueried blocks, thus leading to a score above a certain threshold. A union bound over the list-recovery candidates then rules out every constant improvement over Prange's algorithm within polynomially many queries.

\subsection{Related work}

\paragraph{Provable super-polynomial quantum advantages in optimization.}

Conditional quantum advantages for approximate optimization are known in the standard model. Assuming the classical hardness of RSA inversion, Pirnay et al.\ use cryptographic and computational-learning reductions to obtain a super-polynomial quantum advantage for approximating specially constructed integer-programming instances~\cite{Pirnay2024}. Their quantum algorithm extracts 
the public RSA parameters, uses Shor's algorithm~\cite{Shor1994} to reconstruct the secret key, and converts it into a solution with the claimed approximation guarantee. Szegedy argues that the PCP theorem and Shor's algorithm yield analogous conditional approximation separations~\cite{szegedy2022}, even though this path is yet to be formulated in full. These routes transport a Shor-based speedup through classical reductions, whereas the present separation arises directly from coherent decoding and Fourier interference.
In an oracle model, Leng et al.\ establish unconditional super-polynomial quantum--classical separations for discrete and continuous optimization tasks using adiabatic quantum optimization and quantum Hamiltonian descent, respectively~\cite{leng2025subexponentialquantumspeedupoptimization}.

\paragraph{The power and limitations of DQI and related Regev-style algorithms.}
The Fourier--decoding mechanism of DQI~\cite{JordanEtAl2025DQI} belongs to the broader lineage of Regev-style quantum reductions~\cite{ATS03,Regev2004,AR2005,Regev2009}. Garrido and Chailloux study whether it can yield a non-oracular advantage for Reed--Solomon decoding problems related to discrete logarithms, although the available efficient decoders do not reach the required parameter regime~\cite{garrido2026}. Yamakawa and Zhandry establish a quantum--classical random-oracle separation for exact search on folded Reed--Solomon instances~\cite{YamakawaZhandry2024}. Briaud et al.\ give an expected polynomial-time quantum algorithm for exact search in succinct random low-degree multivariate polynomial systems~\cite{BriaudEtAl2025Multivariate}; for degrees greater than two, the corresponding classical hardness remains conjectural.

Jordan et al.~\cite{JordanEtAl2025DQI} discuss the connection between DQI and Yamakawa--Zhandry~\cite{YamakawaZhandry2024} and develop a folded-DQI construction with an oracle implementation in their Supplementary Section~11. However, they do not explicitly establish the corresponding semicircle score guarantee for folded codes. Building on this construction, we verify that guarantee within the strict distance regime and extend the Yamakawa--Zhandry classical hardness argument from exact search to approximation. \Cref{sec:yz-dictionary} further clarifies this connection and compares the two oracle models.

Scores above the DQI semicircle value are known for OPI, both existentially and through efficient quantum algorithms. Sun and Wootters prove an existential improvement for balanced OPI over prime fields~\cite{SunWootters2026}. Horinaga and Yamakawa and, independently, Jo subsequently give constructive worst-case quantum algorithms beyond the DQI regime ~\cite{HorinagaYamakawa2026,Jo2026}. Although their constructions differ, both go beyond unique decoding by using a list decoder for the dual code. Following this line, our stronger quantum results replace the unique decoder by a complete list decoder for the folded dual and adapt Jo's coherent fiber-summation method to combine all returned labels with the same syndrome.

A related line of work replaces unique decoding or list decoding by \emph{soft decoding}. Chailloux and Tillich develop Regev-style quantum reductions based on soft information from a decoder ~\cite{chailloux2024softdecoders}, and Chailloux subsequently connects this framework directly to scalar OPI in an average-case setting ~\cite{chailloux2025opixsoftdecoders}.
Other recent developments include optimization using locally quantum decoders~\cite{shutty2026optimizationusinglocallyquantumdecoders}, extensions using algebraic geometry codes~\cite{GuJordan2025HOPI,matthews2026cofidqicurvebasedoptimalfunction}, and multivariate extensions of DQI~\cite{bu2026multivariatedecodedquantuminterferometry}. Hamiltonian DQI extends the approach to Gibbs-state preparation and Hamiltonian optimization~\cite{schmidhuber2025hamiltoniandecodedquantuminterferometry,bu2026hamiltoniandecodedquantuminterferometry}. Shang and Stilck Fran\c{c}a develop a common relative-decoding framework encompassing DQI and its Hamiltonian extension~\cite{ShangFranca2026}.

Several results delimit where a DQI advantage can occur. PCP-based inapproximability establishes strong $\mathsf{NP}$-hardness thresholds for general and bounded-degree max-LINSAT~\cite{KramerSchubertEisert2026Tight,KramerSchubertEisert2026Bounded}. These thresholds constrain quantum as well as classical polynomial-time algorithms under the corresponding complexity assumptions, but do not preclude an advantage on structured families. Marwaha et al.\ show that binary DQI constructively realizes Tietavainen's covering-radius bound~\cite{Tietavainen1990} and that DQI's output distributions can be approximately sampled in randomized polynomial time with an $\mathsf{NP}$ oracle~\cite{MarwahaEtAl2026}. Parekh rules out an asymptotic DQI advantage for the relevant MaxCut families~\cite{parekh2025DQI_maxcut}, while Anschuetz, Gamarnik, and Lu derive obstructions for random local max-$k$-XORSAT under their stated ensemble and decoder hypotheses~\cite{anschuetz2025DQI}. On related random-regular instances, structure-aware classical heuristics can match the locally quantum decoder studied by Shutty et al.~\cite{shutty2026optimizationusinglocallyquantumdecoders}. Classical Markov-chain Monte Carlo experiments also reproduce DQI approximation ratios at substantial finite sizes, with data consistent with exponential scaling on the tested OPI instances~\cite{gilfuster2026mcmc}. These works provide important evidence about the power and limitations of DQI.

\subsection{Organization}

After introducing notation and the general problem definitions in \cref{sec:preliminaries}, we define the balanced folded Reed--Solomon family and state the concrete separation in \cref{sec:concrete-problem}. We prove the separation in \cref{sec:concrete-proof}, and show in \cref{sec:sampled-score-hierarchy} that stronger quantum algorithmic guarantees exist. We conclude in \cref{sec:discussion}.

The appendices provide the technical foundations supporting the main results. \Cref{app:dqi} gives the complete score analysis and implementation of DQI. \Cref{app:ordinary-dqi-consequences} proves that sampled instances typically admit perfect solutions and uses this to transfer the classical query lower bound for achieving any fixed score above $\alpha_{\rm C}(R)$ to approximation ratios, completing the DQI ratio separation. \Cref{sec:stronger-sampled-algorithms} proves the stronger guarantees on typical sampled oracles, including exact search, and relates these constructions to the Yamakawa--Zhandry setting.

\section{Preliminaries}
\label{sec:preliminaries}

This section fixes notation, reviews linear and folded codes, and gives a folded code formulation of the optimization problem and access conventions used below.

\subsection{Notation}
Throughout, let $q$ be a prime power and let $\F_q$ be the field of order $q$ with multiplicative group $\F_q^\times\coloneq\F_q\setminus\{0\}$. For a positive integer $L$, write $[L]\coloneq\{0,\ldots,L-1\}$. We write $\mathbf 0$ and $\mathbf 1$ for the all-zero and all-one vectors, respectively, with their dimensions determined by context. For any statement $E$, let $\mathds 1_{E}$ denote its indicator function.

\subsection{Codes and folding}
\label{sec:codes-and-folding}
We recall standard coding-theoretic definitions and facts; see, e.g., Ref.~\cite{Lindell2010CodingTheory} for background on linear and Reed--Solomon codes. An $[N,K]_q$ linear code is a $K$-dimensional $\F_q$-linear subspace $C\subseteq\F_q^N$. The rate of an $[N,K]_q$ code is $K/N$. Equivalently, if $B\in\F_q^{N\times K}$ has full column rank, then
\begin{align}
 C &= \operatorname{im}B = \{Bx:x\in\F_q^K\}
\end{align}
is an $[N,K]_q$ code. With respect to the standard inner product on $\F_q^N$, its dual code is
\begin{align}
 C^\perp &\coloneq{} \{y\in\F_q^N:y^Tc=0\text{ for every }c\in C\}
 = \ker B^T.
 \label{eq:prelim-dual-code}
\end{align}

\begin{definition}[Folding]
\label{def:folding}
Choose $N$ such that $N=hM$ with $h,M \in \mathbb{N}$. For a word $c=(c_t)_{t\in[N]}\in\F_q^N$, define its $h$-folding by
\begin{align}
 \operatorname{Fold}_h(c)
 &\coloneq
 \left(
  (c_{ih+j})_{j\in[h]}
 \right)_{i\in[M]}
 \in
 (\F_q^h)^M
 \label{eq:folding-map} \\
 &= \left[ (c_0, \dots, c_{h-1}), \dots, (c_{N-h}, \dots, c_{N-1})\right].
\nonumber
\end{align}
For a linear code $C\subseteq\F_q^N$, its $h$-folded version is
\begin{align}
 C_{(h)}
 &\coloneq
 \operatorname{Fold}_h(C)
 =
 \{\operatorname{Fold}_h(c):c\in C\}
 \subseteq
 (\F_q^h)^M.
 \label{eq:folded-code}
\end{align}

When $h=1$, we call the code \emph{scalar}, and use this term for related concepts such as the distance or decoder defined at block size one. Explicitly, $C_{(1)}=C$. Since $\operatorname{Fold}_h$ only regroups coordinates into blocks, $C_{(h)}$ contains the same data as $C$, and in particular has the same dimension and rate.
\end{definition}

We next define suitable notions of block support and block weight for any block word $y=(y_i)_{i\in[M]}\in(\F_q^h)^M$:
\begin{align}
 \operatorname{supp}_{\rm blk}(y)
 &\coloneq{}
 \{i\in[M]:y_i\neq \mathbf{0}\},
 &
 \operatorname{wt}_{\rm blk}(y)
 &\coloneq{}
 \abs{\operatorname{supp}_{\rm blk}(y)}.
 \label{eq:prelim-block-support-weight}
\end{align}
We further introduce the following notion of block distance and block inner pai for any pair of block words $y,y'\in(\F_q^h)^M$:
\begin{align}
 d_{\rm blk}(y,y')
 &\coloneq\operatorname{wt}_{\rm blk}(y-y')
 =\abs{\{i\in[M]:y_i\ne y'_i\}}.
 \label{eq:prelim-block-word-distance}\\
 \langle y,y'\rangle_{\rm blk}
 &\coloneq{}
 \sum_{i\in[M]}y_i^Ty_i'.
 \label{eq:folded-block-pairing}
\end{align}
We note that the block inner product corresponds to the regular inner product after unfolding.

For a linear block code $C_{(h)}\subseteq(\F_q^h)^M$, its minimum block distance is
\begin{align}
 d_{\rm blk}(C_{(h)})
 &\coloneq{}
 \min_{\mathbf{0}\neq y\in C_{(h)}}\operatorname{wt}_{\rm blk}(y),
 \label{eq:prelim-block-distance}
\end{align}
with the convention $d_{\rm blk}(\{0\})=+\infty$. Again we highlight that we use the term ``scalar'' to refer to the $h=1$ case, and this definition recovers the usual notion of distance~\cite{Lindell2010CodingTheory}.

Throughout, we use \emph{unique decoding} to mean bounded-distance decoding within a prescribed radius strictly below half the relevant minimum distance. Error weights and distances are measured in scalar coordinates or folded blocks, as specified.

In the problems below, one whole $h$-tuple---rather than each scalar coordinate separately---is the unit on which acceptance is counted. Folding only regroups the coordinates, the field $\mathbb{F}_q$ we work over and the code rate $K/(hM)=K/N$ are preserved.

We write $\F_q[Y]$ for the set of polynomials in the variable $Y$ with coefficients in $\F_q$. Thus every $Q\in\F_q[Y]$ with $\deg Q<K$ can be written uniquely as
\begin{align} 
Q(Y) &= \sum_{r\in[K]}x_rY^r, & x_r&\in\F_q.
\end{align}

\begin{definition}[Generalized Reed--Solomon code]
\label{def:grs-evaluation-code}
Let $\boldsymbol\alpha=(\alpha_t)_{t\in[N]}$ consist of pairwise distinct elements of $\F_q$, and let $\boldsymbol v=(v_t)_{t\in[N]}\in(\F_q^\times)^N$ be nonzero coordinate multipliers. For $0\leq K\leq N\leq q$, the corresponding generalized Reed--Solomon code is
\begin{align}
\operatorname{GRS}_{\boldsymbol\alpha,\boldsymbol v,K}
&\coloneq
\left\{
 \bigl(v_tQ(\alpha_t)\bigr)_{t\in[N]}:
 Q\in\F_q[Y],\ \deg Q<K
\right\}
\subseteq\F_q^N,
\label{eq:grs-evaluation-code}
\end{align}
with the convention $\operatorname{GRS}_{\boldsymbol\alpha,\boldsymbol v,0}\coloneq\{0\}$. The ordinary \emph{Reed--Solomon code} $\operatorname{RS}_{\boldsymbol\alpha,K}$ is the special case $\boldsymbol v=\boldsymbol 1$.
\end{definition}

Because the evaluation points are distinct and $N\geq K$, evaluation is injective on polynomials of degree less than $K$. Hence $\operatorname{GRS}_{\boldsymbol\alpha,\boldsymbol v,K}$ is an $[N,K]_q$ linear code, and in particular so is $\operatorname{RS}_{\boldsymbol\alpha,K}$. The generalized form is needed because the dual of an ordinary Reed--Solomon code generally has nontrivial coordinate multipliers, as we will see below.

With respect to the standard inner product on $\F_q^N$, the dual of a Reed--Solomon code is a generalized Reed--Solomon code on the same evaluation points
\begin{align}
\bigl(\operatorname{RS}_{\boldsymbol\alpha,K}\bigr)^\perp
&=
\operatorname{GRS}_{\boldsymbol\alpha,\boldsymbol\nu,N-K}.
\label{eq:rs-dual-grs}
\end{align}
For background on this duality, see Ref.~\cite[Section~6]{Lindell2010CodingTheory} and Ref.~\cite[Section~4.3.1]{YamakawaZhandry2024}. For arbitrary pairwise distinct evaluation points, the dual multipliers may be chosen as
\begin{align}
\nu_t
&\coloneq
\left(
\prod_{\substack{s\in[N]\\s\neq t}}
(\alpha_t-\alpha_s)
\right)^{-1},
&
\boldsymbol\nu
&\coloneq(\nu_t)_{t\in[N]}.
\label{eq:rs-dual-multipliers}
\end{align}
The distinctness of the evaluation points makes every multiplier nonzero, including when one evaluation point is zero.

This dual has dimension $N-K$ and, when $K<N$, scalar minimum distance $K+1$. Since all $\nu_t$ are nonzero, coordinatewise division by $\nu_t$ maps it to $\operatorname{RS}_{\boldsymbol\alpha,N-K}$ and preserves error support. Thus decoding the generalized Reed--Solomon dual reduces to ordinary Reed--Solomon decoding after a known coordinatewise rescaling; the decoded word is then mapped back by restoring the multipliers.

For the full nonzero evaluation set used below, $N=q-1$ and $\alpha_t=\gamma^t$ for a primitive $\gamma\in\F_q^\times$ and $t\in[N]$. These points enumerate $\F_q^\times$. In characteristic two, the multipliers therefore simplify to
\begin{align}
\nu_t&=\alpha_t.
\label{eq:dqi-full-length-dual-multipliers}
\end{align}
Indeed, differentiating $\prod_{s\in[N]}(Y-\alpha_s)=Y^{q-1}-1$ at $Y=\alpha_t$ gives $\prod_{s\neq t}(\alpha_t-\alpha_s)=\alpha_t^{-1}$.

\begin{definition}[Folded Reed--Solomon code]
\label{def:folded-rs-code}
Let $h,M$ be positive integers, put $N=hM$, and assume $0<K\leq N\leq q$. Let $\boldsymbol\alpha=(\alpha_t)_{t\in[N]}$ consist of pairwise distinct elements of $\F_q$. The $h$-folded Reed--Solomon code is
\begin{align}
\operatorname{FRS}_{\boldsymbol\alpha,K,h}
\coloneq
\operatorname{Fold}_h
\bigl(\operatorname{RS}_{\boldsymbol\alpha,K}\bigr)
=
\left\{
 \left(
  \bigl(Q(\alpha_{ih+j})\bigr)_{j\in[h]}
 \right)_{i\in[M]}:
 Q\in\F_q[Y],\ \deg Q<K
\right\}
\subseteq(\F_q^h)^M.
\label{eq:folded-rs-code}
\end{align}
It is an $\F_q$-linear code of dimension $K$.
\end{definition}

In this work, we only consider folded Reed--Solomon codes with $N=hM\leq q-1$ whose evaluation points are consecutive powers of a primitive element $\gamma\in\F_q^\times$, namely $\alpha_{ih+j}\coloneqq\gamma^{ih+j}$ for $i\in[M]$ and $j\in[h]$.

Under our definition for block inner product, folding commutes with taking duals~\cite[Section~4.1]{YamakawaZhandry2024}. Combining this observation with \eqref{eq:rs-dual-grs} gives
\begin{align}
\bigl(\operatorname{FRS}_{\boldsymbol\alpha,K,h}\bigr)^\perp
&=
\operatorname{Fold}_h
\bigl(
 \operatorname{GRS}_{\boldsymbol\alpha,\boldsymbol\nu,N-K}
\bigr).
\label{eq:folded-rs-dual-grs}
\end{align}
The coordinatewise rescaling above also preserves block error support.

\subsection{Optimization problems and oracle access}
\label{sec:optimization-problems}

We formulate the objective once over a block-linear code. Folded max-LINSAT permits an arbitrary linear encoder, i.e., a matrix $B$, while folded \emph{optimal polynomial intersection} (OPI) specializes the encoder to low-degree polynomial evaluation.

\paragraph{Folded max-LINSAT.} 

We start by defining folded max-LINSAT over $\F_q$.

\begin{definition}[Folded max-LINSAT over $\F_q$]
\label{def:folded-max-linsat}
Fix integers $h,M,K>0$, put $N\coloneq hM$, assume $0<K<N$, and let $\Sigma\coloneq\F_q^h$. An instance consists of a full-column-rank matrix
\begin{align}
 B
 &\coloneq
 \begin{bmatrix}
 B_0\\
 \vdots\\
 B_{M-1}
 \end{bmatrix}
 \in\F_q^{N\times K},
 &
 B_i
 &\in\F_q^{h\times K},
 \label{eq:folded-max-linsat-matrix}
\end{align}
and acceptance sets $A_i\subseteq\Sigma$ for $i\in[M]$. Its feasible block words form the code
\begin{align}
 C
 &\coloneq
 \left\{
 (B_ix)_{i\in[M]}:
 x\in\F_q^K
 \right\}
 =
 \operatorname{Fold}_h(\operatorname{im}B)
 \subseteq\Sigma^M.
 \label{eq:max-linsat-feasible-code}
\end{align}
Because $B$ has full column rank, every feasible word $c \in C$ has a unique message $x\in\F_q^K$. For $c=(c_i)_{i\in[M]}\in C$, define
\begin{align}
 \scor(c)
 &\coloneq
 \frac1M\sum_{i\in[M]}\mathds{1}_{\{c_i\in A_i\}}.
 \label{eq:folded-max-linsat-score}
\end{align}
The optimization task is to find $c\in C$, equivalently $x\in\F_q^K$, maximizing this score. For a threshold $\beta\in[0,1]$, the threshold-search task asks for a feasible $c$ satisfying $M\scor(c)\geq\ceil{\beta M}$. If no such word exists, the relation has no valid output.

For a fixed instance, define $\operatorname{OPT}\coloneq\max_{c\in C}\scor(c)$. The multiplicative approximation ratio of a feasible word is
\begin{align}
 \operatorname{ratio}(c)
 &\coloneq
 \begin{cases}
 \displaystyle
 \frac{\scor(c)}
   {\operatorname{OPT}},
 &\operatorname{OPT}>0,\\[2ex]
 1,
 &\operatorname{OPT}=0.
 \end{cases}
 \label{eq:general-ratio-definition}
\end{align}
Thus $c$ is an $r$-approximation exactly when $\scor(c)\geq r\operatorname{OPT}$.
\end{definition}

We assume without loss of generality that the encoder has full column rank, since deleting redundant columns preserves its image. When $h=1$, this is scalar max-LINSAT over $\F_q$.

\paragraph{Folded OPI.}
\emph{Optimal polynomial intersection} (OPI) asks for a low-degree polynomial whose evaluations belong to specified acceptance sets at as many evaluation points as possible. We use \emph{scalar OPI} to refer to the formulation studied in Refs.~\cite{JordanEtAl2025DQI,SunWootters2026,HorinagaYamakawa2026,Jo2026}, where each acceptance test concerns a single polynomial evaluation in $\F_q$. In \emph{folded OPI}, each test concerns a tuple of $h$ evaluations of the same polynomial, tested jointly against a specified acceptance set in $\F_q^h$.

Under the hypotheses of \cref{def:folded-max-linsat}, additionally assume $N=hM\leq q$ and choose pairwise distinct evaluation points $\boldsymbol\alpha=(\alpha_t)_{t\in[N]}$. Represent a polynomial of degree less than $K$ by its coefficient vector:
\begin{align}
Q(Y) &\coloneq \sum_{r\in[K]}x_rY^r, &
x &\coloneq (x_r)_{r\in[K]} \in\F_q^K.
\end{align}
Polynomial evaluation is linear in the coefficients: multiplying the row $(1,\alpha,\ldots,\alpha^{K-1})$ by the coefficient vector $x$ gives $Q(\alpha)$. Write $\alpha_{i,j}\coloneq\alpha_{ih+j}$ for $i\in[M]$ and $j\in[h]$. For block $i$, we collect the $h$ evaluation rows into $B_i$; stacking these matrices gives the full encoder:
\begin{align}
B_i &\coloneq
\begin{bmatrix}
1&\alpha_{i,0}&\cdots&\alpha_{i,0}^{K-1}\\
\vdots&\vdots&&\vdots\\
1&\alpha_{i,h-1}&\cdots&\alpha_{i,h-1}^{K-1}
\end{bmatrix}
\in\F_q^{h\times K}, \qquad i\in[M], \\
B &\coloneq
\begin{bmatrix}
B_0\\
\vdots\\
B_{M-1}
\end{bmatrix}
\in\F_q^{hM\times K}.
\end{align}
Thus $B_i x=c_i(Q)$ is the evaluation block of $Q$ at the $i$th group of points, and $Bx$ collects all $M$ blocks. Since the evaluation points are distinct and $N\geq K$, this encoder has full column rank. Thus the coefficient vector, the polynomial, and its folded evaluation word $c(Q)\coloneq(c_i(Q))_{i\in[M]}$ are equivalent descriptions of a feasible solution. The feasible code is precisely $\operatorname{FRS}_{\boldsymbol\alpha,K,h}$ from \cref{def:folded-rs-code}.

Given acceptance sets $A_i\subseteq\F_q^h$, the objective is to find a polynomial $Q\in\F_q[Y]$ of degree less than $K$ that maximizes the accepted-block fraction $\scor(Q)\coloneq\frac1M\sum_{i\in[M]} \mathds{1}_{\{c_i(Q)\in A_i\}}$. We refer to this specialization of folded max-LINSAT as \emph{folded OPI}. Each accepted block contributes one unit to the unnormalized objective. The acceptance sets may impose arbitrary joint conditions on the $h$ evaluations within a block; they need not be Cartesian products of scalar acceptance sets. Folding therefore preserves the polynomial space while allowing a different objective and unit of oracle access: one membership query tests one complete block value.
For $h=1$, one has $M=N$ and recovers scalar OPI. The folded OPI formulation above fits the extension of DQI to block predicates developed in Ref.~\cite{JordanEtAl2025DQI}.

The score and approximation-ratio conventions are inherited from \cref{def:folded-max-linsat}, with
\begin{align}
\scor(Q) &\coloneq \scor(c(Q)),
& \operatorname{ratio}(Q) &\coloneq \operatorname{ratio}(c(Q)).
\end{align}

\subsection{Acceptance sets and membership-oracle access}
\label{sec:balanced-oracle}

The encoder $B$ and the block alphabet $\Sigma$ are specified explicitly and independently of the oracle. The oracle-dependent part of an instance is the collection of acceptance sets $\mathcal A=(A_i)_{i\in[M]}$. For each block $i\in[M]$, define its Boolean membership function by
\begin{align}
O_{A_i}(z)
&\coloneq \mathds{1}_{\{z\in A_i\}},
\qquad z\in\Sigma.
\label{eq:membership-table}
\end{align}
The instance provides the family of membership oracles
\begin{align}
\calO_{A_i}\ket{z,t}
&\coloneq
\ket{z,t\oplus O_{A_i}(z)},
\qquad i\in[M],\quad t\in\F_2.
\label{eq:general-membership-oracle}
\end{align} 
Throughout, one query means one invocation of one fixed oracle $\mathcal O_{A_i}$, with $i$ chosen classically. The symbol and answer registers may be in arbitrary superposition, but the interface provides neither a coherent block-index register nor a multiplexed oracle at unit cost. Thus one query tests whether one complete block symbol $z$ is accepted in one chosen block.

When $|\Sigma|$ is even, we call the oracle family \emph{balanced} if $|A_i| =|\Sigma|/2$ for every $i \in [M]$.
The preceding optimization definitions also apply when the acceptance sets are given explicitly; the separation below uses only the membership interface above.

\section{Statement of the concrete separation}
\label{sec:concrete-problem}

We first define the explicit balanced folded-OPI instance family and then state the concrete separation theorem for DQI.

\subsection{The instance family and oracle model}
\label{sec:density-oracle-model}

For the folded Reed-Solomon code, fix a rational target rate $R\in(0,1)$. For a positive integer $M$ such that $M+1$ is a power of two, set
\begin{align}
     h&\coloneq M+2,
     &
     N&\coloneq hM,
     &
     q&\coloneq N+1=(M+1)^2,
     &
     K&\coloneq \floor{RN}.
     \label{eq:frs-parameters}
\end{align}
We also write
\begin{align}
    J &\coloneq \floor{K/h} = \floor{RM}.
\label{eq:frs-complete-block-count}
\end{align}
We can choose an arbitrary target symbol in each of $J$ blocks and interpolate a polynomial of degree less than $K$ matching all these targets simultaneously. Each target specifies $h$ evaluation values, so $J$ is the largest number of complete blocks for which this is always possible: $hJ\leq K<h(J+1)$. Prange's algorithm in \cref{sec:proof-prange} uses this freedom to make the selected blocks accept.
We call $M$ \emph{admissible for $R$} when, in addition, $K\geq h$. For every fixed $R>0$, this excludes only finitely many values of $M$. These choices make $q$ a power of two and give $hM=q-1$, so the nonzero field elements can be partitioned into exactly $M$ blocks of $h$ evaluation points each. The actual finite-length code rate is $\frac KN=R+\calO(1/N)$, so $R$ denotes the limiting target rate rather than the exact finite-length rate. Use the canonical representation of $\F_q$ and primitive element $\gamma\in\F_q^\times$. Define $\alpha_{i,j} \coloneq  \gamma^{ih+j}$ for $i\in[M]$ and $j \in[h]$.
As $i$ and $j$ vary, these points enumerate $1,\gamma,\ldots,\gamma^{N-1}$, and hence all nonzero elements of $\F_q$ exactly once.

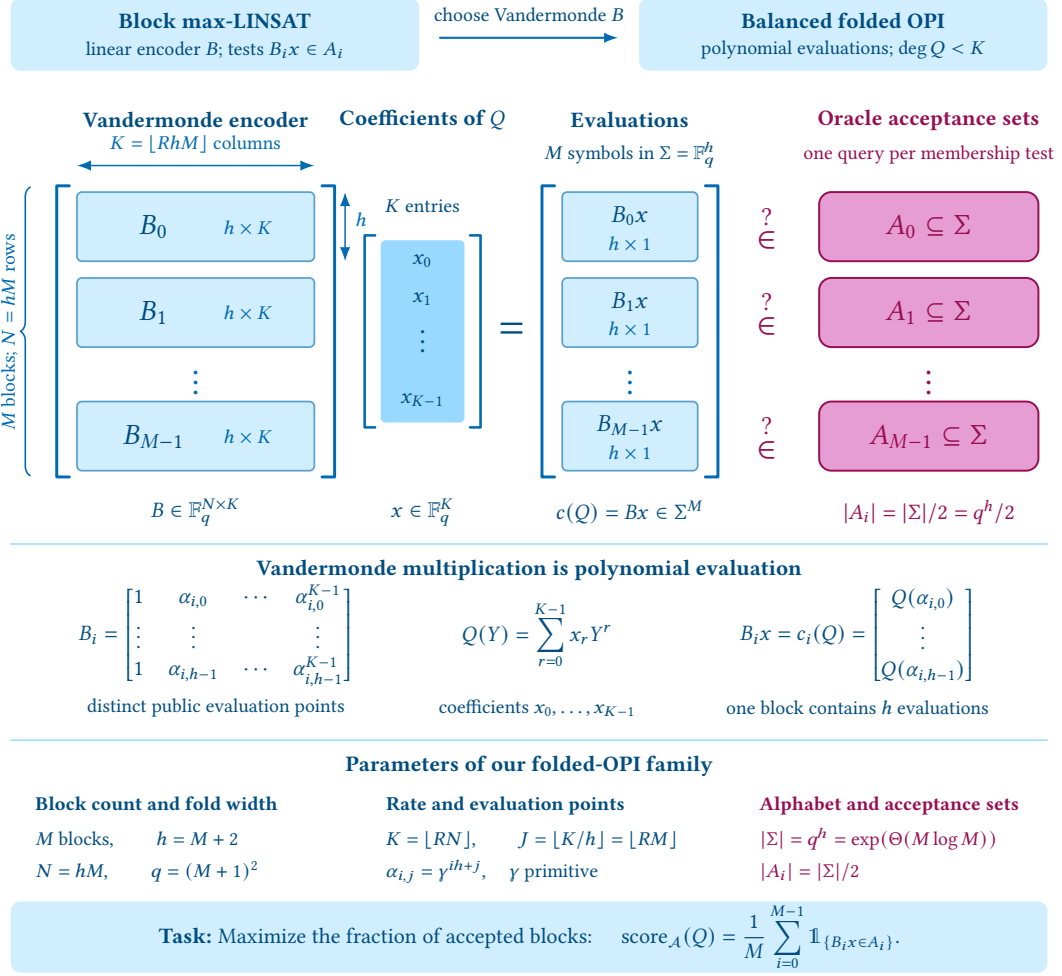
\begin{figure}[!t]
\centering
\resizebox{0.9\linewidth}{!}{
\begingroup
\definecolor{fmInk}{HTML}{004D80}
\definecolor{fmBlue}{HTML}{006BB1}
\definecolor{fmBlueLight}{HTML}{D1EEFF}
\definecolor{fmBlueMid}{HTML}{99D9FF}
\definecolor{fmOracle}{HTML}{970E53}
\definecolor{fmOracleLight}{HTML}{E69DD3}
\definecolor{fmMuted}{HTML}{004D80}
\definecolor{fmRule}{HTML}{99D9FF}
\begin{tikzpicture}[
  x=1cm,y=1cm,
  font=\rmfamily\fontsize{9}{11}\selectfont,
  text=fmInk,
  line cap=round,line join=round,
  every node/.style={inner sep=2pt},
  heading/.style={font=\rmfamily\bfseries\fontsize{9}{11}\selectfont},
  note/.style={font=\rmfamily\fontsize{8}{10}\selectfont,text=fmMuted,align=center},
  dim/.style={font=\rmfamily\fontsize{8}{10}\selectfont,text=fmBlue,align=center},
  block/.style={draw=fmBlue!65,line width=.7pt,fill=fmBlueLight,
    rounded corners=2pt},
  set/.style={draw=fmOracle,line width=.8pt,fill=fmOracleLight,
    rounded corners=5pt},
  brace/.style={decorate,decoration={brace,amplitude=4pt},
    draw=fmBlue,line width=.65pt},
  dimension/.style={Latex-Latex,draw=fmBlue,line width=.6pt}
]
\path[fill=fmBlueLight,rounded corners=4pt]
  (.20,6.68) rectangle (6.05,7.72);
\node[heading] at (3.125,7.40) {Block max-LINSAT};
\node[note] at (3.125,6.98) {linear encoder $B$; tests $B_i x\in A_i$};
\draw[-{Latex[length=2mm]},fmBlue,line width=.85pt]
  (6.35,7.15)--(8.90,7.15);
\node[note] at (7.625,7.49) {choose Vandermonde $B$};
\path[fill=fmBlueLight,rounded corners=4pt]
  (9.20,6.68) rectangle (15.05,7.72);
\node[heading] at (12.125,7.40) {Balanced folded OPI};
\node[note] at (12.125,6.98) {polynomial evaluations; $\deg Q<K$};

\node[heading] at (2.85,5.98) {Vandermonde encoder};
\node[heading] at (6.10,5.98) {Coefficients of $Q$};
\node[heading] at (9.07,5.98) {Evaluations};
\node[heading,text=fmOracle] at (13.35,5.98) {Oracle acceptance sets};

\draw[dimension] (1.15,5.32) -- (4.55,5.32)
  node[midway,above=3pt,dim] {$K=\lfloor RhM\rfloor$ columns};

\draw[block] (1.15,3.95) rectangle (4.55,4.94);
\draw[block] (1.15,2.72) rectangle (4.55,3.71);
\draw[block] (1.15,.95) rectangle (4.55,1.94);
\foreach \yy/\idx in {4.445/0,3.215/1,1.445/{M-1}} {
  \node[font=\fontsize{12}{14}\selectfont] at (2.25,\yy) {$B_{\idx}$};
  \node[dim] at (3.60,\yy) {$h\times K$};
}
\node[font=\fontsize{14.4}{17}\selectfont,text=fmBlue] at (2.85,2.32) {$\vdots$};
\draw[fmBlue,line width=1.2pt]
  (1.04,5.03)--(.86,5.03)--(.86,.86)--(1.04,.86);
\draw[fmBlue,line width=1.2pt]
  (4.66,5.03)--(4.84,5.03)--(4.84,.86)--(4.66,.86);
\draw[brace] (.48,.88)--(.48,5.01)
  node[midway,xshift=-.30cm,dim,rotate=90] {$M$ blocks; $N=hM$ rows};
\draw[dimension] (4.99,3.95)--(4.99,4.94)
  node[midway,right=2pt,yshift=3pt,dim] {$h$};

\draw[fill=fmBlueMid,draw=none,rounded corners=2pt]
  (5.50,1.64) rectangle (6.70,4.25);
\draw[fmBlue,line width=1.05pt]
  (5.42,4.32)--(5.27,4.32)--(5.27,1.57)--(5.42,1.57);
\draw[fmBlue,line width=1.05pt]
  (6.78,4.32)--(6.93,4.32)--(6.93,1.57)--(6.78,1.57);
\node at (6.10,3.96) {$x_0$};
\node at (6.10,3.41) {$x_1$};
\node[font=\fontsize{12}{14}\selectfont] at (6.10,2.91) {$\vdots$};
\node at (6.10,1.97) {$x_{K-1}$};
\node[note] at (6.10,4.74) {$K$ entries};
\node[font=\fontsize{17.28}{21}\selectfont] at (7.36,2.95) {$=$};

\foreach \lo/\hi/\yy/\idx in {
  3.95/4.94/4.445/0,2.72/3.71/3.215/1,.95/1.94/1.445/{M-1}} {
  \draw[block] (8.10,\lo) rectangle (10.04,\hi);
  \node[font=\fontsize{10}{12}\selectfont] at (9.07,{\yy+.16}) {$B_{\idx}x$};
  \node[dim] at (9.07,{\yy-.23}) {$h\times 1$};
}
\node[font=\fontsize{14.4}{17}\selectfont,text=fmBlue] at (9.07,2.32) {$\vdots$};
\draw[fmBlue,line width=1.2pt]
  (7.99,5.03)--(7.81,5.03)--(7.81,.86)--(7.99,.86);
\draw[fmBlue,line width=1.2pt]
  (10.15,5.03)--(10.33,5.03)--(10.33,.86)--(10.15,.86);
\node[note] at (9.07,5.48) {$M$ symbols in $\Sigma=\mathbb F_q^h$};

\foreach \yy/\idx in {4.445/0,3.215/1,1.445/{M-1}} {
  \node[font=\fontsize{14.4}{17}\selectfont,text=fmOracle] at (11.03,\yy)
    {$\overset{?}{\in}$};
  \draw[set] (11.78,{\yy-.495}) rectangle (14.92,{\yy+.495});
  \node[font=\fontsize{12}{14}\selectfont,text=fmOracle] at (13.35,\yy)
    {$A_{\idx}\subseteq\Sigma$};
}
\node[font=\fontsize{14.4}{17}\selectfont,text=fmOracle] at (13.35,2.32) {$\vdots$};
\node[note,text=fmOracle] at (13.35,5.48) {one query per membership test};

\node[align=center] at (2.85,.36)
  {$B\in\mathbb F_q^{N\times K}$};
\node[align=center] at (6.10,.36)
  {$x\in\mathbb F_q^K$};
\node[align=center] at (9.07,.36)
  {$c(Q)=Bx\in\Sigma^M$};
\node[align=center,text=fmOracle] at (13.35,.36)
  {$|A_i|=|\Sigma|/2=q^h/2$};

\draw[fmRule,line width=.6pt] (.20,-.10)--(15.05,-.10);
\node[heading] at (7.625,-.47)
  {Vandermonde multiplication is polynomial evaluation};
\node at (3.14,-1.40) {
  $B_i=\begin{bmatrix}
  1&\alpha_{i,0}&\cdots&\alpha_{i,0}^{K-1}\\
  \vdots&\vdots&&\vdots\\
  1&\alpha_{i,h-1}&\cdots&\alpha_{i,h-1}^{K-1}
  \end{bmatrix}$};
\node at (7.74,-1.40)
  {$\displaystyle Q(Y)=\sum_{r=0}^{K-1}x_rY^r$};
\node at (12.33,-1.40) {
  $B_i x=c_i(Q)=\begin{bmatrix}
  Q(\alpha_{i,0})\\\vdots\\Q(\alpha_{i,h-1})
  \end{bmatrix}$};
\node[note] at (3.14,-2.44)
  {distinct public evaluation points};
\node[note] at (7.74,-2.44)
  {coefficients $x_0,\ldots,x_{K-1}$};
\node[note] at (12.33,-2.44)
  {one block contains $h$ evaluations};

\draw[fmRule,line width=.6pt] (.20,-2.91)--(15.05,-2.91);
\node[heading] at (7.625,-3.27) {Parameters of our folded-OPI family};
\node[anchor=north west,align=left,font=\rmfamily\fontsize{8}{10}\selectfont]
  at (.48,-3.64)
  {\textbf{Block count and fold width}\\[4pt]
   $M\ \text{blocks},\qquad h=M+2$\\[3pt]
   $N=hM,\qquad q=(M+1)^2$};
\node[anchor=north west,align=left,font=\rmfamily\fontsize{8}{10}\selectfont]
  at (5.50,-3.64)
  {\textbf{Rate and evaluation points}\\[4pt]
   $K=\lfloor RN\rfloor,\qquad J=\lfloor K/h\rfloor=\lfloor RM\rfloor$\\[3pt]
$\alpha_{i,j}=\gamma^{ih+j},\quad\gamma\ \text{primitive}$};
\node[anchor=north west,align=left,
      font=\rmfamily\fontsize{8}{10}\selectfont,text=fmOracle]
  at (10.85,-3.64)
  {\textbf{Alphabet and acceptance sets}\\[4pt]
   $|\Sigma|=q^h=\exp(\Theta(M\log M))$\\[3pt]
   $|A_i|=|\Sigma|/2$};

\path[fill=fmBlueLight,rounded corners=4pt]
  (.20,-6.24) rectangle (15.05,-5.22);
\node at (7.63,-5.73)
  {\textbf{Task:} Maximize the fraction of accepted blocks:
   $\displaystyle\quad\operatorname{score}_{\mathcal A}(Q)
       =\frac1M\sum_{i=0}^{M-1}\mathds 1_{\{B_ix\in A_i\}}$.};
\end{tikzpicture}
\endgroup}
\caption{From block max-LINSAT to folded OPI. For a general linear encoder $B$, the objective counts the accepted blocks $B_ix\in A_i$. Choosing Vandermonde rows at distinct evaluation points identifies $x$ with the coefficients of $Q(Y)=\sum_{r=0}^{K-1}x_rY^r$, so $B_ix=(Q(\alpha_{i,j}))_{j=0}^{h-1}$. This gives folded polynomial intersection: one membership query tests an entire $h$-vector of polynomial evaluations. The lower panel records our concrete family parameters for admissible $M$. Each acceptance set contains exactly half of the block symbols; the score is the fraction of blocks accepted by the chosen polynomial. The output is shown with its folded grouping, and dimensions are schematic.}
\label{fig:folded-matrix-shapes}
\end{figure}

Put $\Sigma\coloneq \F_q^h$. For a polynomial $Q\in\F_q[Y]$ with $\deg Q<K$, its $i$th folded evaluation is
\begin{align}
 c_i(Q)
 &\coloneq 
 \bigl(Q(\gamma^{ih+j})\bigr)_{j\in[h]}
 \in
 \Sigma.
\end{align}
The resulting folded Reed--Solomon code is
\begin{align}
 C_{\rm FRS}
 \coloneq 
 \operatorname{FRS}_{\boldsymbol\alpha,K,h} = 
 \left\{
 (c_i(Q))_{i\in[M]}:
 Q\in\F_q[Y],\ \deg Q<K
 \right\}
 \subseteq
 \Sigma^M,
 \label{eq:concrete-frs-code}
\end{align}
where $\boldsymbol\alpha=(\gamma^t)_{t\in[N]}$.

The block alphabet has even cardinality $|\Sigma|=q^h=\exp(\Theta(M\log M))$, with each symbol represented using $h\log_2 q=\Theta(M\log M)$ bits. The acceptance sets are accessed through membership oracles, and all running-time bounds are measured in $M$.
Every block accepts exactly $|\Sigma|/2$ symbols. Thus balance is exact at every admissible size, with no rounding of the acceptance density.

\Cref{fig:folded-matrix-shapes} illustrates the specialization of block max-LINSAT to folded OPI and the parameters of our balanced membership-oracle model. Formally, we define balanced folded OPI as follows:

\begin{definition}[Balanced folded OPI]
\label{def:folded-opi}
For the folded Reed--Solomon data above, an instance is a tuple $\mathcal A=(A_i)_{i\in[M]}$ with $A_i\subseteq\Sigma$ and $\abs{A_i}=|\Sigma|/2$ for every $i$, accessed through the oracle defined in \eqref{eq:general-membership-oracle}.

A feasible output is a polynomial $Q\in\F_q[Y]$ of degree less than $K$, represented by its canonical length-$K$ coefficient vector.
Its score is the fraction of blocks whose evaluation symbols belong to their acceptance sets:
\begin{align}
\scor(Q) &\coloneq \scor(c(Q)) =\frac1M\sum_{i\in[M]} \mathds{1}_{\{c_i(Q)\in A_i\}}.
\end{align}
The threshold and multiplicative-approximation tasks use \cref{def:folded-max-linsat}. In particular, a score threshold $\beta$ requires at least $\ceil{\beta M}$ accepted blocks.
\end{definition}

\subsection{Concrete separation theorem}
For the balanced folded-OPI family and parameter choices above, the following theorem proves that DQI exceeds the asymptotically optimal score attainable with polynomially many classical queries under the sampled-oracle distribution. This establishes a provable quantum advantage in this membership-oracle setting.

\begin{theorem}[Quantum--classical query separation for folded OPI]
\label{thm:frs-separation}
Fix a rational rate $R\in(0,1)$ and the balanced folded-OPI family above. Define
\begin{align}
\alpha_{\rm C}(R)
&\coloneq\frac{1+R}{2},
&
\alpha_{\rm Q}(R)
&\coloneq\frac{1+\sqrt{R(2-R)}}2.
\label{eq:concrete-score-values}
\end{align}

We identify every valid coefficient-vector output with its polynomial. We write $A(\omega)$ for the polynomial output of a classical algorithm $A$ with internal random tape $\omega$, suppressing dependence on $M$ and $\mathcal A$.
In parts \textup{(ii)} and \textup{(iii)}, the sampled-oracle distribution chooses the sets $A_i$ independently and uniformly among the half-size subsets of $\Sigma$.
The following statements hold for all sufficiently large admissible $M$, with limits taken through admissible values.
\begin{enumerate}
    \item There is an explicit DQI algorithm running in quantum polynomial time and using one quantum query to each fixed-block membership oracle $\mathcal O_{A_i}$, hence at most $M$ queries in total. For every promised $\mathcal A$, its output satisfies
    \begin{align}
    \E\bigl[\scor(Q_{\rm DQI})\bigr]
    &=\alpha_{\rm Q}(R)+o(1),
    \label{eq:frs-dqi-value}
    \end{align}
    uniformly over $\mathcal A$. It uses a conventional generalized Reed--Solomon unique decoder strictly below half the minimum distance of the unfolded scalar dual code.
    
    \item There is a randomized polynomial-time implementation of Prange's information-set algorithm using $\calO(M)$ classical membership queries in the worst case such that
    \begin{align}
    \scor(Q_{\rm Prange})
    &\xrightarrow{\Prb}\alpha_{\rm C}(R).
    \label{eq:frs-prange}
    \end{align}
    The convergence is in probability over the sampled instance and the algorithm's randomness.
    
    \item Under the same sampled-oracle distribution, for every fixed $0<\eta<1-\alpha_{\rm C}(R)$, there are constants $\theta_{R,\eta},c_{R,\eta}>0$ and $M_0=M_0(R,\eta)$ such that, for every admissible $M\geq M_0$, every randomized classical algorithm $A$ making at most
    $T_M\leq\exp(\theta_{R,\eta}M/\log M)$ queries in the worst case satisfies
    \begin{align}
    \Prb\left[
    \scor\bigl(A(\omega)\bigr)
    \geq\alpha_{\rm C}(R)+\eta
    \right]&\leq e^{-c_{R,\eta}M}.
    \label{eq:frs-classical-lower}
    \end{align}
    The algorithm may use unlimited computation and arbitrary finite advice, provided its description and advice are fixed before $\mathcal A$ is sampled. The probability is over the sampled oracle $\mathcal A$ and the algorithm's random tape $\omega$, which is independent of $\mathcal A$.
\end{enumerate}
\end{theorem}

\begin{proof}
Part~\textup{(i)} is proved in \cref{sec:proof-quantum}, with full details in \cref{app:dqi}; parts~\textup{(ii)} and \textup{(iii)} are proved in \cref{sec:proof-prange,sec:proof-classical}.
\end{proof}

For every rate $R\in(0,1)$, the quantum score $\alpha_{\rm Q}(R)$ strictly exceeds the classical benchmark $\alpha_{\rm C}(R)$.

\paragraph{From expected score to verified threshold search.}
\Cref{thm:frs-separation}\textup{(i)} guarantees an expected score. To obtain an algorithm that reaches a certain score with bounded error, we first lower-bound the probability that a single DQI run reaches that score.
For any random variable $X\in[0,1]$ and threshold $\beta<\E[X]$, $\Prb[X\geq\beta] \geq \frac{\E[X]-\beta}{1-\beta}$. Fix a rational $\beta\in(\alpha_{\rm C}(R),\alpha_{\rm Q}(R))$. By \eqref{eq:frs-dqi-value}, for all sufficiently large admissible $M$,
\begin{align}
\Prb\left[
\scor(Q_{\rm DQI})\geq\beta
\right]
&\geq
\frac{\alpha_{\rm Q}(R)-\beta}{2(1-\beta)}.
\label{eq:expectation-to-search-concrete}
\end{align}
uniformly over every promised $\mathcal A$. Thus one DQI run succeeds with positive constant probability. A candidate can be verified using $M$ membership queries, so a constant number of verified repetitions raises the success probability to at least $2/3$ using $\calO_{R,\beta}(M)$ queries. If no candidate passes verification within the prescribed number of attempts, the algorithm reports failure.

\paragraph{Score and multiplicative approximation.}
To express the separation in terms of multiplicative approximation, we must relate the achieved scores to the optimum of each instance.
Since the optimum is at most one, the verified quantum threshold guarantee above also gives a $\beta$-approximation on every promised oracle. For the classical direction, a large approximation ratio need not imply a large score when the instance optimum is small. We therefore use \cref{lem:typical-perfect-satisfiability}, which shows that $\operatorname{OPT}=1$ except with doubly exponentially small probability under the sampled-oracle distribution. On the event $\operatorname{OPT}=1$, approximation ratio and score coincide, so the classical score query lower bound transfers to approximation.

\begin{corollary}[Ratio separation for DQI]
\label{cor:typical-dqi-ratio-separation}
Fix a rational rate $R\in(0,1)$ and a rational $r$ satisfying $\alpha_{\rm C}(R)<r<\alpha_{\rm Q}(R)$.
There is a DQI algorithm using $\calO_{R,r}(M)$ queries and quantum polynomial time. For all sufficiently large admissible $M$ and every promised oracle $\mathcal A$, the algorithm returns with probability $\geq 2/3$ a polynomial $Q_{{\rm DQI},r}$ with $\operatorname{ratio}(Q_{{\rm DQI},r})\geq r$. There also exist constants $\theta_{R,r},c_{R,r}>0$ such that, for all sufficiently large admissible $M$, every classical algorithm $A$ satisfying the hypotheses from \cref{thm:frs-separation}\textup{(iii)} with $T_M\leq\exp(\theta_{R,r}M/\log M)$ satisfies
\begin{align}
\Prb\left[
\operatorname{ratio}\bigl(A(\omega)\bigr)\geq r
\right]
&\leq e^{-c_{R,r}M}.
\label{eq:typical-dqi-classical-ratio-bound}
\end{align}
Prange's algorithm from \cref{thm:frs-separation}\textup{(ii)} achieves an approximation ratio converging in probability to $\alpha_{\rm C}(R)$.
The quantum success probability is over the algorithm's randomness for each fixed promised oracle. The classical probability and the Prange convergence are over the sampled oracle and the algorithm's independent randomness.
\end{corollary}

\begin{proof}
See \cref{app:ordinary-dqi-consequences}.
\end{proof}

\section{Proof of the concrete separation}
\label{sec:concrete-proof}

The three parts of \cref{thm:frs-separation} use complementary properties of the same code. First, we show that DQI achieves the limiting expected score $\alpha_{\rm Q}(R)$ using conventional unique decoding of the unfolded scalar dual (\cref{sec:proof-quantum}). Next, Prange's algorithm gives a classical algorithm whose score converges to $\alpha_{\rm C}(R)$ on sampled instances (\cref{sec:proof-prange}). Finally, we combine list recoverability with an analysis of adaptive queries to prove that exceeding $\alpha_{\rm C}(R)$ by any fixed amount with constant probability requires super-polynomially many queries (\cref{sec:proof-classical}).

\subsection{Part \textup{(i)}: quantum guarantee}
\label{sec:proof-quantum}

The central idea of DQI is to reduce the problem of approximating a max-LINSAT instance to coherent syndrome decoding for the dual linear code. Jordan et al.~\cite[Supplementary Theorem~1.1]{JordanEtAl2025DQI} prove the resulting semicircle score law for scalar max-LINSAT. Motivated by the connection to Yamakawa and Zhandry~\cite{YamakawaZhandry2024}, their Supplementary Sections~11.2--11.4 construct the corresponding folded DQI states and the implementation of the required oracle. They do not, however, state the corresponding semicircle score law for folded predicates. 

For our concrete application, two points therefore remain to be verified. First, we show that the semicircle calculation extends to balanced folded predicates and that the required states can be prepared directly from our membership oracles. Second, we show that the resulting decoding problem can be solved efficiently for our folded Reed--Solomon family.

\paragraph{Extension of the semicircle law to folded codes.}
Fix an integer $M$ and consider a balanced folded max-LINSAT instance specified by an encoder $B$ with block rows $B_i$, $i\in[M]$, and acceptance sets $\mathcal A=(A_i)_{i\in[M]}$. Define the signed objective
\begin{align}
f(x)
&\coloneq
\sum_{i\in[M]}
\left(2\cdot \mathds 1_{\{B_i x\in A_i\}}-1\right).
\label{eq:folded-signed-objective}
\end{align}
Thus $f(x)$ is the number of accepted blocks minus the number of rejected blocks.\footnote{Since $\scor(Bx)=1/2+f(x)/(2M)$, maximizing the score is equivalent to maximizing $f(x)$.} The exact-balance promise implies that each local sign predicate has zero average over the block alphabet.

In this language, DQI reduces the task of finding a high-scoring candidate to unique decoding for the dual code $C^\perp=\ker B^T$, in syndrome form. More precisely, suppose that every label $e$ of block weight at most $\ell$ can be recovered from its syndrome $B^Te$. This decoding must be implemented coherently, since DQI applies it to a superposition of syndromes. Given an efficient coherent decoder and efficiently preparable weights, DQI prepares an approximation to the normalized state
\begin{align}
\ket{P(f)}
&\propto
\sum_{x\in\F_q^K}P(f(x))\ket{x},
\label{eq:main-dqi-filter-state}
\end{align}
where $P$ is a nonzero real polynomial of degree at most $\ell$. Appendix~\ref{app:dqi} supplies an explicit asymptotically-optimal choice for that polynomial and proves preparation to inverse-polynomial accuracy. Measuring this state in the computational basis returns a candidate $x$ with probability proportional to $\abs{P(f(x))}^2$. The filter polynomial $P$ can therefore be chosen to bias the output toward candidates of large score.

The same parameter $\ell$ controls both the strength of this bias and the required decoding radius: a degree-$\ell$ filter produces errors supported on at most $\ell$ blocks. Increasing $\ell$ permits a stronger bias toward high scores, but requires correcting more block errors and therefore makes the decoding problem harder.

\begin{proposition}[Balanced folded-DQI semicircle law]
\label{prop:folded-dqi-from-bdd}
Consider a family of balanced folded max-LINSAT objectives $f$ as above, indexed by the number $M$ of blocks, and define $C^\perp \coloneq\ker B^T$. Given a polynomial filter $P$ of degree at most $\ell$, let $\langle\scor\rangle_P$ denote the expected score of the candidate obtained by measuring the corresponding DQI state $\ket{P(f)}$ in the computational basis. Suppose
\begin{align}
2\ell+1
&<d_{\rm blk}(C^\perp).
\label{eq:folded-dqi-collision-condition}
\end{align}
In the limit $M\to\infty$, with $\ell/M\to\mu\in[0,1/2]$, the optimal choice of a degree-at-most-$\ell$ filter $P$ yields
\begin{align}
\langle\scor\rangle_P
&=
\frac12+\sqrt{\mu(1-\mu)}+o(1)
\label{eq:folded-dqi-semicircle}
\end{align}
where the $o(1)$ is uniform over all balanced acceptance-set collections.

\end{proposition}

The condition $2\ell+1<d_{\rm blk}(C^\perp)$ ensures that $C^\perp$ can in principle be decoded from up to $\ell$ worst-case block errors. If this unique decoding of up to $\ell$ errors can be implemented efficiently and coherently, then the corresponding DQI algorithm is efficient as well.
We prove the proposition in \cref{app:dqi}, where we also give an explicit asymptotically optimal filter and implement the folded DQI state from membership-oracle access. The proposition may therefore be treated as a black box here: the only remaining task is to verify the distance condition and supply an efficient coherent decoder.

\paragraph{Decoding the folded Reed--Solomon dual.}
Recall from \eqref{eq:frs-complete-block-count} that $J\coloneq\floor{K/h}=\floor{RM}$, and set $\ell_M\coloneq\floor{(J-1)/2}$. We verify separately that this choice satisfies the distance condition in the proposition and that the corresponding decoder is efficient.

For our encoder $B$ with folded-Vandermonde structure, the primal folded code is $C=C_{\rm{FRS}} = \operatorname{Fold}_h(\operatorname{RS}_{\boldsymbol\alpha,K})$. After unfolding, the dual of $\operatorname{RS}_{\boldsymbol\alpha,K}$ is the generalized Reed--Solomon code $\operatorname{GRS}_{\boldsymbol\alpha,\boldsymbol\nu,N-K}$, with $\boldsymbol\nu$ as defined in \eqref{eq:rs-dual-multipliers}. Therefore, by \eqref{eq:rs-dual-grs} and \eqref{eq:folded-rs-dual-grs},
\begin{align}
{C}^\perp
&=
\operatorname{Fold}_h\left(
\bigl(\operatorname{RS}_{\boldsymbol\alpha,K}\bigr)^\perp
\right)
=
\operatorname{Fold}_h\left(
\operatorname{GRS}_{\boldsymbol\alpha,\boldsymbol\nu,N-K}
\right).
\end{align}
The unfolded generalized Reed--Solomon code has minimum distance $K+1$. A codeword supported on $b$ folded blocks has unfolded scalar weight at most $hb$. Since every nonzero word in the unfolded dual has scalar weight at least $K+1$, every nonzero word in $C^\perp$ must have block weight at least
\begin{align}
\ceil{\frac{K+1}{h}}
&=
\floor{\frac Kh}+1
=
J+1.
\end{align}
Hence $d_{\rm blk}(C^\perp)\geq J+1$. Moreover, our choice $\ell_M=\floor{(J-1)/2}$ gives $2\ell_M+1\leq J$. Combining these two bounds,
\begin{align}
2\ell_M+1
&\leq J
<J+1
\leq d_{\rm blk}(C^\perp).
\label{eq:main-folded-collision-check}
\end{align}
This verifies the distance condition \eqref{eq:folded-dqi-collision-condition} in \cref{prop:folded-dqi-from-bdd}.

It remains to implement the decoder. An error pattern affecting at most $\ell_M$ folded blocks produces at most $h\ell_M$ scalar errors after unfolding. Since
$2\ell_M\leq J-1$,
\begin{align}
h\ell_M
&\leq
\frac{h(J-1)}2
<
\frac{K+1}{2}.
\end{align}
The unfolded generalized Reed--Solomon code has minimum distance $K+1$, so the required number of scalar errors is strictly below half the minimum distance. It can therefore be corrected efficiently by conventional generalized Reed--Solomon unique decoding~\cite{Berlekamp1968,Massey1969}. Thus the required decoder can be implemented efficiently and applied coherently in superposition, making the corresponding DQI algorithm efficient.

Since $J=\floor{RM}$ and $\ell_M=\floor{(J-1)/2}$, we have
\begin{align}
\frac{\ell_M}{M}
&=
\frac R2+\calO\left(\frac1M\right)
\rightarrow\frac R2.
\end{align} 
For our folded Reed--Solomon instance, the measured candidate is the coefficient vector of a polynomial; denote this output polynomial by $Q_{\rm DQI}$. Applying \cref{prop:folded-dqi-from-bdd} with $\mu=R/2$ gives
\begin{align}
\E\left[
\scor\bigl(Q_{\rm DQI}\bigr)
\right]
&=
\frac{1+\sqrt{R(2-R)}}2+o(1)
=
\alpha_{\rm Q}(R)+o(1),
\end{align}
uniformly over every balanced instance. Together with the membership-oracle implementation from \cref{app:dqi}, this gives a quantum polynomial-time algorithm using one quantum query to each block-membership oracle in $(\mathcal O_{A_i})_{i\in [M]}$ and hence at most $M$ queries in total. This concludes the main argument underlying \cref{thm:frs-separation}\textup{(i)}, full details are given in \cref{app:dqi}.

\subsection{Part \textup{(ii)}: Prange's algorithm}
\label{sec:proof-prange}

We use the max-LINSAT adaptation of Prange's algorithm~\cite{Prange1962} described by Jordan et al.~\cite[Supplementary Section~8.3]{JordanEtAl2025DQI}. The idea is to satisfy as many complete blocks as we can by interpolation. For each of the selected blocks, we sample symbols from $\Sigma$ independently and uniformly, with replacement, and test them using the corresponding block-membership oracle until we find an accepted target. For every fixed balanced instance, conditional on the entire preceding history, each trial succeeds with probability exactly $1/2$. We then interpolate a polynomial whose selected evaluation blocks equal these targets. Under the sampled-instance distribution, each remaining block accepts independently with probability $1/2$. Since we can fit exactly $J$ blocks, which amounts to roughly an $R$ fraction of the blocks, the resulting score approaches $R\cdot 1 +(1-R)\cdot \frac{1}{2}=(1+R)/2$.

\begin{proof}[Proof of \cref{thm:frs-separation}\textup{(ii)}]
With $J$ as defined in \eqref{eq:frs-complete-block-count}, select the first $J$ blocks of the folded OPI instance. Each selected block contributes $h$ target values, so the polynomial must match $hJ\leq K$ values in total. The evaluation points are distinct, and therefore any such values can be realized by a polynomial of degree less than $K$.

First find an accepted target $z_i\in A_i$ for each selected block $i \in [J]$. In the current block, sample independent uniform symbols from $\Sigma$, with replacement, until a membership query accepts; then move to the next block. For every fixed balanced instance, conditional on the entire preceding history, each new trial succeeds with probability exactly $1/2$. Thus the trial outcomes are independent Bernoulli random variables with parameter $1/2$. Cap the total number of queries at
\begin{align}
\tau_M&\coloneq 4M.
\label{eq:prange-query-cap}
\end{align}
If the cap is reached before all $J$ targets are found, return the zero polynomial.

For the probability calculation only, append independent dummy trials after all targets have been found. Failure then means that fewer than $J$ successes occur in $\tau_M$ independent Bernoulli trials. Since their mean is $2M$ and $J\leq M$, a Chernoff bound gives
\begin{align}
\Prb[\text{search fails}]
&=\Prb[\operatorname{Bin}(\tau_M,1/2)<J]
\leq e^{-M/4}.
\label{eq:prange-search-failure}
\end{align}

Once an accepted target $z_i=(z_{i,j})_{j\in[h]}\in A_i$ has been found for every selected block $i\in[J]$, choose $Q$ to satisfy
\begin{align}
Q(\gamma^{ih+j})
&=z_{i,j},
\qquad i\in[J],\quad j\in[h].
\label{eq:prange-prescribed-values}
\end{align}
These are $hJ$ target values at $hJ$ distinct evaluation points. Polynomial interpolation therefore gives a unique polynomial $Q$ of degree less than $hJ$ satisfying all of them. Since $hJ\leq K$, this polynomial also has degree less than $K$ and is therefore a valid output. By construction, its $i$th evaluation block equals $z_i\in A_i$ for every $i\in[J]$, so all $J$ selected blocks accept. Interpolation over $\F_q$ can be performed in polynomial time.

Denote the algorithm's output, including the zero-polynomial fallback on search failure, by $Q_{\rm Prange}$. The algorithm uses at most $4M$ membership queries, and both sampling and polynomial interpolation take polynomial time.

It remains to determine the score on the untouched blocks. Condition on a successful target search and on the accepted target chosen in each of the $J$ selected blocks. This fixes the interpolated polynomial. The remaining acceptance sets are still independent uniform half-size subsets of $\Sigma$, because none of them was queried. Each remaining evaluation block therefore accepts independently with probability $1/2$.

Consequently, conditional on search success, the number of accepted blocks has distribution
\begin{align}
J+\operatorname{Bin}(M-J,1/2).
\label{eq:prange-conditional-law}
\end{align}
Here $J$ blocks accept with certainty because we fitted them, and the binomial term counts the additional acceptances among the untouched blocks. This distribution does not depend on which accepted targets were found.

The conditional expected score is $(1+J/M)/2$, which converges to $(1+R)/2=\alpha_{\rm C}(R)$ because $J/M\to R$. Its variance is $(M-J)/(4M^2)\leq1/(4M)$, so Chebyshev's inequality gives convergence in probability to this limit. Finally, the search-failure probability tends to zero, so the same convergence holds for the unconditional output.
\end{proof}

\subsection{Part \textup{(iii)}: classical query lower bound}
\label{sec:proof-classical}
We prove that any classical algorithm achieving a fixed improvement over the score $\alpha_{\rm C}(R)=(1+R)/2$ with high probability requires a super-polynomial number of queries. In principle, our argument works for any code with sufficiently strong so-called \emph{list recoverability} properties. Here, we exploit that for folded Reed-Solomon codes, particularly strong guarantees are available.

We start by developing a proof strategy for general codes. The general argument proceeds in two stages. First, using list recoverability, we upper-bound the number of codewords that could have been uncovered within $L$ queries. Then, in a second step, we bound the probability that any of them exceeds a given target score. With this framework in place, we then insert concrete guarantees for folded Reed-Solomon codes to obtain the final classical query lower bound.

Throughout this subsection, fix a nonempty code $C\subseteq\Sigma^M$ independently of the oracle. The sets $A_i$ are sampled independently and uniformly among the $|\Sigma  |/2$-element subsets of $\Sigma$. The output of an algorithm is a codeword $c_{\rm out}\in C$. Throughout this section, we assume that the number of queries $L=L(M)$ is small compared to the alphabet size, formally $L(M)=o(|\Sigma|)$. Prior queries then change the acceptance probability of any unqueried symbol only negligibly: conditioned on the answers so far, it is at most $1/2 + \epsilon$, where $\epsilon = o(1)$ is independent of the queries made and their answers.

\paragraph{A parametrized upper bound on the number of possible code words.}
Making at most $L$ queries, we learn the membership of at most $L$ distinct addressed symbols across the blocks $i=0, \dots, M-1$. Sorting these $L$ queried symbols into lists $S_0, \dots S_{M-1}$, where $S_i$ contains the distinct symbols queried in block $i$, we obtain a collection of $M$ lists such that
\begin{align}
  \sum_{i=0}^{M-1} |S_i| \leq L.
\end{align}
For now, consider a fixed set of $L$ queries. For a given threshold $\tau \in [0,1]$, consider the set of codewords for which membership is known on at least a $\tau$-fraction of the blocks as a result of the $L$ queries. That is, codewords $c =(c_0, \dots c_{M-1})$ such that for at least $\lceil \tau M\rceil$ of the $c_i$, membership has been queried within the $L$ queries. Formally, we write 
\begin{align}
  \mathcal{B}_\tau(S_0, \dots, S_{M-1}) \coloneqq \left\{c \in C : \sum_{i=0}^{M-1} \mathds{1}_{\{c_i \in S_i\}} \geq \tau M \right\}.
\end{align}
\begin{remark}
  The natural quantity that one would like to control is the set of codewords that are known to accept on at least a $\tau$-fraction of blocks, that is
  \begin{align}
    \mathcal{B}_\tau^{\mathrm{ideal}}(S_0, \dots, S_{M-1}) \coloneqq \left\{c \in C: \sum_{i=0}^{M-1} \mathds{1}_{\{c_i \in S_i \cap A_i\}} \geq \tau M \right\}.
  \end{align}
  Note that since $\{c_i \in S_i\}$ is the more permissive condition, we have $\mathcal{B}_\tau^{\mathrm{ideal}}(S_0, \dots, S_{M-1}) \subseteq \mathcal{B}_\tau(S_0, \dots, S_{M-1})$.
\end{remark}

For the analysis, we can assume w.l.o.g. that we may append $M$ verification queries to an algorithm using at most $T$ queries, giving a total budget $L=T+M$. These \emph{verification queries} are used to determine on how many blocks the output $c_{\mathrm{out}}$ in fact accepts. These $M$ additional queries do not change the output produced by the algorithm, but simplify the analysis: They ensure that the final output $c_{\text{out}}$ is contained in $\mathcal{B}_\tau(S_0, \dots, S_{M-1})$ for all values $\tau \in [0,1]$.
Note that this does not place a restriction on the algorithms we consider. Any algorithm without verification queries maps to a corresponding algorithm with verification queries that exactly reproduces the original algorithms input-output relation.

So far, we have fixed a sequence of $L$ queries. Now, we generalize to the setting of arbitrary query sequences of length at most $L$. In particular, we are interested in the maximal size that $\mathcal{B}_\tau(S_0, \dots, S_{M-1})$ can take, a quantity which we will denote with $\LR_C(\tau,L)$. We define
\begin{align}
  \LR_C(\tau,L) \coloneqq \max_{\substack{S_i\subseteq\Sigma\ (i\in[M])\\ |S_i|\leq L}} \left|\mathcal B_\tau(S_0,\dots,S_{M-1})\right|.
  \label{eq:LR}
\end{align}

The $L$ queries give lists satisfying $\sum_i|S_i|\leq L$, and hence $|S_i|\leq L$ for every block. This form lets us apply the list-recovery bounds for folded Reed--Solomon codes.

\paragraph{Binomial tail estimate.}

We defined the quantity $\LR_C(\tau,L)$, which bounds the number of codewords for which we could have determined membership on $\lceil \tau M \rceil$ blocks within $L$ queries. We can also view this as an upper bound on the number of codewords which, after $L$ queries, we know to be accepted on at least $\lceil \tau M \rceil$ of the blocks. Now, for some fixed target $\beta$ and threshold $\tau$, we ask for the probability that any codeword from our candidate pool is accepted on at least $\ceil{\beta M}$ blocks.

\begin{lemma} \label{thrm:lower-bound-framework}
  Let $\tau, \beta \in [0,1]$. Let a randomized classical algorithm make at most $T$ queries and output $c_{\text{out}} \in C$. Set 
  \begin{align}
    p(\tau, \beta, M) \coloneqq \mathbb{P} \left[ \lceil \tau M \rceil + \mathrm{Bin}\left( M-\lceil \tau M \rceil, 1/2 + \epsilon \right) \geq \beta M \right].
  \end{align}
  Then, the output satisfies
  \begin{align}
    \mathbb{P} \left[ \mathrm{score}(c_{\text{out}}) \geq \beta \right] \leq \LR_C(\tau, T+M) p(\tau,\beta, M).
  \end{align}
\end{lemma}

\begin{proof}
Append the $M$ verification queries, so the algorithm makes at most $T+M$ queries and its output is unchanged. Denote with $E^{(\tau)}_c$ the event that the codeword $c$ is queried on at least $\lceil \tau M \rceil$ of its symbols, that is 
\begin{align}
    E_c^{(\tau)} \coloneqq \left\{c \in \mathcal{B}_\tau(S_0, \dots, S_{M-1}) \right\}
\end{align}

The verification queries ensure $E_{c_{\mathrm{out}}}^{(\tau)}$ occurs for all $\tau \in [0,1]$ Then, we have 
\begin{align}
    \mathbb{P}\left[\mathrm{score}(c_{\mathrm{out}}) \geq \beta\right] &\leq \sum_{c \in C} \mathbb{P} \left[ E_c^{(\tau)} \cap \{ \mathrm{score}(c) \geq \beta \}\right] \\
    &= \sum_{c \in C} \mathbb{P} \left[ \{ \mathrm{score}(c) \geq \beta \} \middle| E_c^{(\tau)}\right] \mathbb{P} \left[ E_c^{(\tau)}\right] \\
    &\leq \sum_{c \in C} p(\tau, \beta, M) \mathbb{P} \left[ E_c^{(\tau)}\right] \\
    &= p(\tau, \beta, M) \mathbb{E} \left[ \sum_{c \in C} \mathds{1}_{E_c^{(\tau)}}\right] \\
    &\leq p(\tau, \beta, M) \LR_C(\tau, T+M)
\end{align}
The first inequality follows because the $M$ verification queries ensure that $c_{\mathrm{out}} \in E_{c_{\mathrm{out}}}^{(\tau)}$ for all $\tau \in [0,1]$. For the second inequality, consider the first time a given codeword $c$ has been queried on $\lceil \tau M \rceil$ blocks, i.e., the moment at which $E_c^{(\tau)}$ occurs. Since the randomness of the algorithm is independent from the oracle, we can consider the remaining $M-\lceil \tau M \rceil$ blocks independent random variables, each with acceptance probability at most $1/2 + \epsilon$. Even if codeword $c$ was accepted on all its $\lceil \tau M \rceil$ queried symbols, its probability of achieving a score at least $\beta$ is upper bounded by $p(\tau, \beta, M)$. For the last inequality, note that 
\begin{align}
    \mathbb{E} \left[ \sum_{c \in C} \mathds{1}_{E_c^{(\tau)}}\right] &= \mathbb{E} \left[ \left| \mathcal{B}_\tau(S_0, \dots, S_{M-1}) \right|\right] \\
    &\leq \LR_C(\tau, T+M).
\end{align}
\end{proof}

\begin{remark}
  \cref{thrm:lower-bound-framework} gives a framework for proving classical query lower bounds, requiring as an ingredient only sufficiently strong upper bounds on $\LR_C(\tau, L)$. In particular, note that for fixed $\beta$, any choice of $\tau$ gives rise to a valid, although potentially vacuous, upper bound on the success probability $P[\mathrm{score}(c_{\text{out}})\geq \beta]$. Therefore, it is sufficient to identify a single suitable parameter choice for $\tau$ to obtain the query lower bound we desire.
\end{remark}
\paragraph{Query lower bound for folded Reed-Solomon codes.}
For the rest of this subsection, fix a rational $R\in(0,1)$ and $0<\eta<1-\alpha_{\rm C}(R)$, and take $C=C_{\rm FRS}$ from \cref{sec:density-oracle-model}. To obtain a super-polynomial classical query lower bound for a score at least $\beta$, we bound both $\LR_C(\tau,L)$ and $p(\tau,\beta,M)$ such that the probability of exceeding score $\beta$ vanishes super-polynomially in $M$. For folded Reed-Solomon codes, strong bounds on $\LR_C(\tau,L)$ are known by way of list recoverability guarantees. 

\begin{lemma}[List recoverability bound {\cite[Sec.~3.6]{Rudra2007}} via {\cite[Lemma~4.3]{YamakawaZhandry2024}}] \label{lem:list-recoverability-bound}
Consider a folded Reed-Solomon code with parameters $q=(M+1)^2$, $N=q-1$, $h=M+2$ and $K=\lfloor RN \rfloor$. Fix $\tau \coloneqq R+\eta$ and $\theta_{R,\eta} \coloneqq \frac{\eta^2}{800} \log \left( \frac{R+\eta}{R} \right)$. Then, whenever
\begin{align}
  T \leq \exp \left(\theta_{R,\eta} \frac{M}{\log(M)} \right),
\end{align}
it holds that for sufficiently large admissible $M$,
\begin{align}
  \LR_C(R+\eta, T+M) \leq \exp \left( \frac{3 \eta^2}{100} M \right).
\end{align}
\end{lemma}

\begin{proof}
Apply \mbox{\cite[Lemma 4.3]{YamakawaZhandry2024}} with $\zeta=1-\tau$, $m = h$, $k = K-1$, $\ell = T+M$,
$s \coloneqq \lceil \frac{\eta^2}{100}\frac{M}{\log M}\rceil$ and $r \coloneqq \lceil\sqrt{sM}\,\rceil$. Its code
equals $C_{\mathrm{FRS}}$ via $Q(Y)\longmapsto Q(\gamma Y)$, and $0<\zeta<1$, $s\le m$. Since $s/r\to0$, $s/h\to0$
and $(N/k)^{1/(s+1)}\to1$, while $\ell^{1/(s+1)} \le (1+o(1))\,e^{100\theta_{R,\eta}/\eta^2} = (1+o(1))\big(\frac{R+\eta}{R}\big)^{1/8}$,
the right-hand side of condition (8) is at most $(1+o(1))R^{7/8}(R+\eta)^{1/8}M < (R+\eta)M$, its left-hand side.
Condition (9) holds since its left-hand side is $O(M) < q$. Finally,
$q^s = \exp(2s\log(M+1)) \le \exp\big((1+o(1))\frac{\eta^2}{50}M\big) \le \exp\big(\frac{3\eta^2}{100}M\big)$.
\end{proof}

\begin{lemma}[Binomial tail bound] \label{lem:binomial-tail-bound}
  Define $p(\tau, \beta, M)$ as in \cref{thrm:lower-bound-framework}. Let $\tau \coloneqq R+\eta$ and $\beta \coloneqq (1+R)/2 + \eta$. Furthermore, let $T+M = o(|\Sigma|)$, that is the number of queries we make is negligible compared to the alphabet size. Then, for sufficiently large $M$, 
  \begin{align}
    p\left(\tau,\beta,M\right) \leq \exp \left(- \frac{2 \eta^2}{9}M \right).
  \end{align}
\end{lemma}
\begin{proof}
Set $X\sim\operatorname{Bin}(M-\ceil{\tau M},1/2 + \epsilon)$, where $T+M=o(|\Sigma|)$ ensures that the acceptance probability is in fact at most $1/2+\epsilon$ with $\epsilon=o(1)$. Also note that $M-\ceil{\tau M}>0$ for sufficiently large $M$, since $\tau<1$. Then, $\mathbb{E}[X] = (1-R-\eta)M/2 + o(M)$, such that $\beta M - \lceil \tau M \rceil - \mathbb{E}[X] \geq (\eta M)/3$. Then, applying Hoeffding's inequality for the $M - \lceil \tau M \rceil$ independent Bernoulli variables gives 
\begin{align}
    p(\tau,\beta,M)= \mathbb{P} \left[ X \geq \beta M - \lceil \tau M \rceil \right] \leq \exp\left( - \frac{2 (\eta M/3)^2}{M- \lceil \tau M \rceil} \right) \leq \exp \left( -\frac{2 \eta^2}{9} M\right).
\end{align}
\end{proof}

\begin{lemma}[Classical hardness]
\label{thm:classical-lower-bound}
Fix a rational $R\in(0,1)$ and $0<\eta<1-\alpha_{\rm C}(R)$. For all sufficiently large admissible $M$, every classical randomized algorithm making at most $T\leq\exp(\theta_{R,\eta}M/\log M)$ queries in the worst case satisfies
\begin{align}
\Prb\left[\scor(c_{\text{out}})\geq\frac{1+R}{2}+\eta\right]
&\leq\exp\left(-\frac{\eta^2M}{10}\right).
\end{align}
Here $\theta_{R,\eta}$ is defined in \cref{lem:list-recoverability-bound}. Probability is over the sampled oracle and the algorithm's independent randomness, with its description and advice fixed before sampling.
\end{lemma}
\begin{proof}
  Set $\tau \coloneqq R + \eta$ and $\beta = \alpha_C(R) + \eta$. Then, following \cref{thrm:lower-bound-framework}, we have
  \begin{align}
    \mathbb{P} \left[\text{score}(c_{\text{out}}) \geq \beta \right] &\leq \LR_{C_{\rm FRS}}(\tau,T+M)\,p(\tau,\beta,M) \\
    & \leq \exp\left(\frac{3 \eta^2}{100}M\right) \exp \left( -\frac{2 \eta^2}{9} M \right) \\
    &\leq \exp \left( - \frac{\eta^2}{10} M\right),
  \end{align}
where we used \cref{lem:list-recoverability-bound} and \cref{lem:binomial-tail-bound} to bound the quantities for the second inequality. Since $T+M\leq\exp(\theta_{R,\eta}M/\log M)+M$ and $|\Sigma|=M^{\Theta(M)}$, we have $T+M = o(|\Sigma|)$. Thus the tail bound also holds uniformly for all sufficiently large admissible $M$, with the size threshold depending only on $R,\eta$.
\end{proof}

\begin{proof}[Proof of \cref{thm:frs-separation}\textup{(iii)}]
    This is immediate from \cref{thm:classical-lower-bound}.
\end{proof}

\section{Stronger quantum guarantees on typical oracles}
\label{sec:sampled-score-hierarchy}

We obtain stronger quantum guarantees for the same balanced folded OPI family by combining complete list decoding of the folded dual with Jo's coherent fiber-summation method~\cite{Jo2026}. Recall the balanced semicircle law from
\cref{prop:folded-dqi-from-bdd}. Write
\begin{align}
\alpha_{\rm sc}(\mu)
&\coloneq
\frac12+\sqrt{\mu(1-\mu)},
\qquad 0\leq\mu\leq\frac12,
\label{eq:semicircle-score}
\end{align}
for its limiting score at relative Fourier-layer cutoff $\mu$. DQI uses unique decoding, which permits a cutoff approaching $R/2$ and gives $\alpha_{\rm Q}(R)=\alpha_{\rm sc}(R/2)$.

Beyond this cutoff, a syndrome may correspond to several retained errors. These labels form a \emph{syndrome fiber}. Following Jo~\cite{Jo2026}, we use a complete list decoder to find every retained label in a fiber and combine their amplitudes coherently before the final Fourier transform.

The decoding ingredient differs from Jo's scalar Reed--Solomon instantiation, which uses the deterministic list decoder of Chatterjee, Harsha, and Kumar~\cite{ChatterjeeHarshaKumar2026} within the scalar Johnson radius. We instead invoke Guruswami's linear-algebraic list decoder specifically for \emph{folded} Reed--Solomon codes~\cite{Guruswami2011Linear}, together with the list-size bound of Chen and Zhang~\cite{ChenZhang2025Relaxed}. Applied to our folded dual, these imported results yield the complete decoder summarized in \cref{lem:frs-complete-list-decoder}: 
for every fixed rational cutoff $0<\mu<R$, it runs in polynomial time and returns a list of constant size for all sufficiently large $M$. The list-size constant may depend on the fixed gap $R-\mu$ between the cutoff and the decoding radius.

Small lists alone do not ensure appreciable preparation probability, because amplitudes within a fiber can cancel. For $0<\mu\leq1/2$ with $\mu<R$, we prove that, outside a doubly exponentially small set of sampled instances, fiber summation succeeds with constant probability and, conditioned on this success, the output has limiting expected score $\alpha_{\rm sc}(\mu)$. Thus these stronger guarantees hold on typical sampled instances, whereas DQI works for every promised instance. Optimizing over the allowed cutoffs gives the envelope in the following theorem. When $R>1/2$, a separate construction also achieves exact score one; this does not follow merely from a limiting expected score of one.

\begin{theorem}[Stronger quantum algorithms on typical oracles]
\label{thm:sampled-score-hierarchy}
Fix a rational rate $R\in(0,1)$ and put
\begin{align}
    \alpha_{\rm fib}(R) &\coloneq
    \begin{cases}
        \frac12+\sqrt{R(1-R)},&0<R\leq\frac12,\\
        1,&\frac12<R<1.
    \end{cases}
    \label{eq:headline-fiber-envelope}
\end{align}
The first branch is the semicircle curve $\alpha_{\rm sc}(R)$. For the balanced sampled-oracle family, the following hold. The success probabilities below are over the algorithm's randomness for each fixed instance.

\begin{enumerate}
    \item \emph{Improved score thresholds.}
    For every fixed rational $0<\beta<\alpha_{\rm fib}(R)$, there is a quantum algorithm using $\calO_{R,\beta}(M)$ queries and polynomially many gates. For all sufficiently large admissible $M$, except on a set of sampled instances of probability at most $\exp\{-\exp(\Omega_{R,\beta}(M\log M))\}$, the algorithm returns with probability $\geq 2/3$ a polynomial $Q_{{\rm fib},\beta}$ with $\scor(Q_{{\rm fib},\beta})\geq\beta$.
    
    \item \emph{Exact search for $R>1/2$.}
    If $R>1/2$, there is a quantum algorithm using $\calO_R(M)$ queries and polynomially many gates. For all sufficiently large admissible $M$, except on a set of sampled instances of probability at most $\exp\{-\exp(\Omega_R(M\log M))\}$, the algorithm returns with probability $\geq 2/3$ a polynomial $Q_{\rm exact}$ with $\scor(Q_{\rm exact})=1$.
    Every returned polynomial satisfies every block and therefore has approximation ratio one.
\end{enumerate}

The function $\alpha_{\rm fib}$ is an achievable envelope, not an upper bound on arbitrary quantum algorithms. Part~\textup{(i)} attains only thresholds strictly below it; exact score one is obtained by part~\textup{(ii)} only when $R>1/2$. No exact-search claim is made at $R=1/2$.
\end{theorem}

\begin{proof}
    Part~\textup{(i)} is proved in \cref{sec:proof-sampled-coherent-fib}. Part~\textup{(ii)} is proved in \cref{sec:proof-sampled-exact}. Both parts use complete list decoding of the folded dual, obtained from the folded Reed--Solomon list-decoding algorithm of Guruswami~\cite{Guruswami2011Linear} and the list-size bound of Chen and Zhang~\cite{ChenZhang2025Relaxed}. Part~\textup{(ii)} also uses the concentration of perfect solutions from \cref{lem:typical-perfect-satisfiability}.
\end{proof}

\paragraph{Proof mechanism.}
For approximation, we retain Fourier layers up to a fixed rational relative cutoff $\mu$ satisfying $0<\mu\leq1/2$ and $\mu<R$. Complete list decoding finds all retained Fourier labels compatible with each syndrome, and Jo's coherent fiber-summation method~\cite{Jo2026} combines their amplitudes. A bounded list size makes this procedure efficient, but does not by itself ensure a non-negligible success probability, because amplitudes within a fiber may interfere destructively. We show that, on typical sampled instances, both the state normalization and the score numerator concentrate around their expectations. Consequently, fiber summation succeeds with constant probability and, conditioned on success, the output has limiting expected score $\alpha_{\rm sc}(\mu)$. Verification and a constant number of independent attempts then give every fixed threshold below the stated envelope.

For exact search, the goal is to find a block string that is both a codeword and accepted in every block. The starting point, due to Yamakawa and Zhandry~\cite{YamakawaZhandry2024}, is to prepare a superposition over strings accepted in every block and use Fourier decoding to enforce the code constraint. Horinaga and Yamakawa~\cite{HorinagaYamakawa2026} develop this approach for worst-case OPI using list decoding. Because a list decoder may return several errors compatible with the same syndrome, we use Jo's coherent fiber-summation method~\cite{Jo2026} to add their amplitudes before transforming back to candidate polynomials.

The threshold $R>1/2$ is where the strong list-decoding guarantees of folded Reed--Solomon codes enter~\cite{GuruswamiRudra2008,Guruswami2011Linear}. The folded dual has rate $1-R+o(1)$ and can be list decoded from any fixed fraction of block errors below $R$.\footnote{This is where folding improves the exact-search range over the scalar Reed--Solomon setting of Horinaga and Yamakawa~\cite{HorinagaYamakawa2026}. Their analysis uses the scalar Reed--Solomon list-decoding radius $1-\sqrt{1-R}$ for the dual code, whose rate is $1-R$. Reaching radius $1/2$ therefore requires $1-\sqrt{1-R}>1/2$, equivalently $R>3/4$. The stronger capacity-approaching list-decoding guarantee for the folded dual reaches every fixed radius below $R$, reducing the corresponding threshold to $R>1/2$.} 
Hence, when $R>1/2$, the decoder handles every error supported on at most half of the blocks. Balance supplies the complementary fact: in each block, the Fourier probability is split equally between the zero and nonzero frequencies. Consequently, since $M$ is odd, exactly half of the total Fourier probability lies within this decodable range. Together with concentration of the number of perfect solutions, this gives a constant probability of obtaining a verified polynomial accepted in every block.

Separately, there is a useful comparison at the level of decoding parameters. Yamakawa and Zhandry also instantiate their construction with folded Reed--Solomon codes, but their quantum decoding step unfolds the dual code and applies a scalar generalized Reed--Solomon list decoder~\cite{YamakawaZhandry2024}. Together with their random-error analysis, this leads them to assume a primal rate greater than $5/6$. Our exact-search construction instead uses complete list decoding directly in the folded block metric~\cite{Guruswami2011Linear,ChenZhang2025Relaxed}. Since the folded dual can be decoded from every fixed block-error fraction below $R$, this covers the half-block regime whenever $R>1/2$; the accompanying constant list-size bound makes coherent fiber summation efficient. This is a comparison of decoding ingredients, not a direct improvement of their oracle separation. Yamakawa and Zhandry use an i.i.d.\ vector-output random oracle, whereas our model uses separately callable one-bit block predicates subject to exact balance. The precise connection is established in \cref{sec:yz-dictionary}.

\paragraph{Resulting query separations.}
The classical side is unchanged: \cref{thm:frs-separation} places the classical threshold at $\alpha_{\rm C}(R)$, and Prange's algorithm attains it. The stronger quantum algorithms above simply enlarge the separation window. For every fixed rational threshold
\begin{align}
    \alpha_{\rm C}(R)<r<\alpha_{\rm fib}(R),
\end{align}
linearly many quantum queries suffice on typical sampled instances, whereas achieving constant success classically requires $\exp(\Omega_{R,r}(M/\log M))$ queries under the same distributional quantifiers as in \cref{thm:frs-separation}\textup{(iii)}.

The same statement holds for multiplicative approximation. A score of at least $r$ is automatically an $r$-approximation because the optimum is at most one; conversely, typical sampled instances have optimum one by \cref{lem:typical-perfect-satisfiability}, so score and approximation ratio coincide on those instances. For $R>1/2$, part~\textup{(ii)} additionally gives an exact-search separation: every perfect classical output would also cross any fixed threshold $r\in(\alpha_{\rm C}(R),1)$ and is therefore covered by the same classical query lower bound. Thus, compared with DQI, coherent fiber summation replaces the upper endpoint $\alpha_{\rm Q}(R)$ by the strictly larger $\alpha_{\rm fib}(R)$ and, above rate one half, reaches exact search.

\section{Discussion}
\label{sec:discussion}

In this work, we have established a query-complexity separation for approximate optimization by DQI. For our balanced folded-OPI family, Prange's algorithm attains the classical benchmark $\alpha_{\rm C}(R)$, whereas DQI achieves the strictly larger limiting expected score $\alpha_{\rm Q}(R)$ at every fixed rational rate $R$. Here, DQI uses unique decoding of the dual code and at most $M$ coherent membership queries. Conversely, under the sampled-oracle distribution, achieving any fixed improvement over $\alpha_{\rm C}(R)$ with constant probability requires superpolynomially many classical queries in the range stated in \cref{thm:frs-separation}, even for fully adaptive algorithms with unlimited computation between queries.

We further showed that stronger guarantees are available for this problem setup. The key insight here is to replace \emph{unique decoding} with \emph{complete list decoding} of the folded dual and then to coherently combine labels with the same syndrome. This yields higher score guarantees on \emph{typical sampled oracles} and \emph{exact search} for rates above one half; see \cref{thm:sampled-score-hierarchy}. For completeness, we establish separations both in terms of scores and approximation ratios, following \cref{cor:typical-dqi-ratio-separation}.

Our work constitutes an \emph{unconditional approximation separation for DQI} relative to an oracle, against computationally unbounded classical algorithms, thus resolving an open question raised in Refs.~\cite{JordanEtAl2025DQI,Aaronson2024DQI,MarwahaEtAl2026}. The novelty of this work lies primarily in the resulting separation, rather than in the invention of a new coding-theoretic or quantum primitive. We do not claim a new code family, decoding paradigm, or DQI mechanism: many of the individual ingredients are inherited from earlier work. What was previously missing was a common setting in which these ingredients yield a \emph{rigorous quantum advantage for approximation}, together with a classical query lower bound identifying the matching benchmark.

The proof uses the duality of the underlying code $C$: On the primal side, \emph{list recoverability} of the code $C$ limits the number of codewords exposed by the classical queries, while a conditional binomial-tail argument controls the probability that any such codeword receives sufficiently many additional acceptances to exceed a certain score threshold. On the dual side, the available \emph{coherent decoding radius} of $C^\perp$ directly determines the score that the quantum algorithm can achieve. A provable separation results whenever this quantum score exceeds the classical threshold. Identifying further code families that combine sufficiently strong primal list recoverability with efficient coherent decoding of the dual is a natural direction for future work.

One caveat of our work is that it relies on an oracle. That being said, the code, evaluation points, and all algebraic structure are explicit; the oracle only answers whether a proposed block symbol belongs to the corresponding acceptance set, and both the classical and quantum algorithms receive access to this same membership interface. This is a natural model whenever acceptance can be tested without explicitly enumerating an exponentially large set.
The main role of the oracles is to prevent efficient classical algorithms from exploiting any structure in the acceptance sets.
A clear open direction is thus to instantiate the oracle: to provide a similar result to ours conditionally on computational complexity assumptions.
Finally, the exact balance condition is used to simplify the analysis and keep the exposition clear, but we do not expect it to be essential. For any fixed rational acceptance density, the same proof architecture should apply after centering and normalizing the local predicates: the classical analysis would involve biased acceptance trials, while the quantum analysis would produce corresponding density-dependent score guarantees.

\subsection*{Acknowledgements}
We thank Alexander Nietner for useful discussions, and Jean-Pierre Seifert and Carsten Schubert for early discussions about the possibility of proving a rigorous quantum advantage for DQI.
This work has been supported by the BMFTR (PraktiQOM, QuSol, HYBRID++), the Munich Quantum Valley, Berlin Quantum, the QuantERA, the Quantum Flagship (Millenion, PasQuans2), the DFG (CRC 183, SPP 2514), the Clusters of Excellence (MATH+, ML4Q), and the European Research Council (DebuQC).
E.~G.-F. is a recipient of the 2023 Google PhD Fellowship.

\subsection*{AI-disclosure}
Some of the authors had investigated this question extensively before the public release of OpenAI's GPT-5.6 Sol. The authors supplied the research question, mathematical setting, and intended scope to GPT-5.6 Sol. Through an iterative exchange, the model proposed candidate proof strategies and generated preliminary theorem statements and proof drafts. The central proof strategy ultimately used in this manuscript emerged through that exchange.
The authors subsequently used OpenAI's GPT-5.6 Sol and GPT-6 Astra and Anthropic's Claude Opus 5.5 to help check calculations and proof consistency, compare notation and assumptions with cited work, and assist with simplification, exposition, and manuscript restructuring. The authors verified the mathematical claims, corrected and substantially revised the model-generated material, and developed the final arguments. The authors understand, can defend, and clearly explain all parts of the work and take full responsibility for it.

\begingroup
\sloppy
\printbibliography
\endgroup

\appendix
\addtocontents{toc}{\protect\setcounter{tocdepth}{1}}
\crefalias{section}{appendix}

\section{Detailed proof of the DQI guarantee}
\label{app:dqi}

We prove \cref{thm:frs-separation}\textup{(i)} for every fixed collection of balanced acceptance sets. Jordan et al.~\cite[Supplementary Theorem~1.1 and Supplementary Sections~6.2--6.4]{JordanEtAl2025DQI} derive the scalar Jacobi-matrix formula and the resulting semicircle score law for scalar DQI. Their Supplementary Sections~11.2--11.3 develop the circuit for folded instances, while Supplementary Section~11.4 treats its implementation under oracle access to the acceptance sets. They do not, however, derive the corresponding semicircle score law for folded codes. Here we establish this extension in the balanced setting and apply it to our folded Reed--Solomon family. References to sections and lemmas of Jordan et al. below refer to the journal version and its supplementary material.

\subsection{The code and its syndrome decoder}
\label{app:dqi-frs-specialization}

Fix a rational rate $R\in(0,1)$ and an admissible block count $M$ from the construction in \cref{sec:density-oracle-model}. For convenience, we repeat the relevant parameters here. By definition of admissibility for the parameter $M$, $M+1$ is a power of two; write $M+1=2^u$. The associated parameters are
\begin{align}
     h&\coloneq M+2,&
     N&\coloneq hM,&
     q&\coloneq N+1=(M+1)^2,\\
     K&\coloneq\floor{RN},&
     J&\coloneq\floor{K/h}=\floor{RM}.
\end{align}
Admissibility additionally requires $K\geq h$. Here $J$ is the largest number of complete blocks satisfying $hJ\leq K$. Put $V\coloneq\F_q^K$ for the candidate space and $\Sigma\coloneq\F_q^h$ for the block alphabet. For a primitive element $\gamma\in\F_q^\times$, the folded evaluation encoder is
\begin{align}
     Q_x(Y)&\coloneq\sum_{j=0}^{K-1}x_jY^j,&
     Bx&\coloneq
     \left(\bigl(Q_x(\gamma^{ih+t})\bigr)_{t\in[h]}\right)_{i\in[M]}.
\end{align}
Since $K<N$ and the evaluation points are distinct, $B$ is injective, and $C\coloneq\operatorname{im}B=C_{\rm FRS}$. Fix an arbitrary balanced family $\mathcal A=(A_i)_{i\in[M]}$, with oracle access as in \cref{sec:balanced-oracle}. In particular, each oracle call invokes one classically chosen $\calO_{A_i}$. As in \eqref{eq:folded-signed-objective}, define the signed objective
\begin{align}
    f(x) &\coloneq \sum_{i\in[M]} \left(2\cdot \mathds{1}_{\{B_i x\in A_i\}}-1\right).
    \label{eq:app-dqi-signed-objective}
\end{align}
Thus $f(x)$ is the number of accepted blocks minus the number of rejected blocks. Consequently, the score of $Q_x$ is
\begin{align}
    \scor(Q_x) =\frac12+\frac{f(x)}{2M}.
    \label{eq:app-dqi-concrete-score}
\end{align}
For the blockwise analysis below, denote the $i$th local sign predicate by
\begin{align}
    f_i(z) &\coloneq 2\cdot \mathds{1}_{\{z\in A_i\}}-1.
\end{align}
Since $|A_i|=|\Sigma|/2$, each $f_i$ is balanced over the block alphabet $\Sigma$. 
For Fourier analysis when $q=2^m$ for a positive integer $m$, we use the binary-valued pairing $\langle u,v\rangle_W\coloneq\operatorname{Tr}_{\F_q/\F_2}(u^Tv)$ for $u,v\in W$, where $W\in\{V,\Sigma,\Sigma^M\}$, unfolding block vectors when taking the dot product. Here, the field trace is $\operatorname{Tr}_{\F_q/\F_2}(a)=\sum_{j=0}^{m-1}a^{2^j}\in\F_2$, so $\langle u,v\rangle_W$ takes values in ${0,1}$.

\paragraph{The decoding radius.}
\label{app:dqi-frs-geometry}
In the scalar setting, the number of columns of the encoder controls the possible dual distance: for an $N\times K$ encoder with $K<N$, the scalar dual distance is at most $K+1$. Our Reed--Solomon encoder attains this bound. Conventional unique decoding therefore corrects every error pattern of scalar Hamming weight strictly less than $(K+1)/2$.

In the folded setting, an error word consists of $M$ blocks, each containing $h$ field elements. Its block weight counts the nonzero blocks, whereas its scalar Hamming weight counts the nonzero field elements after unfolding. An error supported on at most $\ell_M$ blocks can therefore contain at most $h\ell_M$ scalar errors. Thus the scalar decoder suffices whenever
$2h\ell_M<K+1$. Likewise, every nonzero dual codeword must occupy at least
\begin{align}
    d_{0,M}  \coloneq \ceil{\frac{K+1}{h}} =J+1
\end{align}
blocks, where $J=\floor{K/h}$. Hence $d_{\rm blk}(C^\perp)\geq d_{0,M}$.

For DQI, the cutoff must satisfy both this decoding requirement and the slightly stronger distance condition used by the score calculation. We therefore choose, as in \cref{sec:proof-quantum},
\begin{align}
    \ell_M &\coloneq \floor{\frac{J-1}{2}}.
\end{align}
This gives $2h\ell_M<K+1$ and
$2\ell_M+1\leq J<d_{\rm blk}(C^\perp)$.
The following lemma verifies these bounds and constructs the
coherent syndrome decoder required by DQI.

\begin{lemma}[Folded Reed--Solomon decoding interface]
    \label{lem:dqi-frs-geometry}
    The map $B^T$ is a syndrome map for $C^\perp$. For the above cutoff,
    \begin{align}
         2\ell_M+1&<d_{\rm blk}(C^\perp),& 2h\ell_M&<K+1,&\ell_M/M&\rightarrow R/2.
         \label{eq:dqi-frs-radius-check}
    \end{align}
    A  polynomial-time decoder $\operatorname{Dec}_M:V\to\Sigma^M$ satisfies $\operatorname{Dec}_M(B^Te)=e$ whenever $\operatorname{wt}_{\rm blk}(e)\leq\ell_M$ and admits a reversible implementation with polynomial overhead.
\end{lemma}
\begin{proof}
The identity $C^\perp=\ker B^T$ shows that $B^T$ is a syndrome map for $C^\perp$. As reviewed in \cref{sec:codes-and-folding}, unfolding $C^\perp$ gives the generalized Reed--Solomon dual of the evaluation code. This code has scalar minimum distance $K+1$. Dividing by its known nonzero coordinate multipliers, applying conventional Reed--Solomon decoding~\cite{Berlekamp1968,Massey1969}, and restoring the multipliers therefore gives an efficient \emph{word decoder} for $C^\perp$: whenever a received word $r$ lies within fewer than $(K+1)/2$ scalar errors of a codeword in $C^\perp$, the decoder recovers that codeword.

A word supported on $b$ nonzero blocks has at most $hb$ nonzero scalar coordinates. Since every nonzero dual codeword has scalar weight at least $K+1$, it follows that
\begin{align}
    d_{\rm blk}(C^\perp) &\geq \ceil{(K+1)/h}  = d_{0,M}  = J+1.
    \label{eq:dqi-frs-distance-lower}
\end{align}
Using $\ell_M=\lfloor(J-1)/2\rfloor$, we obtain
\begin{align}
    2\ell_M+1\leq J<d_{0,M}\leq d_{\rm blk}(C^\perp), \qquad 2h\ell_M\leq h(J-1)<K+1.
\end{align}
Moreover, $J=\lfloor RM\rfloor$ gives $\ell_M=RM/2+\calO(1)$, and hence $\ell_M/M\rightarrow R/2$.

Consequently, every error supported on at most $\ell_M$ blocks unfolds to at most $h\ell_M<(K+1)/2$ scalar errors. Although DQI measures errors in block weight, conventional Reed--Solomon unique decoding therefore suffices at the chosen cutoff; no dedicated unique decoder for folded codes is needed.

DQI, however, requires a syndrome decoder rather than a word decoder. We obtain one as follows. Because $B$ is injective, $B^T$ is surjective. Given a syndrome $\sigma\in V$, Gaussian elimination therefore finds a word $r\in\Sigma^M$ satisfying $B^Tr=\sigma$. If $\sigma=B^Te$ for an error pattern $e$, then $B^T(r-e)=B^Tr-B^Te=0$, so $r-e\in C^\perp$. In other words, $r$ is a dual codeword corrupted by exactly the error $e$. Recovering $e$ from its syndrome therefore reduces to decoding $r$ to that dual codeword and subtracting.

Apply the word decoder to $r$. For every promised error $e$ with $\operatorname{wt}_{\rm blk}(e)\leq\ell_M$, its scalar weight is at most $h\ell_M<(K+1)/2$. Unique decoding therefore returns the dual codeword $\widehat r=r-e$, and hence $e=r-\widehat r$. Return $r-\widehat r$ when decoding succeeds and zero on failure. This defines a total deterministic polynomial-time function $\operatorname{Dec}_M:V\to\Sigma^M$ satisfying $\operatorname{Dec}_M(B^Te)=e$ for every promised error.

Compute--copy--uncompute implements this function coherently as
\begin{align}
    \ket{\sigma}\ket{0}\ket{0} &\longmapsto \ket{\sigma}\ket{\operatorname{Dec}_M(\sigma)}\ket{0},
    \label{eq:decoder-interface}
\end{align}
where the last register is workspace. The reversible implementation has polynomial overhead.
\end{proof}

\subsection{The folded DQI score law}
\label{app:dqi-score-proof}

Jordan et al.~\cite[Supplementary Lemmas~6.2--6.3]{JordanEtAl2025DQI} derive the balanced scalar DQI score law by restricting the score observable to the retained Fourier layers. We follow the same argument with one Fourier symbol per folded block. The only additional point is that the dual block distance must exclude collisions between retained labels. The following score calculation uses no Reed--Solomon structure. It applies to any full-column-rank encoder $B:\F_q^K\to(\F_q^h)^M$ over a field of characteristic two, with balanced local predicates and the trace pairings above. Put $C\coloneq\operatorname{im}B$, so that $C^\perp=\ker B^T$, and let $0\leq\ell\leq M$ satisfy
\begin{align}
    2\ell+1&<d_{\rm blk}(C^\perp).
    \label{eq:strict-radius}
\end{align}
For our folded Reed--Solomon family, we take $\ell=\ell_M$; this condition then follows from \cref{lem:dqi-frs-geometry}.

\paragraph{Fourier layers.}
Use the unitary Fourier transforms
\begin{align}
    \operatorname{QFT}_\Sigma\ket z &\coloneq\frac1{\sqrt{|\Sigma|}}\sum_{\eta\in\Sigma} (-1)^{\langle z,\eta\rangle_\Sigma}\ket\eta, 
    \label{eq:app-block-fourier}\\
    \operatorname{QFT}_V\ket x &\coloneq\frac1{\sqrt{|V|}}\sum_{\xi\in V} (-1)^{\langle x,\xi\rangle_V}\ket\xi.
    \label{eq:app-message-fourier}
\end{align}
Nondegeneracy makes these transforms unitary; symmetry and characteristic two make them self-inverse. For each block $i$, define the Fourier coefficients of its sign predicate $f_i$ by
\begin{align}
    \widetilde\phi_i(\eta) &\coloneq\frac1{|\Sigma|}\sum_{z\in\Sigma} f_i(z)(-1)^{\langle z,\eta\rangle_\Sigma}.
    \label{eq:app-fourier}
\end{align}
Balance gives $\widetilde\phi_i(0)=|\Sigma|^{-1}\sum_{z\in\Sigma}f_i(z)=0$. Parseval's identity, with the normalization in \eqref{eq:app-fourier}, gives
\begin{align}
    \sum_{\eta\in\Sigma}|\widetilde\phi_i(\eta)|^2 =\frac1{|\Sigma|}\sum_{z\in\Sigma}|f_i(z)|^2=1,
\end{align}
because $f_i(z)\in\{\pm1\}$. For $y=(y_i)_{i\in[M]}\in\Sigma^M$, put
\begin{align}
    \widetilde\phi(y) &\coloneq\prod_{i:y_i\ne0}\widetilde\phi_i(y_i).
\end{align}
Equivalently, the block Fourier transform maps the normalized sign state to its Fourier coefficients:
\begin{align}
    \operatorname{QFT}_\Sigma \left( \frac1{\sqrt{|\Sigma|}}\sum_{z\in\Sigma}f_i(z)\ket z \right) &=\sum_{\eta\in\Sigma}\widetilde\phi_i(\eta)\ket\eta.
    \label{eq:app-local-fourier-state}
\end{align}
The zero-frequency amplitude vanishes by the balance condition.

We now follow the order of the folded-layer construction in Jordan et al.~\cite[Supplementary Equations~(234)--(242)]{JordanEtAl2025DQI}. For $0\leq j\leq\ell$, define the degree-$j$ layer polynomial, evaluated at $f(x)$, by
\begin{align}
    P_j(f(x)) &\coloneq\frac1{\sqrt{\binom Mj}} \sum_{\substack{S\subseteq[M]\\|S|=j}} \prod_{i\in S}f_i(B_ix),
    \label{eq:app-polynomial-dqi-layer}
\end{align}
and the corresponding computational-basis layer by
\begin{align}
    \ket{P_j(f)} &\coloneq \frac1{\sqrt{|V|}}\sum_{x\in V}P_j(f(x))\ket x.
    \label{eq:app-polynomial-dqi-layer-state}
\end{align}
Fourier inversion of every selected local predicate gives
\begin{align}
    f_i(B_ix) &=\sum_{\eta\in\Sigma}\widetilde\phi_i(\eta) (-1)^{\langle B_ix,\eta\rangle_\Sigma}.
\end{align}
Because $\widetilde\phi_i(0)=0$, multiplying these expansions and collecting the labels into $y\in\Sigma^M$ yields
\begin{align}
    P_j(f(x)) &=\frac1{\sqrt{\binom Mj}} \sum_{\operatorname{wt}_{\rm blk}(y)=j} \widetilde\phi(y) (-1)^{\langle Bx,y\rangle_{\Sigma^M}} =\frac1{\sqrt{\binom Mj}} \sum_{\operatorname{wt}_{\rm blk}(y)=j} \widetilde\phi(y) (-1)^{\langle x,B^Ty\rangle_V}.
    \label{eq:app-support-character-expansion}
\end{align}
The second equality is the defining adjoint identity for $B^T$. For every $\xi\in V$, the message-space Fourier transform satisfies
\begin{align}
    \operatorname{QFT}_V \left( \frac1{\sqrt{|V|}}\sum_{x\in V} (-1)^{\langle x,\xi\rangle_V}\ket x \right) &=\ket\xi.
    \label{eq:app-character-to-syndrome}
\end{align}
Applying this identity term by term to \eqref{eq:app-support-character-expansion} gives
\begin{align}
    \ket{\widetilde{P}_j} &\coloneq \operatorname{QFT}_V\ket{P_j(f)} =\frac1{\sqrt{\binom Mj}} \sum_{\operatorname{wt}_{\rm blk}(y)=j} \widetilde\phi(y)\ket{B^Ty}.
    \label{eq:app-layer}
\end{align}
Here $y$ is the blockwise error and $B^Ty$ is its syndrome. The states $\ket{\widetilde P_0},\ldots,\ket{\widetilde P_\ell}$ are orthonormal. Indeed, the map $y\longmapsto B^Ty$ is injective on labels of block weight at most $\ell$: equality of two syndromes would give $y-y'\in C^\perp$ with block weight at most $2\ell<d_{\rm blk}(C^\perp)$, and hence $y=y'$. Moreover, the balance condition and Parseval give
\begin{align}
    \norm{\widetilde P_j}^2 &= \frac1{\binom Mj} \sum_{\substack{S\subseteq[M]\\|S|=j}} \prod_{i\in S} \sum_{\eta\ne0}\left|\widetilde\phi_i (\eta)\right|^2 =1.
\end{align}
Injectivity also makes layers of different weights orthogonal. Since $\operatorname{QFT}_V$ is unitary, the computational-basis layers $\ket{P_j(f)}$ are orthonormal as well.

Finally, the layer amplitude depends on $x$ only through $f(x)$. Indeed, exactly $(M+f(x))/2$ of the local signs are positive, and hence
\begin{align}
    P_j(f(x)) &=\frac1{\sqrt{\binom Mj}} \sum_{a=0}^j(-1)^{j-a} \binom{(M+f(x))/2}{a} \binom{(M-f(x))/2}{j-a}.
    \label{eq:app-layer-counting-polynomial}
\end{align}
Viewed as a function of its scalar argument, the right-hand side is a univariate polynomial of degree $j$. Hence $P_0,\ldots,P_\ell$ form a basis for the filters of degree at most $\ell$.

For a real unit vector $\mathbf w=(w_0,\ldots,w_\ell)^T$, the corresponding filter assigns amplitude
\begin{align}
    P_{\mathbf w}(f(x)) &\coloneq \sum_{j=0}^{\ell}w_jP_j(f(x)).
    \label{eq:app-filter-from-layers}
\end{align}
Here $P_{\mathbf w}$ is a real polynomial of degree at most $\ell$, evaluated at the signed objective $f(x)$.

The corresponding normalized quantum state is
\begin{align}
 \ket{P_{\mathbf w}(f)}
 &\coloneq \frac1{\sqrt{|V|}}\sum_{x\in V}P_{\mathbf w}(f(x))\ket x
 =\sum_{j=0}^{\ell}w_j\ket{P_j(f)}.
 \label{eq:app-jordan-dqi-superposition}
\end{align}

\paragraph{The expected-score matrix.}
Define the score observable
\begin{align}
    \mathsf S_f \coloneq\sum_{x\in V}\scor(Bx)\ket x \bra x.
    \label{eq:app-score-observable}
\end{align}
If $X_{\mathbf w}$ denotes the computational-basis measurement outcome of \eqref{eq:app-jordan-dqi-superposition}, then
\begin{align}
    \E\bigl[\scor(BX_{\mathbf w})\bigr]
    &=\bra{P_{\mathbf w}(f)}\mathsf S_f\ket{P_{\mathbf w}(f)}.
    \label{eq:app-score-expectation}
\end{align}
The following is the folded, balanced analogue of \cite[Supplementary Lemma~6.2]{JordanEtAl2025DQI}.

\begin{proposition}[Balanced folded-DQI score law]
\label{prop:dqi}
For every real unit vector $\mathbf w\in\mathbb R^{\ell+1}$,
\begin{align}
 \E\bigl[\scor(BX_{\mathbf w})\bigr]
 &=\frac12+\frac{\mathbf w^TA^{(M,\ell,0)}\mathbf w}{2M},
 \label{eq:finite-dqi-general}
\end{align}
where $A^{(M,\ell,0)}$ is the symmetric tridiagonal matrix with zero diagonal and
\begin{align}
 A^{(M,\ell,0)}_{j+1,j}
 &=A^{(M,\ell,0)}_{j,j+1}
 =\sqrt{(j+1)(M-j)},
 \qquad 0\leq j<\ell.
 \label{eq:jacobi}
\end{align}
Consequently, the optimal expected score over these retained layers is
\begin{align}
 \frac12+\frac{\lambda_{\max}(A^{(M,\ell,0)})}{2M}.
 \label{eq:finite-dqi}
\end{align}
For fixed $\mathbf w$, the expectation is independent of the particular balanced acceptance sets.
\end{proposition}

\begin{proof}
Let $\widetilde{\mathsf S}_f\coloneq \operatorname{QFT}_V\mathsf S_f\operatorname{QFT}_V^{-1}$.
Fourier inversion and $\langle Bx,y\rangle_{\Sigma^M}=\langle x,B^Ty\rangle_V$ show how one local objective term acts on an error. If $e_i\eta\in\Sigma^M$ denotes the word supported on block $i$ with value $\eta$, then
\begin{align}
\widetilde{\mathsf S}_f\ket{B^Ty'}
&=\frac12\ket{B^Ty'}
 +\frac1{2M}\sum_{i\in[M]}\sum_{\eta\in\Sigma}
 \widetilde\phi_i(\eta)
 \ket{B^T(y'+e_i\eta)}.
\label{eq:app-fourier-shift}
\end{align}
Thus one objective term shifts at most one folded block of the Fourier label.

The block-distance condition ensures that there are no additional syndrome collisions among the labels relevant to this matrix element. For retained $y,y'$,
\begin{align}
    B^Ty=B^T(y'+e_i\eta)
    \quad\Longleftrightarrow\quad
    y=y'+e_i\eta,
    \label{eq:app-folded-collision}
\end{align}
because equality of the syndromes puts $y-y'-e_i\eta$ in $C^\perp$, while its block weight is at most $2\ell+1<d_{\rm blk}(C^\perp)$.

We can now read off which layers may couple. A one-block shift has only three possibilities: it changes a zero block into a nonzero block and increases the weight by one; it changes a nonzero block into zero and decreases the weight by one; or it changes one nonzero block value into another and leaves the weight unchanged. Consequently,
\begin{align}
    \bra{\widetilde P_k}\widetilde{\mathsf S}_f\ket{\widetilde P_j} &=0
    \qquad\text{whenever }|k-j|>1.
    \label{eq:app-nonadjacent-layers}
\end{align}
It remains to compute the adjacent-layer and same-layer matrix elements.
For layers $j$ and $j+1$, a contributing shift inserts a nonzero value into one previously zero block. For each support of size $j$, there are $M-j$ choices of this new block, and Parseval gives $\sum_{\eta\ne0}|\widetilde\phi_i(\eta)|^2=1$ for its total local contribution. Including the layer normalizations yields
\begin{align}
    \bra{\widetilde P_{j+1}}
    \widetilde{\mathsf S}_f
    \ket{\widetilde P_j}
    &=\frac1{2M}
    \frac{\binom Mj(M-j)}
    {\sqrt{\binom Mj\binom M{j+1}}}
    =\frac1{2M}\sqrt{(j+1)(M-j)}.
    \label{eq:app-adjacent-entry}
\end{align}
The reverse entry is the same by symmetry.
On a fixed layer, the constant part of \eqref{eq:app-fourier-shift} contributes $1/2$. Consider a nonconstant contribution that preserves the layer weight. If the ket label has value $\zeta$ in block $i$ and the local objective term introduces the shift $\eta$, then the corresponding bra label has value $\zeta+\eta$. The ket layer, objective term, and bra layer contribute $\widetilde\phi_i(\zeta)$, $\widetilde\phi_i(\eta)$, and $\widetilde\phi_i(\zeta+\eta)$, respectively. Terms involving a zero frequency vanish because $\widetilde\phi_i(0)=0$, so we may sum over all $\zeta,\eta\in\Sigma$. The resulting triple Fourier convolution is
\begin{align}
    \sum_{\zeta,\eta\in\Sigma}
    \widetilde\phi_i(\zeta)
    \widetilde\phi_i(\eta)
    \widetilde\phi_i(\zeta+\eta)
    &=\frac1{|\Sigma|}\sum_{z\in\Sigma}f_i(z)^3.
    \label{eq:app-diagonal-entry}
\end{align}
Since $f_i(z)\in\{\pm1\}$, we have $f_i(z)^3=f_i(z)$, and balance makes the right-hand side zero. Thus the nonconstant same-layer contribution vanishes.

We have now exhausted all cases: layers separated by more than one do not couple, adjacent layers have the matrix element in \eqref{eq:app-adjacent-entry}, and a fixed layer receives only the constant contribution $1/2$. Therefore,
\begin{align}
    \bra{\widetilde P_k}\widetilde{\mathsf S}_f\ket{\widetilde P_j}
    &=\frac12\delta_{k,j}
     +\frac1{2M}A^{(M,\ell,0)}_{k,j}.
    \label{eq:app-layer-score-matrix}
\end{align}
Using \eqref{eq:app-score-expectation} and $\operatorname{QFT}_V\ket{P_{\mathbf w}(f)} =\sum_jw_j\ket{\widetilde P_j}$ gives \eqref{eq:finite-dqi-general}. The Rayleigh principle gives \eqref{eq:finite-dqi}.
\end{proof}

\paragraph{Explicit weights and the semicircle limit.}
Following the asymptotic argument of \cite[Supplementary Lemma~6.3]{JordanEtAl2025DQI}, we use equal amplitudes on a short window of the highest retained layers. Choosing the window width to be a power of two also permits exact preparation of its layer-register superposition.

\begin{lemma}[Explicit asymptotically optimal DQI weights]
\label{lem:dqi-terminal-window}
Let $\Delta$ be the largest power of two not exceeding $\min\{\ell+1,\floor{\sqrt M}\}$, and set
\begin{align}
 w_j^{\rm tw}&\coloneq
 \begin{cases}
 \Delta^{-1/2},&\ell-\Delta+1\leq j\leq\ell,\\
 0,&\text{otherwise}.
 \end{cases}
 \label{eq:terminal-window-vector}
\end{align}
The resulting expected score is
\begin{align}
 \E\bigl[\scor(BX_{\mathbf w^{\rm tw}})\bigr]
 &=\frac12+\frac1{M\Delta}
  \sum_{j=\ell-\Delta+1}^{\ell-1}\sqrt{(j+1)(M-j)},
 \label{eq:terminal-window-score}
\end{align}
with an empty sum equal to zero. If $M\to\infty$ and $\ell/M\to\mu\in[0,1/2]$, this score and the optimum in \eqref{eq:finite-dqi} both converge to
\begin{align}
 \alpha_{\rm sc}(\mu)&=\frac12+\sqrt{\mu(1-\mu)}.
\end{align}
\end{lemma}
\begin{proof}
The vector has $\Delta$ equal nonzero entries and is normalized. Its quadratic form is twice the sum of the $\Delta-1$ adjacent entries within the window, divided by $\Delta$, proving \eqref{eq:terminal-window-score}. If $\mu>0$, then $\Delta\to\infty$, $\Delta=o(M)$ and $M^{-1}\sqrt{(j+1)(M-j)}\to\sqrt{\mu(1-\mu)}$ uniformly on the window. The factor $(\Delta-1)/\Delta$ tends to one. If $\mu=0$, each normalized summand is at most $\sqrt{(\ell+1)/M}=o(1)$.

For the matching upper bound, the maximum absolute row-sum bound shows that the largest eigenvalue is at most twice the largest off-diagonal entry. When $\mu<1/2$ and $\ell\geq1$, eventually this gives $\lambda_{\max}\leq2\sqrt{\ell(M-\ell+1)}$; for $\ell=0$ the matrix is zero. When $\mu=1/2$, use $2\sqrt{(j+1)(M-j)}\leq M+1$. The window supplies the matching lower limit in both cases.
\end{proof}

\begin{proof}[Proof of \cref{prop:folded-dqi-from-bdd}]
Under the distance condition, the layer states are orthonormal, and the layer polynomials $P_0,\ldots,P_\ell$ form a basis for the real polynomial filters of degree at most $\ell$. Thus optimizing over these filters is equivalent to optimizing over real unit vectors of layer weights. By \cref{prop:dqi}, the optimal expected score is
\begin{align}
    \frac12+\frac{\lambda_{\max}(A^{(M,\ell,0)})}{2M}.
\end{align}
By \cref{lem:dqi-terminal-window}, this converges to $\frac12+\sqrt{\mu(1-\mu)}$ whenever $\ell/M\to\mu\in[0,1/2]$. Since the finite-size expression depends only on $M$ and $\ell$, the convergence is uniform over the balanced acceptance sets.
\end{proof}

\subsection{Membership queries and the folded-DQI circuit}
\label{app:uniform-dqi-proof}

We now verify that the state analyzed above can be prepared efficiently from membership-oracle access. We use the Fourier transforms, coefficients, and layer states defined in \eqref{eq:app-block-fourier}--\eqref{eq:app-layer}. The distance condition makes $B^T$ injective on all retained labels, while the coherent decoder of \cref{lem:dqi-frs-geometry} allows those labels to be erased from their syndromes. It remains to implement the local predicate phases and assemble the retained layers.

Supplementary Section~11.4 of Jordan et al.~\cite{JordanEtAl2025DQI} implements the corresponding local Fourier gate from phase-oracle access under the normalization $f_i(0)=1$. The following adapts that step to membership-oracle access and requires no assumption on $f_i(0)$.

\begin{lemma}[One-query predicate phase]
\label{lem:membership-phase-oracle}
For each fixed block $i$, one membership query implements $\ket{\kappa,z}\longmapsto f_i(z)^\kappa\ket{\kappa,z}$, with a constant number of additional gates.
\end{lemma}
\begin{proof}
Prepare an answer qubit in $\ket{\chi_\kappa}=(\ket0+(-1)^\kappa\ket1)/\sqrt2$ using a CNOT from $\kappa$ followed by a Hadamard. Query $\mathcal O_{A_i}$, flip the answer qubit, and undo its preparation. The qubit returns to zero and phase kickback gives $(-1)^{\kappa O_{A_i}(z)}(-1)^\kappa=f_i(z)^\kappa$.
\end{proof}

\begin{lemma}[Implementation]
\label{lem:uniform-dqi}
For the folded Reed--Solomon family above with cutoff $\ell=\ell_M$, and for every fixed $c>0$, the state with weights \eqref{eq:terminal-window-vector} can be prepared to trace-distance error at most $M^{-c}$ using at most $M$ membership queries and polynomially many additional gates. Its measured distribution and expected score differ from the ideal ones by at most $M^{-c}$, uniformly over all balanced acceptance sets.
\end{lemma}
\begin{proof}
We follow the folded circuit of \cite[Supplementary Section~11.3]{JordanEtAl2025DQI}.

\emph{Prepare the supports.}
Because $\Delta$ is a power of two, Hadamards followed by a reversible offset prepare exactly
\begin{align}
 \frac1{\sqrt\Delta}\sum_{a=0}^{\Delta-1}
 \ket{\ell-\Delta+1+a}.
\end{align}
Conditioned on $j$, prepare a uniform binary mask of weight $j$. Initialize a separate remaining-ones counter to $j$, leaving the layer register unchanged. For completeness, with $r$ positions and $k$ ones left to place, rotate the next bit to $\sqrt{(r-k)/r}\ket0+\sqrt{k/r}\ket1$ and update the $k$-counter coherently by decrementing it when that bit is one. Every weight-$j$ mask has squared amplitude $j!(M-j)!/M!=\binom Mj^{-1}$. The remaining-ones counter ends at zero. Compute the mask's weight into a temporary counter, subtract it from the layer register and uncompute that counter. The resulting mask state is
\begin{align}
 \sum_{j=0}^{\ell}\frac{w_j^{\rm tw}}{\sqrt{\binom Mj}}
 \sum_{\substack{\kappa\in\{0,1\}^M\\|\kappa|=j}}\ket\kappa.
\end{align}
All counter operations are reversible; choose identity rotations on unreachable counter values to define the circuit everywhere.

\emph{Prepare the errors.}
Initialize one label register to zero per block. Apply a local QFT, the predicate phase controlled by $\kappa_i$, and an inverse QFT. The local result is
\begin{align}
 \ket{\kappa_i}\ket0&\longmapsto
 \begin{cases}
 \ket0\ket0,&\kappa_i=0,\\
 \ket1\displaystyle\sum_{\eta\ne0}\widetilde\phi_i(\eta)\ket\eta,
 &\kappa_i=1.
 \end{cases}
 \label{eq:app-balanced-local-preparation}
\end{align}
Balance removes the zero frequency on the selected branch. Thus $\kappa_i$ is precisely the nonzero indicator of the label $y_i$. Apply the reversible update $\kappa_i\leftarrow\kappa_i\oplus\mathds1_{\{y_i\ne0\}}$, returning any temporary workspace to zero. Since $\kappa_i$ equals this indicator, the update erases the mask bit. This uses one query per block, without requiring $f_i(0)=1$. The remaining state is
\begin{align}
 \sum_{j=0}^{\ell}\frac{w_j^{\rm tw}}{\sqrt{\binom Mj}}
 \sum_{\operatorname{wt}_{\rm blk}(y)=j}\widetilde\phi(y)\ket y.
 \label{eq:app-pre-syndrome-dqi-state}
\end{align}
This preparation before syndrome computation does not require the distance condition: distinct words $y$ are still stored in distinct basis states. It is also used by the stronger algorithm's.

\emph{Decode and transform.}
Compute $B^Ty$ into a second register. Use the syndrome decoder to compute $y$ into temporary workspace, subtract it from the first register, and reverse the decoder:
\begin{align}
 \ket y\ket0&\longmapsto\ket y\ket{B^Ty}
 \longmapsto\ket0\ket{B^Ty}.
\end{align}
Every ideal label has block weight at most $\ell$, so decoding is correct. By \eqref{eq:app-layer}, the remaining register is $\sum_jw_j^{\rm tw}\ket{\widetilde P_j}$. Applying $\operatorname{QFT}_V^{-1}$ gives $\ket{P_{\mathbf w^{\rm tw}}(f)}$.

\emph{Complexity and approximation.}
The binary QFTs are Hadamards composed with the invertible linear maps representing the pairings. Arithmetic, nonzero tests and the decoder have polynomial-size reversible implementations. The support rotations have efficiently computable angles, and each counter has at most $M+1$ relevant values. They therefore decompose into polynomially many elementary controlled rotations. Let $L(M)$ be the total number of synthesized elementary controlled rotations; it is polynomial in $M$. Synthesize each to operator-norm error at most $M^{-c}/L(M)$. A telescoping bound gives total operator-norm error at most $M^{-c}$, and hence the same trace-distance bound on the joint output state. All auxiliary qubits are clean in the ideal circuit. Approximate preparation may create a small component outside the decoding promise, but the decoder is a total reversible map, so the circuit remains well defined on that component. Applying this unitary cannot increase the state-vector error relative to the ideal computation. Discarding workspace and measuring cannot increase trace distance; the resulting distributions differ by at most the same amount in total variation, and therefore so do the expectations of any score in $[0,1]$. These bounds do not depend on the oracle.
\end{proof}

\subsection{Completion of the quantum guarantee}
\label{app:dqi-theorem-proof}

\begin{proof}[Proof of \cref{thm:frs-separation}\textup{(i)}]
Apply the algorithm described above to the folded Reed--Solomon family. Use as a cutoff $\ell_M$ as defined above. The decoder lemma verifies the strict distance condition and recovers every retained label, and both $\log|V|=K\log q$ and $\log|\Sigma|=h\log q$ are polynomial in $M$. Let $\bar\alpha_M$ denote the ideal mean in \eqref{eq:terminal-window-score}; it depends only on the 
parameters. Since $\ell_M/M\to R/2$, the window lemma gives
\begin{align}
 \bar\alpha_M&\rightarrow
 \alpha_{\rm sc}(R/2)
 =\frac{1+\sqrt{R(2-R)}}2=\alpha_{\rm Q}(R).
\end{align}
Use compilation accuracy $M^{-2}$ and interpret the measured coefficient vector as $Q_{\rm DQI}$. The circuit uses at most $M$ queries and polynomially many gates, and
\begin{align}
 \sup_{\substack{\mathcal A=(A_i)\\|A_i|=|\Sigma|/2\ (i\in[M])}}
 \left|\E[\scor(Q_{\rm DQI})]-\alpha_{\rm Q}(R)\right|
 &\leq |\bar\alpha_M-\alpha_{\rm Q}(R)|+M^{-2}=o(1).
\end{align}
This proves the claimed guarantee uniformly over every balanced instance.
\end{proof}

\section{Additional results for the DQI separation}
\label{app:ordinary-dqi-consequences}

This appendix proves \cref{cor:typical-dqi-ratio-separation}. We first show that a sampled instance has a perfect solution with overwhelming probability, and then use this fact to transfer the classical score query lower bound and the Prange benchmark to approximation ratios. The quantum ratio guarantee holds for every promised oracle.

Throughout this appendix, fix a rational rate $R\in(0,1)$ and the balanced folded family. Recall that $|\Sigma|=q^h$. Each promised set $A_i$ contains exactly $\abs{\Sigma}/2$ symbols. For sampled oracles $\mathcal A=(A_i)_{i\in[M]}$, these sets are independent and uniform among subsets of that size.

A perfect solution is a codeword accepted in every block. We show that a sampled instance has such a solution by proving that the number of perfect codewords is close to its expected value with overwhelming probability.

\begin{lemma}[Concentration of perfect solutions]
\label{lem:typical-perfect-satisfiability} 
Let $N_{\rm perf}(\mathcal A)$ denote the number of perfect codewords:
\begin{align}
    N_{\rm perf}(\mathcal A) &\coloneq \abs{\left\{
    c\in C_{\rm FRS}:c_i\in A_i\text{ for every }i\in[M]
    \right\}}.
\end{align}
Define the normalized count as
\begin{align}
    Z_{\rm perf}(\mathcal A)
    &\coloneq
    \frac{N_{\rm perf}(\mathcal A)}{|C_{\rm FRS}|\,2^{-M}}.
    \label{eq:exact-normalized-solution-count}
\end{align}
The denominator is the expected number of perfect codewords (each fixed codeword is accepted independently across the $M$ blocks with probability $1/2$ per block, so it is perfect with probability $2^{-M}$). Thus $Z_{\rm perf}=1$ means that the actual count equals its expectation. For all sufficiently large admissible $M$,
\begin{align}
    \Prb_{\mathcal A}\left[
    |Z_{\rm perf}(\mathcal A)-1|>\frac12
    \right]
    &\leq
    2(\abs{\Sigma}+1)^M
    \exp\left(-\frac{\abs{\Sigma}}{2M\,2^{2M}}\right)
    \leq
    \exp\{-\exp(\Omega_R(M\log M))\}.
    \label{eq:perfect-solution-concentration}
\end{align}
Consequently,
\begin{align}
    \Prb_{\mathcal A}[\operatorname{OPT}=1] &\geq 1-\exp\{-\exp(\Omega_R(M\log M))\}.
    \label{eq:typical-perfect-satisfiability}
\end{align}
\end{lemma}

\begin{proof}
We show that $Z_{\rm perf}$ concentrates around its mean of one. To apply a concentration inequality for independent variables, first consider an auxiliary oracle in which all $M\abs{\Sigma}$ membership bits are independent fair bits. The same expectation calculation gives $\E_{\rm aux}[Z_{\rm perf}]=1$.

Let $E_{\rm{bal}}$ be the event that each block accepts exactly $\abs{\Sigma}/2$ symbols. Conditioning the auxiliary oracle on $E_{\rm{bal}}$ gives precisely our promised sampled-oracle distribution.

This conditioning has a controlled cost. In one block, the number of accepted symbols has distribution $\operatorname{Bin}(\abs{\Sigma},1/2)$. Its most likely value is $\abs{\Sigma}/2$, and there are $\abs{\Sigma}+1$ possible values. Thus the probability of exact balance in one block is at least $1/(\abs{\Sigma}+1)$. Independence across blocks gives
\begin{align}
    \Prb_{\rm aux}[E_{\rm{bal}}]
    &\geq(\abs{\Sigma}+1)^{-M}.
    \label{eq:fixed-density-conditioning-probability}
\end{align}
For any bad event $E$, we therefore have
\begin{align}
    \Prb_{\mathcal A}[E]
    &=
    \Prb_{\rm aux}[E\mid E_{\rm{bal}}]
    \leq
    (\abs{\Sigma}+1)^M\Prb_{\rm aux}[E].
\end{align}

It remains to bound how much one membership bit can change the count. Changing the answer at $(i,z)$ can affect only codewords whose $i$th block equals $z$. Because $K\geq h$, interpolation allows any prescribed value in one block. The block projection is therefore a surjective linear map, so every symbol has exactly $|C_{\rm FRS}|/\abs{\Sigma}$ preimages. Consequently, changing one membership bit changes $N_{\rm perf}$ by at most $|C_{\rm FRS}|/\abs{\Sigma}$, and hence changes $Z_{\rm perf}$ by at most $2^M/\abs{\Sigma}$.

There are $M\abs{\Sigma}$ independent membership bits in the auxiliary model. The sum of their squared change bounds is
\begin{align}
    M\abs{\Sigma}\left(\frac{2^M}{\abs{\Sigma}}\right)^2
    &=\frac{M\,2^{2M}}{\abs{\Sigma}}.
\end{align}
McDiarmid's inequality, followed by the conditioning bound above, therefore gives, for every $\varepsilon>0$,
\begin{align}
    \Prb_{\mathcal A}\left[
    |Z_{\rm perf}(\mathcal A)-1|>\varepsilon
    \right]
    &\leq
    2(\abs{\Sigma}+1)^M
    \exp\left(
    -\frac{2\varepsilon^2\abs{\Sigma}}{M\,2^{2M}}
    \right).
    \label{eq:exact-normalization-mcdiarmid}
\end{align}
Any fixed tolerance $\varepsilon\in(0,1)$ would suffice to prove existence: on the event $|Z_{\rm perf}-1|\leq\varepsilon$, we have $Z_{\rm perf}\geq1-\varepsilon>0$. We take $\varepsilon=1/2$ for convenience. Since $\abs{\Sigma}=(M+1)^{2(M+2)}=\exp(\Theta(M\log M))$, the exponential decay dominates the conditioning factor, giving \eqref{eq:perfect-solution-concentration}. Outside this exceptional event,
\begin{align}
    N_{\rm perf}
    &\geq \frac{|C_{\rm FRS}|}{2^{M+1}}>0.
\end{align}
Thus a perfect codeword exists and $\operatorname{OPT}=1$.
\end{proof}

Perfect satisfiability is a conclusion about typical sampled oracles, not part of the balance promise; none of our every-oracle guarantees assumes that $\operatorname{OPT}=1$. Equipped with this, we can prove \cref{cor:typical-dqi-ratio-separation}.

\begin{proof}[Proof of \cref{cor:typical-dqi-ratio-separation}]
First fix any promised oracle $\mathcal A$; probability in the quantum argument is over the algorithm's internal randomness and measurement outcomes. For $\alpha_{\rm C}(R)<r<\alpha_{\rm Q}(R)$, applying \eqref{eq:expectation-to-search-concrete} with $\beta=r$ gives a positive constant single-run probability of score at least $r$, uniformly in $\mathcal A$. Exact verification and a constant number of independent runs raise this probability to at least $2/3$ using $\calO_{R,r}(M)$ queries. Every verified output of score at least $r$ satisfies
\begin{align}
    \operatorname{ratio}(Q) \geq \scor(Q) \geq r.
\end{align}
For the classical bound, probability is over the sampled oracle and the algorithm's independent random tape $\omega$. Put $\eta\coloneq r-\alpha_{\rm C}(R)>0$. Since $r<1$, this lies in the range of \cref{thm:frs-separation}\textup{(iii)}. On $\mathcal E_M\coloneq\{\operatorname{OPT}=1\}$, score and ratio agree. Hence
\begin{align}
    \Prb\left[
    \operatorname{ratio}(A(\omega))\geq r
    \right]
    &\leq
    \Prb[\mathcal E_M^c]
    +\Prb\left[\scor(A(\omega))\geq r\right].
\end{align}
The first term is doubly exponentially small by \cref{lem:typical-perfect-satisfiability}; the second is exponentially small by the classical theorem. Their sum remains exponentially small after reducing the positive constant in the exponent.

Finally, under the sampled-oracle distribution and Prange's independent randomness, its output score converges in probability to $\alpha_{\rm C}(R)$. Its ratio equals its score on $\mathcal E_M$, whose probability tends to one. This proves the Prange ratio convergence as well.
\end{proof}

\begin{remark}[A ratio separation without typical satisfiability]
\label{rem:ratio-without-typical-optimum}
We can also obtain a ratio separation without proving typical perfect satisfiability, provided
\begin{align}
\frac{\alpha_{\rm C}(R)}{\alpha_{\rm Q}(R)}
&<r<\alpha_{\rm Q}(R).
\label{eq:ratio-without-typical-window}
\end{align}
The DQI expected-score guarantee implies $\operatorname{OPT}\geq\alpha_{\rm Q}(R)-o(1)$ uniformly over all promised oracles: the optimum is at least the algorithm's expected score. Thus any output with ratio at least $r$ has
\begin{align}
\scor(Q)
&\geq r\operatorname{OPT}
\geq r\alpha_{\rm Q}(R)-o(1).
\end{align}
Since $r\alpha_{\rm Q}(R)>\alpha_{\rm C}(R)$, this exceeds $\alpha_{\rm C}(R)+\eta$ for all sufficiently large $M$, where $\eta\coloneq (r\alpha_{\rm Q}(R)-\alpha_{\rm C}(R))/2>0$. The classical score query lower bound therefore applies.
On the quantum side, $r<\alpha_{\rm Q}(R)$ allows verified DQI repetition to achieve score at least $r$ with bounded error. Since $\operatorname{OPT}\leq1$, this also gives approximation ratio at least $r$.
This smaller interval is nonempty exactly when $\alpha_{\rm Q}(R)^2>\alpha_{\rm C}(R)$. Typical perfect satisfiability gives the wider interval $\alpha_{\rm C}(R)<r<\alpha_{\rm Q}(R)$ for every fixed $R\in(0,1)$, as stated in \cref{cor:typical-dqi-ratio-separation}.
\end{remark}

\section{Proof of the stronger quantum guarantees on typical oracles}
\label{sec:stronger-sampled-algorithms}

This appendix proves \cref{thm:sampled-score-hierarchy} by adapting Jo's coherent summation over complete syndrome fibers~\cite{Jo2026} to our balanced folded-OPI family. We combine this construction with complete list decoding of the folded dual and establish the required normalization and score bounds. We first give the shared Fourier and decoding interfaces, then prove the stronger score thresholds and exact-search guarantee. The final subsection compares these constructions with the Yamakawa--Zhandry random-oracle model.

Throughout, fix a rational rate $R\in(0,1)$ and the balanced folded family. Retain the block alphabet $\Sigma$, coefficient space $V$, and encoder $B:V\to\Sigma^M$ from \cref{app:dqi}. The acceptance sets $\mathcal A=(A_i)_{i\in[M]}$ are sampled independently and uniformly among the $\abs{\Sigma}/2$-element subsets of $\Sigma$, where $|\Sigma|=q^h$. Write $f_i(z)=2\cdot \mathds{1}_{\{z\in A_i\}}-1$ for the local sign predicates, as in \cref{app:dqi}.

We first develop the quantum primitives used by the two constructions, then prove the stronger score thresholds and the exact search guarantee.

\subsection{Quantum primitives for the stronger constructions}
\label{sec:shared-fiber-interfaces}

We use the standard additive Fourier transforms, with syndrome map $B^T$. Recall that $\abs{\Sigma}=(M+1)^{2(M+2)}$. We use the Fourier transforms from
\cref{eq:app-block-fourier,eq:app-message-fourier} and write their phases explicitly as $(-1)^{\langle x,\sigma\rangle_V}$ for $\sigma\in V$. For density operators $\varpi$ and $\varsigma$, write
\begin{align}
 d_{\rm tr}(\varpi,\varsigma)
 &\coloneq \frac12\norm{\varpi-\varsigma}_1.
 \label{eq:exact-trace-distance-definition}
\end{align}

For exact search, we initialize each block register in the uniform superposition of all symbols accepted by that block. The following lemma shows how to prepare this state using membership queries.

\begin{lemma}[Uniform accepted-symbol state]
\label{lem:accepted-symbol-preparation}
For every fixed block $i\in[M]$, membership access to a balanced set $A_i\subseteq\Sigma$ permits the clean preparation
\begin{align}
     \ket0 \longmapsto \frac1{\sqrt{\abs{A_i}}}\sum_{z\in A_i}\ket z
     \label{eq:app-local-prepared-states}
\end{align}
with all auxiliary registers returned to $\ket0$. The preparation uses three queries to $\mathcal O_{A_i}$ and polynomially many additional gates in $\log \abs{\Sigma}$. It can be placed under an additional coherent control with only a constant-factor query overhead. Throughout, the block index $i$ remains classically fixed.
\end{lemma}
\begin{proof}
Fix a block $i$. Use a symbol register $\mathrm Z$, a membership qubit $\mathrm B$, and a coin qubit $\mathrm C$. Let $U_i$ prepare a uniform symbol, compute its membership bit, and prepare the coin in $\ket+$. Using one membership query, this gives
\begin{align}
    \ket{\psi}
    &\coloneq
    U_i\ket0_{\mathrm Z}\ket0_{\mathrm B}\ket0_{\mathrm C}
    =
    \frac1{\sqrt{2\abs{\Sigma}}}
    \sum_{z\in\Sigma}
    \ket z_{\mathrm Z}\ket{O_{A_i}(z)}_{\mathrm B}
    \bigl(\ket0_{\mathrm C}+\ket1_{\mathrm C}\bigr).
\end{align}

Call a basis state good when both auxiliary qubits contain one. Since $A_i$ contains exactly half the symbols, the good component has squared norm $1/4$. Its normalized state is
\begin{align}
    \ket G \coloneq \frac1{\sqrt{|A_i|}} \sum_{z\in A_i} \ket z_{\mathrm Z}\ket1_{\mathrm B}\ket1_{\mathrm C}.
\end{align}
The projector onto this good subspace is
\begin{align}
\Pi_{\rm good}
&\coloneq
I_{\mathrm Z}\otimes
(\ket1\bra1)_{\mathrm B}\otimes
(\ket1\bra1)_{\mathrm C},
\end{align}
so $\Pi_{\rm good}\ket{\psi}=\frac12\ket G$. The extra coin qubit makes the initial good probability exactly $1/4$, allowing one Grover iteration to amplify it to one.

First apply $S_{\rm good}\coloneq I-2\Pi_{\rm good}$. This flips the sign precisely when both auxiliary qubits contain one, so it is a controlled-$Z$ gate on $\mathrm B$ and $\mathrm C$ and requires no oracle query. Then reflect about $\ket{\psi}$ using
\begin{align}
    R_\psi &\coloneq 2\ket{\psi}\bra{\psi}-I = U_i\left( 2(\ket0\bra0)_{\mathrm Z} \otimes(\ket0\bra0)_{\mathrm B} \otimes(\ket0\bra0)_{\mathrm C} -I \right)U_i^{-1}.
\end{align}
Because the good component has squared norm $1/4$,
\begin{align}
    \bra{\psi}S_{\rm good}\ket{\psi} =1-2\cdot\frac14=\frac12.
\end{align}
Consequently, these two reflections give
\begin{align}
    R_\psi S_{\rm good}\ket{\psi} &=\ket{\psi}-S_{\rm good}\ket{\psi} =2\Pi_{\rm good}\ket{\psi} =\ket G.
\end{align}
This is one exact amplitude-amplification iteration~\cite{BrassardHoyerMoscaTapp2002}; no measurement or postselection is required. In $\ket G$, both auxiliary qubits are fixed at one. Applying an $X$ gate to each therefore gives
\begin{align}
    \ket G
    &\longmapsto
    \frac1{\sqrt{|A_i|}}
    \sum_{z\in A_i}
    \ket z_{\mathrm Z}\ket0_{\mathrm B}\ket0_{\mathrm C}.
\end{align} 
The initial application of $U_i$ uses one membership query, and the reflection $R_\psi$ uses one query in each of $U_i^{-1}$ and $U_i$. The total is therefore three queries. All remaining gates have complexity polynomial in $\log \abs{\Sigma}$.

To place the preparation under an additional coherent control, implement each controlled oracle call by computing the membership bit into a clean auxiliary qubit, using it conditionally, and uncomputing it. This uses two membership queries per controlled call; the remaining gates can be controlled directly. The block index $i$ remains fixed throughout. Applying the preparation separately to all $M$ blocks therefore uses $\calO(M)$ queries.
\end{proof}

DQI uses unique decoding to erase an error $y$ once its syndrome $B^Ty$ identifies it uniquely. Beyond the unique-decoding radius, several retained labels may have the same syndrome. The next lemma replaces unique erasure by coherent aggregation. A complete list decoder enumerates the retained labels in each syndrome fiber. Coherently recording the position of the input label within this list allows the label to be erased reversibly, leaving its syndrome and list index. Applying Hadamards to the index register and accepting the all-zero outcome makes the amplitudes of all labels in the same fiber add on the same syndrome state. The polynomial fiber-size bound keeps this procedure efficient, while the success flag rejects labels outside the enumerated set.

\begin{lemma}[Coherent syndrome-fiber summation]
\label{lem:coherent-fiber-primitive}
For every admissible $M$, let $\mathcal E\subseteq\Sigma^M$ be a set. Suppose there is a single deterministic algorithm that, for every such $M$ and every syndrome $\sigma\in V$, lists exactly
\begin{align}
    \mathcal L_\sigma^{\mathcal E} &\coloneq \{y\in\mathcal E:B^Ty=\sigma\},
    \label{eq:generic-complete-fib}
\end{align}
in canonical order and in time polynomial in $M$, uniformly over $\sigma$. Suppose also that every list has size at most a bound $L_M$, which is computable in time polynomial in $M$ and satisfies $1\leq L_M\leq\poly(M)$. Let $L_{\rm pad}\coloneq 2^{\ceil{\log_2L_M}}$. This is the smallest power of two at least $L_M$, so every padded list can be indexed by $\ceil{\log_2L_M}$ qubits.
When $\mathcal E$ consists of the words with block weights in $[j_-,j_+]$, the list $\mathcal L_\sigma^{\mathcal E}$ is precisely $\mathcal L_{\sigma,[j_-,j_+]}$ from \eqref{eq:frs-complete-fib}.

Then there is a polynomial-size, oracle-free circuit with a measured success flag that performs coherent summation over the syndrome fibers. On any normalized input $\sum_{y\in\Sigma^M}a_y\ket y$, its successful branch has the unnormalized state
\begin{align}
 \frac1{\sqrt{L_{\rm pad}}}
 \sum_{\sigma\in V}
 \left(\sum_{y\in\mathcal L_\sigma^{\mathcal E}}a_y\right)\ket\sigma.
 \label{eq:generic-fiber-success-state}
\end{align}
The success flag is observed with probability
\begin{align}
 p_{\rm succ}
 &=\frac1{L_{\rm pad}}
 \sum_{\sigma\in V}
 \left|
 \sum_{y\in\mathcal L_\sigma^{\mathcal E}}a_y
 \right|^2.
 \label{eq:generic-fiber-success-probability}
\end{align}
The input need not be supported on $\mathcal E$. In particular, labels outside $\mathcal E$ are rejected by the measured flag rather than silently removed or renormalized.
\end{lemma}

\begin{proof}
Starting from $\sum_y a_y\ket y$, compute the syndrome $\sigma=B^Ty$ coherently into a fresh register:
\begin{align}
     \sum_{y\in\Sigma^M}a_y\ket y\ket0
     &\longmapsto
     \sum_{y\in\Sigma^M}a_y\ket y\ket{B^Ty}.
     \label{eq:generic-fiber-compute-syndrome}
\end{align}
This leaves the original label register unchanged. From $\sigma$, reversibly compute the complete list in canonical order and pad it to $L_{\rm pad}$ slots, marking the added slots as invalid. Coherently compare the original label register with the valid list entries. If $y$ occurs in slot $j$, record $(\mathrm{Idx},b_{\rm match})=(j,1)$; otherwise record $(0,0)$. The match flag distinguishes the latter case from a genuine match in slot zero. Uncompute the comparison workspace while retaining $\mathrm{Idx}$ and $b_{\rm match}$.

When $b_{\rm match}=1$, subtract the list entry in slot $j$ from the label register, and then uncompute the list and its workspace. For the $j$th label $y_{\sigma,j}$ in the fiber this realizes
\begin{align}
     \ket{y_{\sigma,j}}\ket\sigma
     \ket0_{\mathrm{Idx}}\ket0_{b_{\rm match}}
     &\longmapsto
     \ket0\ket\sigma
     \ket j_{\mathrm{Idx}}\ket1_{b_{\rm match}}.
     \label{eq:generic-fiber-indexing}
\end{align}
The index retains which list element was erased, so the map remains reversible. For a nonmatching input, $b_{\rm match}=0$ and the label register is left unchanged. Such inputs are rejected by the final success measurement.

At this point, the matched part of the state is
\begin{align}
    \sum_{\sigma\in V}
    \sum_{j\in[|\mathcal L_\sigma^{\mathcal E}|]}
    a_{y_{\sigma,j}}
    \ket\sigma\ket j_{\mathrm{Idx}}\ket1_{b_{\rm match}}.
\end{align}
We now erase the list index in a way that makes the amplitudes within each syndrome fiber add coherently. Let $k=\log_2L_{\rm pad}$. Applying $H^{\otimes k}$ and accepting the all-zero index outcome projects the index register onto its uniform superposition, since
\begin{align}
    \langle0^k|H^{\otimes k}|j\rangle &=\frac1{\sqrt{L_{\rm pad}}} \qquad\text{for every }j\in[L_{\rm pad}].
\end{align}
Accepting only when the match flag is also one therefore leaves the unnormalized syndrome-register state
\begin{align}
    \frac1{\sqrt{L_{\rm pad}}} \sum_{\sigma\in V} \left(\sum_{y\in\mathcal L_\sigma^{\mathcal E}}a_y \right)\ket\sigma.
\end{align}
Thus the amplitudes of all retained labels with the same syndrome have been added coherently, proving \eqref{eq:generic-fiber-success-state}. Taking its squared norm gives \eqref{eq:generic-fiber-success-probability}. All list computation and uncomputation have polynomial overhead and make no oracle queries. This is the complete-fiber mechanism of Jo~\cite{Jo2026}, stated with explicit rejection of labels outside the enumerated set and with a charged, observed success flag.
\end{proof}

After applying $\operatorname{QFT}_V^{-1}$ to the syndrome register in \eqref{eq:generic-fiber-success-state}, the successful branch becomes
\begin{align}
    \frac1{\sqrt{L_{\rm pad}|V|}} \sum_{x\in V} \left( \sum_{y\in\mathcal E} a_y(-1)^{\langle x,B^Ty\rangle_V} \right)\ket x,
    \label{eq:generic-fiber-message-amplitude}
\end{align}
because the fibers $\mathcal L_\sigma^{\mathcal E}$ partition $\mathcal E$ and each $y\in\mathcal E$ belongs to the unique fiber with $\sigma=B^Ty$.

For arbitrary input amplitudes, the fiber-summation circuit need not succeed with constant probability, since amplitudes within a fiber may cancel. For each state used below, we separately prove a constant lower bound on its fiber-success probability.

Both stronger quantum constructions use complete list decoding of the folded dual. Given a syndrome $\sigma$ and a block-weight interval, the decoder enumerates every Fourier label $y$ in that interval satisfying $B^Ty=\sigma$. The following lemma records the imported folded Reed--Solomon decoding guarantee in the form needed here.

\begin{lemma}[Complete folded-dual fiber decoder]
\label{lem:frs-complete-list-decoder}
Fix a rational constant $\tau\in(0,R)$ and put
\begin{align}
    L_\tau &\coloneq \ceil{\frac{\tau}{R-\tau}}+1.
\end{align}
For all sufficiently large admissible $M$, there is a deterministic algorithm that, given a syndrome $\sigma\in V$ and integers
\begin{align}
    0\leq j_-\leq j_+\leq\floor{\tau M},
\end{align}
outputs exactly
\begin{align}
    \mathcal L_{\sigma,[j_-,j_+]} &\coloneq \left\{ y\in\Sigma^M: B^Ty=\sigma,\quad j_-\leq\operatorname{wt}_{\rm blk}(y)\leq j_+ \right\}.
\label{eq:frs-complete-fib}
\end{align}
The algorithm runs in time $M^{O_{R-\tau}(1)}$, and every output list has size at most
\begin{align}
    \left|\mathcal L_{\sigma,[j_-,j_+]}\right| &\leq L_\tau=O_{R-\tau}(1).
\end{align}
The output can be deduplicated and sorted canonically, and the algorithm can be padded to a known running time independent of $\sigma$.
\end{lemma}

\begin{proof}
Let $D\coloneq\operatorname{diag}(\alpha_t^{-1})_{t\in[N]}$. By the Reed--Solomon duality identities in
\cref{eq:rs-dual-grs,eq:folded-rs-dual-grs,eq:dqi-full-length-dual-multipliers},
$D(C_{\rm FRS}^{\perp})$ is the standard $h$-folded Reed--Solomon code of scalar dimension $N-K$ and rate $R_\perp=\frac{N-K}{N}=1-R+o(1)$. The map $D$ preserves block Hamming distance.

Set $s\coloneq L_\tau$. Guruswami's linear-algebraic decoder \cite[Theorem~7 and Lemma~6]{Guruswami2011Linear} produces an affine space of dimension at most $s-1$ containing every codeword within relative block distance
\begin{align}
    \delta_{M,s} &\coloneq \frac{s}{s+1} \left( 1-\frac{hR_{\perp}}{h-s+1} \right).
\end{align}
Since $s$ is fixed,
\begin{align}
    \delta_{M,s} &\rightarrow \frac{s}{s+1}R >\tau,
\end{align}
where the strict inequality follows from $s>\tau/(R-\tau)$. Hence, for all sufficiently large $M$, the decoder contains every codeword within block distance $\floor{\tau M}$ of its input. Its affine candidate space has at most $q^{s-1}=M^{\calO_{R-\tau}(1)}$ elements, so all candidates can be enumerated and checked exactly.

The list-size bound of Chen and Zhang \cite[Theorem~2.20]{ChenZhang2025Relaxed} shows that at most $s=L_\tau$ candidates survive this verification. Its hypotheses hold for our parameters because $N=q-1=hM$, the expanded evaluation points are distinct, $q>M$, and $s\leq h$ for all sufficiently large $M$.

It remains to convert this received-word decoder into the required syndrome-fiber decoder. Since $B^T$ is surjective, Gaussian elimination computes a canonical word $e_\sigma\in\Sigma^M$ satisfying $B^Te_\sigma=\sigma$. Then
\begin{align}
    B^Ty=\sigma \quad\Longleftrightarrow\quad e_\sigma-y\in C_{\rm FRS}^{\perp},
\end{align}
and
\begin{align}
    \operatorname{wt}_{\rm blk}(y) &= d_{\rm blk}(e_\sigma,e_\sigma-y).
\end{align}
Run the standard folded Reed--Solomon decoder on $De_\sigma$. For every returned codeword $c'\in D(C_{\rm FRS}^{\perp})$, set
\begin{align}
    d\coloneq D^{-1}c'\in C_{\rm FRS}^{\perp}, \qquad y\coloneq e_\sigma-d.
\end{align}
Since
\begin{align}
    d_{\rm blk}(De_\sigma,c') &= d_{\rm blk}(e_\sigma,d) = \operatorname{wt}_{\rm blk}(y),
\end{align}
retain exactly those $y$ whose block weights lie in $[j_-,j_+]$. This produces precisely \eqref{eq:frs-complete-fib}, with the same list-size and running-time bounds. Canonical sorting and fixed-time padding add only polynomial overhead.
\end{proof}

\subsection{Part \textup{(i)}: coherent-fiber thresholds}
\label{sec:proof-sampled-coherent-fib}

Recall from \cref{app:dqi-score-proof} that the index $j$ denotes the block weight of a normalized folded-DQI layer; see \eqref{eq:app-polynomial-dqi-layer}. Fix a rational relative layer cutoff $\mu$ with $0<\mu\leq1/2$ and $\mu<R$. Put
\begin{align}
 \ell_M&\coloneq\floor{\mu M}, &
 \Delta&\coloneq2^{\floor{(\log_2M)/2}}.
 \label{eq:cl-cutoff-width}
\end{align}
Here $\ell_M$ is the layer cutoff, while $\Delta$ is the number of consecutive layers used immediately below it. For sufficiently large $M$, one has $\ell_M-\Delta+1\geq0$, while $\Delta\to\infty$ and $\Delta=o(M)$. Thus the number of selected layers grows, but every selected index satisfies $j/M=\mu+o(1)$.

We use the normalized layer polynomials $P_j$ from \eqref{eq:app-polynomial-dqi-layer}; their algebraic definition applies for every $0\leq j\leq M$. As in \eqref{eq:app-filter-from-layers}, assign equal amplitudes to the $\Delta$ retained layers. The resulting filter evaluates to
\begin{align}
    P(f(x)) &\coloneq \frac1{\sqrt\Delta} \sum_{j=\ell_M-\Delta+1}^{\ell_M}P_j(f(x)).
    \label{eq:cl-filter}
\end{align}
Unlike in the DQI setting of \cref{app:dqi}, different retained labels may have the same syndrome, and their amplitudes can interfere. We therefore keep the normalization explicit in the corresponding quantum state below. Whenever the amplitude vector is nonzero, define
\begin{align}
    \ket{P(f)} &\coloneq \frac1{\sqrt{\sum_{x\in V}P(f(x))^2}} \sum_{x\in V}P(f(x))\ket x.
\label{eq:cl-filter-state}
\end{align}
Let $X\in V$ denote the coefficient vector obtained by measuring this state in the computational basis and denote by $Q_X$ the corresponding polynomial. 

\begin{proposition}[Coherent-fiber thresholds at a fixed cutoff]
\label{prop:coherent-fiber-fixed-radius}
For the balanced sampled-oracle folded Reed--Solomon family, fix a rational rate $R\in(0,1)$, a rational relative cutoff $\mu$ satisfying
\begin{align}
    0<\mu\leq\frac12, \qquad \mu<R,
\end{align}
and a constant $c>0$. Use $\ell_M=\floor{\mu M}$ and the filter $P$ from \eqref{eq:cl-filter}. There is a quantum algorithm using $\calO(M)$ queries per attempt and $M^{\calO_{R,\mu,c}(1)}$ gates. For all sufficiently large admissible $M$ and every oracle outside an exceptional set whose probability under the balanced sampled-oracle distribution is at most
\begin{align}
    \exp\{-\exp(\Omega_{R,\mu}(M\log M))\},
    \label{eq:cl-bad-oracle-probability}
\end{align}
one attempt observes the fiber-success flag with probability $\Omega_{R-\mu}(1)$. Conditioned on successful fiber summation, the output polynomial has a distribution within total-variation distance $M^{-c}$ of the ideal distribution obtained by measuring $\ket{P(f)}$ and returning the polynomial $Q_X$ with the measured coefficients. Moreover, for each such fixed oracle,
\begin{align}
    \E[\scor(Q_X)] &=\alpha_{\rm sc}(\mu)+o(1) =\frac12+\sqrt{\mu(1-\mu)}+o(1).
    \label{eq:cl-score}
\end{align}
Here the expectation is over measurement of the ideal filter state $\ket{P(f)}$, and the $o(1)$ is uniform over oracles outside the exceptional set. Every fixed rational threshold $0<\beta<\alpha_{\rm sc}(\mu)$ is therefore attained with bounded error using $\calO_{R,\mu,\beta}(M)$ queries outside the same exceptional set. If also $\beta>\alpha_{\rm C}(R)$, \cref{thm:frs-separation}\textup{(iii)} supplies the corresponding distributional classical query lower bound.
\end{proposition}

\begin{proof}
We analyze the coherent summation, its normalization, and verification in turn.

\emph{Complete decoding and coherent summation.}
Define the retained set
\begin{align}
    \mathcal E &\coloneq \left\{y\in\Sigma^M: \ell_M-\Delta+1\leq\operatorname{wt}_{\rm blk}(y)\leq\ell_M \right\}.
    \label{eq:cl-retained-set}
\end{align}
For the folded Reed--Solomon dual, the complete decoder from \cref{lem:frs-complete-list-decoder} returns all labels in $\mathcal E$ with a given syndrome $\sigma$. Write this list as $\mathcal L_\sigma\coloneq \mathcal L_{\sigma,[\ell_M-\Delta+1,\ell_M]}$ in the notation of \eqref{eq:frs-complete-fib}. Every such list has at most $L\coloneq L_{\mu}=\calO_{R-\mu}(1)$ entries. The positive slack $\mu<R$ makes the list bound constant in $M$. The lists are complete and canonically ordered. As above, use the common power-of-two padding length $L_{\rm pad} \coloneq 2^{\ceil{\log_2 L}}$.

Use only the pre-syndrome layer and Fourier-label preparation from \eqref{eq:app-pre-syndrome-dqi-state}, with equal weights on the $\Delta$ selected layers. Compute the syndrome into a second register, but retain the label register instead of applying unique decoding. This produces
\begin{align}
    \ket{\Psi_0} &\coloneq \sum_{y\in\Sigma^M}a_y\ket y\ket{B^Ty}.
    \label{eq:cl-labeled-state}
\end{align}
Here the first register holds the error label $y$, the second holds its syndrome, and $a_y$ is the amplitude of this joint basis state. Explicitly,
\begin{align}
    a_y &\coloneq
    \begin{cases}
        \displaystyle \frac{\widetilde\phi(y)}{\sqrt{\Delta\binom M{\operatorname{wt}_{\rm blk}(y)}}}, &\ell_M-\Delta+1\leq\operatorname{wt}_{\rm blk}(y)\leq\ell_M,\\[1ex] 
        0,&\text{otherwise},
    \end{cases}
    \label{eq:cl-ay}
\end{align}
where, as in \cref{app:dqi}, $\widetilde\phi(y)=\prod_{i:y_i\ne0}\widetilde\phi_i(y_i)$. The factor $1/\sqrt{\Delta}$ weights the selected layers equally. The binomial factor comes from choosing uniformly among all supports of size $\operatorname{wt}_{\rm blk}(y)$. The factor $\widetilde\phi(y)$ then supplies the Fourier amplitudes of the nonzero block values on that support.

The superposition $\sum_y a_y\ket y$ is normalized because each layer has norm one, different layers are orthogonal, and their squared weights sum to one. Computing the syndrome into a second register preserves this normalization. Even when distinct labels $y$ and $y'$ have the same syndrome $\sigma$, the joint states $\ket y\ket\sigma$ and $\ket{y'}\ket\sigma$ remain orthogonal.

The difficulty arises when we want to remove the error register. In DQI, the syndrome uniquely determines the retained error label, so a decoder can recover and erase it. Here several retained labels may share a syndrome. Instead, the complete fiber-summation procedure combines their amplitudes into that syndrome's amplitude and removes the error register on its reported success branch. The local sign-state routine is applied separately for each fixed block $i$, using one query to $\mathcal O_{A_i}$ per block and hence $M$ queries in total. 

We now identify the amplitudes produced by summing over each syndrome fiber. By Fourier inversion, each label $y$ contributes $a_y(-1)^{\langle x,B^Ty\rangle_V}$ to $P(f(x))$. Labels with the same syndrome $\sigma=B^Ty$ have the same phase factor, so collecting their coefficients gives
\begin{align}
    P(f(x)) &= \sum_{\sigma\in V} \left(\sum_{y\in\mathcal L_\sigma}a_y\right) (-1)^{\langle x,\sigma\rangle_V}.
    \label{eq:cl-fiber-amplitudes}
\end{align}
Thus the sum $\sum_{y\in\mathcal L_\sigma}a_y$ is precisely the coefficient needed at syndrome $\sigma$ before the final inverse Fourier transform.

We must also determine the normalization of these summed amplitudes. Their squared norm is
\begin{align}
    Z &\coloneq \sum_{\sigma\in V} \left|\sum_{y\in\mathcal L_\sigma}a_y\right|^2 = \E_{x\in V}[P(f(x))^2],
    \label{eq:cl-fourier-identity}
\end{align}
where $\E_{x\in V}$ denotes the uniform average over $x\in V$. To obtain the equality, square the preceding expression for $P(f(x))$ and average over $x$. Cross terms between distinct syndromes average to zero, leaving exactly the sum of squared amplitudes on the left. Within each fiber, amplitudes are added before their squared magnitude is taken. They may reinforce or cancel one another, so $Z$ need not equal one.

Starting from the joint state $\ket{\Psi_0}$ in \eqref{eq:cl-labeled-state}, apply the coherent-fiber procedure of \cref{lem:coherent-fiber-primitive} with retained set $\mathcal E$, beginning after its syndrome-computation step. The procedure reports success or failure, and its component corresponding to success is
\begin{align}
    \frac1{\sqrt{L_{\rm pad}}} \sum_{\sigma\in V} \left(\sum_{y\in\mathcal L_\sigma}a_y\right)\ket\sigma.
    \label{eq:cl-fiber-sum}
\end{align}
The squared norm of this unnormalized component is the probability that fiber summation reports success in one attempt. By the definition of $Z$, this probability is
\begin{align}
    p_{\rm fib} &=\frac{Z}{L_{\rm pad}}.
    \label{eq:cl-fiber-success}
\end{align}

If $Z=0$, success never occurs. Otherwise, conditioning on success normalizes the displayed vector by dividing it by $\sqrt{Z/L_{\rm pad}}$. The padding factor cancels, leaving the summed syndrome amplitudes divided by $\sqrt Z$. Applying the inverse Fourier transform then gives amplitude
\begin{align}
    \frac1{\sqrt{|V|Z}} \sum_{\sigma\in V} \left(\sum_{y\in\mathcal L_\sigma}a_y\right) (-1)^{\langle x,\sigma\rangle_V} &=\frac{P(f(x))}{\sqrt{|V|Z}}
\end{align}
at the coefficient vector $x$ of $Q_x$, by \eqref{eq:cl-fiber-amplitudes}. Thus, conditioned on success, the circuit prepares exactly the filter state $\ket{P(f)}$ from \eqref{eq:cl-filter-state}.

We next show that $Z$ is close to one on \emph{typical sampled oracles}. Since $L_{\rm pad}$ is bounded independently of $M$, this gives a success probability bounded below by a positive constant. Controlling $Z$ also controls the normalization needed to calculate the expected score after success.

\emph{Concentration of normalization and score.}
We follow the concentration argument from \cref{lem:typical-perfect-satisfiability}, now applied to $Z$ and the score-weighted quantity
\begin{align}
    N &\coloneq \E_{x\in V}[P(f(x))^2\scor(Q_x)].
    \label{eq:cl-NH}
\end{align}
For a fixed balanced oracle with $Z>0$, the ideal filtered distribution assigns probability $P(f(x))^2/(|V|Z)$ to $x$. Its expected score is therefore $N/Z$. We will control this ratio by showing that numerator and denominator are simultaneously close to their auxiliary expectations.

Use the auxiliary model in which all membership bits are independent fair bits, and write $\E_{\rm aux}$ for expectation over this oracle. Evaluate $Z$ using the message-average expression in \eqref{eq:cl-fourier-identity}. This expression and the definition of $N$ remain well-defined even when the auxiliary oracle is not balanced.

We first compute the auxiliary expectations. Fix $x\in V$. The signs $f_i(B_ix)$ are independent and uniform in $\{\pm1\}$. For two subsets $S,T\subseteq[M]$, the factors belonging to $S\cap T$ cancel because $f_i(B_ix)^2=1$. If $S\ne T$, at least one independent mean-zero sign remains. Consequently,
\begin{align}
    \E_{\rm aux}\left[\left(\prod_{i\in S}f_i(B_ix)\right) \left(\prod_{i\in T}f_i(B_ix)\right) \right] &=\mathds{1}_{\{S=T\}}.
\end{align}
Each $P_j(f(x))$ sums these products over all supports of size $j$, with normalization $1/\sqrt{\binom Mj}$. Different layers therefore have zero expected product. Within the same layer, the $\binom Mj$ matching supports contribute one each, and the normalization cancels this factor. Hence
\begin{align}
    \E_{\rm aux}\left[P_j(f(x))P_k(f(x))\right] &=\mathds{1}_{\{j=k\}}.
    \label{eq:cl-layer-orthogonality}
\end{align}

Since $\scor(Q_x)=1/2+f(x)/(2M)$, computing the score-weighted expectation also requires multiplication by $f(x)=\sum_i f_i(B_ix)$. Multiplying a product supported on $j$ blocks by $f_i(B_ix)$ adds block $i$ if it is absent and removes it if it is present. Each resulting support of size $j+1$ arises in $j+1$ ways, while each support of size $j-1$ arises in $M-j+1$ ways. Accounting for the layer normalizations gives
\begin{align}
    f(x)P_j(f(x)) &=\sqrt{(j+1)(M-j)}\,P_{j+1}(f(x)) +\sqrt{j(M-j+1)}\,P_{j-1}(f(x)),
    \label{eq:cl-three-term-recurrence}
\end{align}
with the boundary convention $P_{-1}=P_{M+1}=0$.

For the equal weights on the $\Delta$ selected layers, orthogonality and this recurrence yield
\begin{align}
    \E_{\rm aux}[Z] &=1, \label{eq:cl-ED}\\
    \E_{\rm aux}[N] &=\frac12+\frac1{M\Delta} \sum_{j=\ell_M-\Delta+1}^{\ell_M-1} \sqrt{(j+1)(M-j)}.
    \label{eq:cl-EN}
\end{align}

To turn these oracle averages into guarantees for individual instances, we bound how much changing one membership bit can change $Z$ or $N$, and then apply McDiarmid's inequality. Changing the answer at $(i,z)$ can affect only messages $x\in V$ whose output in block $i$ equals $z$, that is, $B_ix=z$. For sufficiently large $M$, we have $K\geq h$. Because the $h$ evaluation points in each block are distinct, interpolation realizes every element of $\Sigma$ in that block. Hence $B_i:V\to\Sigma$ is surjective. Since this map is linear, exactly $|V|/|\Sigma|$ messages have each block value. Thus only the fraction $1/|\Sigma|$ of messages can be affected.

We also need a uniform bound on the contribution of each affected message. Each $P_j(f(x))$ is a normalized sum of $\binom Mj$ products of signs, so $|P_j(f(x))|\leq\sqrt{\binom Mj}$. Put $P_{\max}\coloneq2^{M/2}$. Cauchy--Schwarz gives
\begin{align}
    |P(f(x))| &\leq \frac1{\sqrt\Delta} \sum_{j=\ell_M-\Delta+1}^{\ell_M} \sqrt{\binom Mj} \leq P_{\max}.
    \label{eq:cl-filter-bound}
\end{align}
For the last inequality, use $\sum_{j=0}^M\binom Mj=2^M$.

Because scores lie in $[0,1]$, both $P(f(x))^2$ and $P(f(x))^2\scor(Q_x)$ lie in $[0,P_{\max}^2]$. Changing one membership bit therefore changes either average $Z$ or $N$ by at most $P_{\max}^2/|\Sigma|$. There are $M|\Sigma|$ independent membership bits in the auxiliary model, so the sum of squared change bounds is $MP_{\max}^4/|\Sigma|$.

Apply McDiarmid's inequality separately to $Z$ and $N$, then take a union bound. As in the perfect-solution argument, conditioning on exact balance costs at most the factor $(|\Sigma|+1)^M$ from \eqref{eq:fixed-density-conditioning-probability}. Under the balanced sampled-oracle distribution, this gives, for every $\delta>0$,
\begin{align}
    \Prb_{\mathcal A}\left[ \substack{|Z-1|>\delta\ \text{or}\\ |N-\E_{\rm aux}[N]|>\delta} \right]
    &\leq 4(|\Sigma|+1)^M \exp\left( -\frac{2\delta^2|\Sigma|}{MP_{\max}^4} \right).
    \label{eq:cl-mcdiarmid}
\end{align}
Recall that
\begin{align}
    |\Sigma| &=(M+1)^{2(M+2)}.
    \label{eq:cl-U-size}
\end{align}
Taking $\delta=M^{-2}$ and using $P_{\max}^4=2^{2M}$, the magnitude of the negative exponent grows as $\exp(\Theta(M\log M))$, whereas the logarithm of the prefactor is only $\calO(M^2\log M)$. This yields the exceptional-set bound \eqref{eq:cl-bad-oracle-probability}.

Outside this exceptional set, both $|Z-1|\leq\delta$ and $|N-\E_{\rm aux}[N]|\leq\delta$ hold. Moreover, $0\leq N\leq Z$ for every oracle, and $\E_{\rm aux}[Z]=1$, so $0\leq\E_{\rm aux}[N]\leq1$. For sufficiently large $M$, we therefore have
\begin{align}
    Z &\geq1-\delta\geq\frac12,\\ \left|\frac NZ-\E_{\rm aux}[N]\right| &\leq \frac{|N-\E_{\rm aux}[N]|+\E_{\rm aux}[N]|1-Z|}{Z} \leq\frac{2\delta}{1-\delta} \leq4\delta.
\end{align}
Thus, for each fixed balanced oracle outside the exceptional set, the expected score of the ideal filtered output satisfies
\begin{align}
    \E[\scor(Q_X)] &=\frac NZ =\E_{\rm aux}[N]+\calO(M^{-2}) =\alpha_{\rm sc}(\mu)+o(1).
    \label{eq:cl-score-proof}
\end{align}
For the last equality, $j/M\to\mu$ uniformly throughout the selected window, because $\ell_M/M\to\mu$ and $\Delta=o(M)$. Hence $\sqrt{(j+1)(M-j)}/M\to\sqrt{\mu(1-\mu)}$ uniformly over the summands in \eqref{eq:cl-EN}. There are $\Delta-1$ such summands, and $(\Delta-1)/\Delta\to1$ because $\Delta\to\infty$.

The same event gives $p_{\rm fib}=Z/L_{\rm pad}\geq1/(2L_{\rm pad})$. Since $L_{\rm pad}$ is bounded independently of $M$, one attempt succeeds with probability bounded below by a positive constant.

\emph{Implementation, conditioning, and verification.}
Fix an oracle in the concentration event above and put $\xi\coloneq M^{-c}$. Let $\ket\Phi$ be the ideal normalized state immediately before measuring the fiber-success flag, and let $\ket{\widetilde\Phi}$ be its implementation. Compile the preparation so that
\begin{align}
    \norm{\ket{\widetilde\Phi}-\ket\Phi} &\leq \varepsilon_{\rm imp} \leq \min\left\{ \frac1{8L_{\rm pad}}, \frac{\xi}{4\sqrt{2L_{\rm pad}}} \right\}.
    \label{eq:cl-precision-choice}
\end{align}
Let $\Pi_{\rm fib}$ project onto the measured outcome $(\mathrm{Idx},b_{\rm match})=(0,1)$ of \cref{lem:coherent-fiber-primitive}. The ideal and implemented success probabilities are
\begin{align}
    p_{\rm fib}&=\norm{\Pi_{\rm fib}\ket\Phi}^2, &\widetilde p_{\rm fib} &\coloneq \norm{\Pi_{\rm fib}\ket{\widetilde\Phi}}^2.
\end{align}
Projection cannot increase the vector error. Since both unprojected states have norm one, the difference of these squared norms is at most $2\varepsilon_{\rm imp}$. Thus the first precision requirement preserves a constant success probability:
\begin{align}
     \widetilde p_{\rm fib}
     &\geq p_{\rm fib}-2\varepsilon_{\rm imp}
     \geq\frac1{2L_{\rm pad}}-2\varepsilon_{\rm imp}
     \geq\frac1{4L_{\rm pad}}.
     \label{eq:cl-implemented-fiber-success}
\end{align}
Conditioning on success normalizes each projected vector. This can amplify its error, so we compare the normalized branches explicitly:
\begin{align}
     \norm{
     \frac{\Pi_{\rm fib}\ket{\widetilde\Phi}}
       {\sqrt{\widetilde p_{\rm fib}}}
     -\frac{\Pi_{\rm fib}\ket\Phi}{\sqrt{p_{\rm fib}}}
     }
     &\leq\frac{2\varepsilon_{\rm imp}}{\sqrt{p_{\rm fib}}} \leq 2\varepsilon_{\rm imp}\sqrt{2L_{\rm pad}}\leq\xi/2.
     \label{eq:cl-conditioning-stability}
\end{align}
The first inequality follows by adding and subtracting the implemented projected vector divided by $\sqrt{p_{\rm fib}}$, then using the reverse triangle inequality for the two projected norms. This explains the second requirement in \eqref{eq:cl-precision-choice}. Implement the final inverse Fourier transform to operator-norm error at most $\xi/2$. The triangle inequality then bounds the total conditional state-vector error by $\xi$. Computational-basis measurement has total-variation error at most this amount, giving the claimed $M^{-c}$ distance from the law of $Q_X$. Since $L_{\rm pad}=O_{R-\mu}(1)$, all requested precisions are inverse polynomial in $M$. They require only polynomial gate overhead and no extra query factors.

For fixed rational $0<\beta<\alpha_{\rm sc}(\mu)$, put $\gamma\coloneq\alpha_{\rm sc}(\mu)-\beta$ and
$p_0\coloneq\gamma/[2(1-\beta)]$. For sufficiently large $M$,
\begin{align}
     \E[\scor(Q_X)]
     &\geq\alpha_{\rm sc}(\mu)-\gamma/4, &
     M^{-c}&\leq\frac{\gamma}{4(1-\beta)}.
     \label{eq:cl-threshold-slack}
\end{align}
The expectation-to-threshold bound first gives an ideal conditional success probability
\begin{align}
 \Prb[\scor(Q_X)\geq\beta]
 &\geq\frac{\E[\scor(Q_X)]-\beta}
         {1-\beta}
 \geq\frac{3\gamma}{4(1-\beta)}.
\end{align}
Write $Q$ for the candidate polynomial produced by the implemented circuit, before score verification. Conditioned on fiber success, its distribution is within total-variation distance $M^{-c}$ of the ideal law of $Q_X$. The probability of reaching threshold $\beta$ can therefore decrease by at most $M^{-c}$. Hence
\begin{align}
    \Prb[\scor(Q)\geq\beta\mid\text{fiber success}]
    &\geq\frac{3\gamma}{4(1-\beta)}-M^{-c}
    \geq p_0.
    \label{eq:cl-threshold-probability}
\end{align}
This is conditional on observing the flag. Multiplying by \eqref{eq:cl-implemented-fiber-success} shows that one full attempt produces a score-$\beta$ candidate with probability at least $p_0/(4L_{\rm pad})$. Verify each such candidate by querying all $M$ blocks, returning it only if its score is at least $\beta$. Charging these $M$ queries to every attempt, including failed fiber projections, gives $\calO(M)$ queries per attempt. With
\begin{align}
     n&\coloneq \left\lceil4L_{\rm pad}p_0^{-1}\log3\right\rceil,
     &\Prb[\text{all attempts fail}]
     &\leq\left(1-\frac{p_0}{4L_{\rm pad}}\right)^n
     \leq\frac13,
\end{align}
independent verified attempts give bounded error. If no candidate passes verification within the prescribed number of attempts, the algorithm reports failure. Both $L_{\rm pad}$ and $p_0^{-1}$ are constants for the fixed parameters, so the total query count is $\calO_{R,\mu,\beta}(M)$. When $\beta>\alpha_{\rm C}(R)$, \cref{thm:frs-separation}\textup{(iii)} with $\eta=\beta-\alpha_{\rm C}(R)$ supplies the corresponding distributional classical query lower bound.
\end{proof}

\paragraph{Resulting score envelope and scope.}

The semicircle curve $\alpha_{\rm sc}(\mu)$ is continuous and strictly increasing on $[0,1/2]$. For our folded Reed--Solomon family, \cref{lem:frs-complete-list-decoder} supports every fixed rational cutoff $\mu<R$. Taking the supremum over these cutoffs within the range of the score analysis gives
\begin{align}
    \sup_{\substack{\mu\in\mathbb Q,\ 0<\mu\leq1/2\\\mu<R}}
    \alpha_{\rm sc}(\mu)
    &=
    \begin{cases}
        \frac12+\sqrt{R(1-R)},&0<R\leq1/2,\\
        1,&1/2<R<1,
    \end{cases}
    =\alpha_{\rm fib}(R).
    \label{eq:cl-envelope}
\end{align}
For $R\leq1/2$, rational cutoffs can approach $R$ from below, but the endpoint $\mu=R$ is excluded by the decoder's requirement of positive slack $R-\mu$. Every fixed threshold strictly below the envelope can nevertheless be reached using a fixed positive slack. For $R>1/2$, the cutoff $\mu=1/2$ is allowed and gives limiting expected score one. This certifies every fixed threshold below one, but does not by itself guarantee an output satisfying every block; exact search requires the separate argument below.

\begin{proof}[Proof of \cref{thm:sampled-score-hierarchy}\textup{(i)}]
Fix a rational threshold $0<\beta<\alpha_{\rm fib}(R)$. We choose a rational cutoff $\mu$ supported by our decoder and satisfying $\alpha_{\rm sc}(\mu)>\beta$. If $R\leq1/2$, choose $0<\mu<R$ sufficiently close to $R$; continuity and $\beta<\alpha_{\rm sc}(R)$ ensure that this is possible. If $R>1/2$, take $\mu=1/2$, for which $\alpha_{\rm sc}(\mu)=1>\beta$.

In either case, $\mu<R$, so \cref{lem:frs-complete-list-decoder} provides complete syndrome lists up to block weight $\floor{\mu M}$, with list size $\calO_{R-\mu}(1)$ and polynomial running time. The cutoff is fixed independently of $M$, so both the decoding slack $R-\mu$ and the score gap $\alpha_{\rm sc}(\mu)-\beta$ are positive constants.

Applying \cref{prop:coherent-fiber-fixed-radius} therefore gives bounded-error success at threshold $\beta$ outside its exceptional set. Since the chosen cutoff depends only on $R$ and $\beta$, the algorithm uses $\calO_{R,\beta}(M)$ queries and polynomially many gates, with exceptional probability $\exp\{-\exp(\Omega_{R,\beta}(M\log M))\}$, as claimed.
\end{proof}

\subsection{Part \textup{(ii)}: exact search by truncated coherent fiber summation}
\label{sec:proof-sampled-exact}

For exact search, we begin with a uniform superposition of accepted symbols in each block. The resulting block words need not be codewords. Fourier transformation and coherent fiber summation convert this state into a distribution on valid candidate polynomials. We calculate directly the probability that one full attempt returns a polynomial accepted in every block.

\begin{proof}[Proof of \cref{thm:sampled-score-hierarchy}\textup{(ii)}]
Fix a rational rate $R\in(1/2,1)$ and a sufficiently large admissible block count $M$.

\emph{Choosing the retained labels.}
We use complete list decoding at relative block radius $1/2$. This is available because $R>1/2$: the required decoding slack $R-1/2$ is a positive constant. Every admissible $M$ is odd,  since $M+1$ is a power of two. Put
\begin{align}
    \tau&\coloneq\frac12,
    &k_M&\coloneq\floor{M/2}=\frac{M-1}{2},
    &L_{\rm exact}&\coloneq2^{\ceil{\log_2L_{\tau}}}.
    \label{eq:exact-radius-choice}
\end{align}
Here $L_{\rm exact}$ is the common power-of-two padding length for the decoder's lists. Define the retained set
\begin{align}
    \mathcal E
    &\coloneq\{y\in\Sigma^M:\operatorname{wt}_{\rm blk}(y)\leq k_M\}.
\end{align}
By \cref{lem:frs-complete-list-decoder}, the complete lists are $\mathcal L_{\sigma,[0,k_M]}$ from \eqref{eq:frs-complete-fib}. They can be computed in polynomial time, and $L_{\rm exact}=\calO_R(1)$.

\emph{Preparing the accepted-symbol state.}
For each fixed block $i$, \cref{lem:accepted-symbol-preparation} prepares
\begin{align}
    \ket0
    &\longmapsto
    \frac1{\sqrt{|\Sigma|/2}}\sum_{z\in A_i}\ket z
    \label{eq:exact-local-state-preparation}
\end{align}
using three queries to $\mathcal O_{A_i}$. Applying these routines separately to the $M$ blocks gives the normalized product state
\begin{align}
    \ket{W^{\mathcal A}}
    &\coloneq\sum_{z\in\Sigma^M}W^{\mathcal A}(z)\ket z,
    &W^{\mathcal A}(z)
    &\coloneq\prod_{i\in[M]}
    \frac{\mathds{1}_{\{z_i\in A_i\}}}{\sqrt{|\Sigma|/2}}.
    \label{eq:exact-product-weight}
\end{align}
This uses $3M$ membership queries. We account for circuit approximation errors at the end of the proof.

Apply a Fourier transform to each block register. Using the coefficients $\widetilde\phi_i$ from
\eqref{eq:app-fourier}, the local state becomes
\begin{align}
    \operatorname{QFT}_\Sigma
    \left(\frac1{\sqrt{|\Sigma|/2}}\sum_{z\in A_i}\ket z\right)
    &=\frac1{\sqrt2}\ket0
    +\frac1{\sqrt2}\sum_{\eta\ne0}\widetilde\phi_i(\eta)\ket\eta.
    \label{eq:exact-local-fourier-state}
\end{align}
Indeed, the acceptance indicator is $(1+f_i)/2$. The constant term supplies the zero error, while balance gives $\widetilde\phi_i(0)=0$ for the sign term. Parseval gives $\sum_{\eta\ne0}|\widetilde\phi_i(\eta)|^2=1$. Thus a measurement of this register would give zero with probability $1/2$ and a nonzero label with probability $1/2$.

The block registers are in a product state, so their zero/nonzero choices are independent. Write $\widehat W^{\mathcal A}(y)$ for the amplitude of label $y$ after these Fourier transforms. Then
\begin{align}
    \widehat W^{\mathcal A}(y)
    &=2^{-M/2}\prod_{i:y_i\ne0}\widetilde\phi_i(y_i),
    \label{eq:exact-prepared-fourier-amplitude}\\
    \sum_{\operatorname{wt}_{\rm blk}(y)=j}
    |\widehat W^{\mathcal A}(y)|^2
    &=2^{-M}\binom Mj.
    \label{eq:exact-binomial-layers}
\end{align}
In particular, the Fourier block weight has distribution $\operatorname{Bin}(M,1/2)$. The total squared amplitude on the retained labels is therefore
\begin{align}
    b_M &\coloneq2^{-M}\sum_{j=0}^{k_M}\binom Mj =\Prb[\operatorname{Bin}(M,1/2)\leq k_M].
    \label{eq:exact-truncated-binomial-mass}
\end{align}
Since $M$ is odd, symmetry pairs each weight at most $(M-1)/2$ with a weight at least $(M+1)/2$, with no middle layer left over. Consequently, $b_M=1/2$. Thus our decoder covers exactly half of the Fourier probability. This statement holds for every balanced promised oracle.

\emph{Computing the probability of a perfect output.}
Apply \cref{lem:coherent-fiber-primitive} to the full normalized state $\sum_y\widehat W^{\mathcal A}(y)\ket y$, with retained set $\mathcal E$ and padding length $L_{\rm exact}$. Here the procedure includes syndrome computation. Success in \cref{lem:coherent-fiber-primitive} is the joint outcome $(\mathrm{Idx},b_{\rm match})=(0,1)$ after the Hadamards on the index register. The match bit rejects labels outside $\mathcal E$, while projection of the padded index onto zero contributes the factor $L_{\rm exact}^{-1/2}$. We neither project onto $\mathcal E$ in advance nor rescale the retained amplitudes by $b_M^{-1/2}$.

Without measuring the success flag, apply $\operatorname{QFT}_V^{-1}$ to the syndrome register and consider the unnormalized successful component. By \eqref{eq:generic-fiber-message-amplitude}, its amplitude at $x\in V$ is
\begin{align}
    &\frac1{\sqrt{L_{\rm exact}|V|}}
    \sum_{y\in\mathcal E}
    \widehat W^{\mathcal A}(y)
    (-1)^{\langle x,B^Ty\rangle_V}.
    \label{eq:exact-pre-support-amplitude}
\end{align}
We evaluate this sum by grouping the retained labels according to their block support. For a fixed $T\subseteq[M]$, the adjoint identity and local Fourier inversion give
\begin{align}
    &\sum_{\operatorname{supp}_{\rm blk}(y)=T}
    \left(\prod_{i\in T}\widetilde\phi_i(y_i)\right)
    (-1)^{\langle x,B^Ty\rangle_V} =
    \prod_{i\in T}\left(
    \sum_{\eta\ne0}\widetilde\phi_i(\eta)
    (-1)^{\langle B_ix,\eta\rangle_\Sigma}\right)
    =\prod_{i\in T}f_i(B_ix).
\end{align}
Here $\operatorname{supp}_{\rm blk}(y)=\{i:y_i\ne0\}$. The equality uses $\langle x,B^Ty\rangle_V=\sum_{i\in T}\langle B_ix,y_i\rangle_\Sigma$, as in \eqref{eq:app-support-character-expansion}. No zero-frequency term is missing, since $\widetilde\phi_i(0)=0$.

Adding the contributions of all retained supports gives the truncated amplitude function
\begin{align}
    F^{\leq k_M}(x)
    &\coloneq
    2^{-M/2}
    \sum_{\substack{T\subseteq[M]\\|T|\leq k_M}}
    \prod_{i\in T}f_i(B_ix)
    =
    2^{-M/2}
    \sum_{j=0}^{k_M}
    \sqrt{\binom Mj}\,P_j(f(x)).
    \label{eq:exact-truncated-amplitude}
\end{align}
Here the second equality uses the same normalized elementary-symmetric expression as \eqref{eq:app-polynomial-dqi-layer}, now merely as algebraic notation for every $0\leq j\leq k_M$. No orthogonality or preparation claim for the corresponding higher-degree layer states is needed.

It follows that the amplitude of $\ket{x}$ in the unnormalized successful branch is
\begin{align}
    \frac{F^{\leq k_M}(x)}
    {\sqrt{L_{\rm exact}|V|}}.
    \label{eq:exact-successful-message-amplitude}
\end{align}
The factor $L_{\rm exact}^{-1/2}$ comes from coherent fiber summation, while $|V|^{-1/2}$ comes from the inverse Fourier transform.

If $Q_x$ satisfies every block, all the signs $f_i(B_ix)$ equal $+1$. Every support product then equals one, so
\begin{align}
    P_j(f(x))&=\sqrt{\binom Mj},
    &F^{\leq k_M}(x)&=2^{M/2}b_M.
\end{align}
All perfect polynomials therefore receive the same successful amplitude. Recall that $N_{\rm perf}(\mathcal A)$ counts their codewords and that $Z_{\rm perf}(\mathcal A)$ is this count divided by its expectation, as in \eqref{eq:exact-normalized-solution-count}. The encoder is injective, so $|V|=|C_{\rm FRS}|$ and perfect codewords correspond one-to-one to perfect coefficient vectors. Squaring the amplitude above and summing over them gives
\begin{align}
    p_{\rm perf}(\mathcal A)
    &=\frac{2^M b_M^2N_{\rm perf}(\mathcal A)}{L_{\rm exact}|V|}
    =\frac{b_M^2Z_{\rm perf}(\mathcal A)}{L_{\rm exact}}.
    \label{eq:exact-direct-success-identity}
\end{align}
Here $p_{\rm perf}(\mathcal A)$ is the joint probability that one ideal attempt reports fiber success and produces a perfect polynomial. Because we used the unnormalized branch, this probability already accounts for attempts in which fiber summation fails. There is no need to calculate the distribution conditional on fiber success.

\emph{Typical oracles, implementation, and verification.}
The remaining oracle-dependent quantity is $Z_{\rm perf}$. The concentration result from \cref{lem:typical-perfect-satisfiability} shows that it is at least $1/2$ on almost all sampled oracles. Define
\begin{align}
    \mathsf{Good}_M
    &\coloneq\{\mathcal A:Z_{\rm perf}(\mathcal A)\geq1/2\}.
    \label{eq:exact-good-oracle-event}
\end{align}
Under the balanced sampled-oracle distribution,
\begin{align}
    \Prb[\mathcal A\notin\mathsf{Good}_M]
    &\leq\exp\{-\exp(\Omega_R(M\log M))\}.
    \label{eq:exact-good-oracle-probability}
\end{align}
Fix any oracle in $\mathsf{Good}_M$. Since $b_M=1/2$, \eqref{eq:exact-direct-success-identity} gives $p_{\rm perf}(\mathcal A)\geq1/(8L_{\rm exact})$. Probabilities in the remainder of this argument are over the quantum algorithm for this fixed oracle.

We may defer measuring the fiber-success flag until after the inverse Fourier transform, since they act on different registers. Compile the gates, including those of the final inverse Fourier transform, using only the parameters and target accuracy, independently of $\mathcal A$. Choose the accuracy so that, together with the exact membership oracle calls, the implemented normalized joint state before measurement differs from the ideal state by at most $1/(32L_{\rm exact})$ in vector norm, uniformly over all promised oracles. Include all auxiliary registers in this comparison. As in the preceding proof, the probability of any measurement event changes by at most twice the vector error. Applying this to the joint event of fiber success and a perfect candidate gives
\begin{align}
    \widetilde p_{\rm perf}(\mathcal A)
    &\geq\frac1{8L_{\rm exact}}-
    \frac2{32L_{\rm exact}}
    =\frac1{16L_{\rm exact}}.
    \label{eq:exact-single-run-success}
\end{align}
Since $L_{\rm exact}=\calO_R(1)$, this precision can be achieved with polynomial gate overhead.

Whenever fiber summation succeeds, measure the candidate and verify it by querying all $M$ blocks. Return it only if every block accepts. One full attempt therefore returns a perfect polynomial with probability at least $1/(16L_{\rm exact})$. Repeat independently, with the same fixed oracle, for
\begin{align}
    n&\coloneq\ceil{16L_{\rm exact}\log3}
\end{align}
attempts, stopping when verification succeeds. If no candidate passes verification within the prescribed number of attempts, the algorithm reports failure. The probability that all attempts fail is at most
\begin{align}
    \left(1-\frac1{16L_{\rm exact}}\right)^n
    &\leq\exp\left(-\frac{n}{16L_{\rm exact}}\right)
    \leq\frac13.
\end{align}
Each attempt uses $3M$ fixed-block membership queries for state preparation and at most $M$ for verification; list decoding and coherent fiber summation are oracle-free. The remaining gates have polynomial complexity. Since $L_{\rm exact}=\calO_R(1)$, the predetermined number of attempts is $\calO_R(1)$. Hence the worst-case total query count is $\calO_R(M)$ and the running time is polynomial. Every returned polynomial is perfect and hence has approximation ratio one.
\end{proof}

\subsection{Yamakawa--Zhandry exact search and its DQI interpretation}
\label{sec:yz-dictionary}

The ideal target of Yamakawa and Zhandry's exact-search algorithm can be interpreted as a DQI filter that keeps only candidates satisfying every block. We explain this connection and record their oracle interface, which is useful for comparing the two settings.

\paragraph{Their oracle and search task.}
Yamakawa and Zhandry consider folded Reed--Solomon codewords $c(Q)\in\Sigma^M$, where $\Sigma=\mathbb F_q^h$ is the alphabet of one block, and a uniformly random vector-output oracle $G$~\cite{YamakawaZhandry2024}. Write
\begin{align}
    G:\Sigma&\rightarrow\{0,1\}^M,
    &G_i(z)&\coloneq G(z)_i.
\end{align}
Their task is to find a polynomial $Q$ whose codeword satisfies $G_i(c_i(Q))=1$ for every block $i\in[M]$. Thus each coordinate function $G_i$ determines which symbols are accepted in block $i$, and the goal is to obtain exact score one.

Their quantum oracle is
\begin{align}
    \calO_G\ket{z,y}
    &\coloneq\ket{z,y\oplus G(z)},
    \qquad y\in\{0,1\}^M.
    \label{eq:yz-vector-output-oracle}
\end{align}
One query returns all $M$ coordinate bits at the same symbol address $z$. There is no block-index register in this interface. For each fixed $i$, the corresponding coordinate oracle is
\begin{align}
    \calO_{G_i}\ket{z,t}
    &\coloneq\ket{z,t\oplus G_i(z)},
    \qquad t\in\{0,1\}.
    \label{eq:yz-coordinate-oracles}
\end{align}
This has the same access form as our separately callable block-membership oracle $\calO_{A_i}$.

The two interfaces can simulate one another with polynomial query overhead. A vector-output query uses $M$ coordinate queries, one for each output bit. Conversely, a query to a fixed coordinate uses two vector-output queries: compute $G(z)$, copy its $i$th bit into the answer qubit, and uncompute $G(z)$. These simulations compare access to the same predicates; the distributions of the predicates are a separate issue.

\paragraph{The polynomial-filter interpretation.}
DQI uses a polynomial filter to favor high-scoring solutions. An exact-search filter makes the more restrictive choice of assigning zero amplitude to every imperfect candidate and equal amplitude to every perfect one. Such a filter can also be expressed as a polynomial: the number of accepted blocks takes values $0,\ldots,M$, and a degree-$M$ polynomial can equal zero at $0,\ldots,M-1$ and one at $M$. The accepted-block count is an affine function of the signed objective, so this is a DQI polynomial filter as well. Whenever perfect solutions exist, its normalized state is their uniform superposition. Encoding the coefficient vectors gives the corresponding uniform superposition of perfect codewords.

At a high level, Yamakawa and Zhandry prepare a superposition over codewords together with a product superposition over locally accepted symbols, pass to the Fourier basis, and use decoding of the dual code to combine these two constraints coherently~\cite{YamakawaZhandry2024}. For typical random oracles, the resulting state is close to one having equal amplitude on all codewords accepted in every block and zero amplitude elsewhere. Their ideal target can therefore be viewed as the degree-$M$ exact-satisfaction DQI filter state, a connection already discussed by Jordan et al.~\cite{JordanEtAl2025DQI}. The degree-$M$ description characterizes the target amplitudes, while efficient preparation depends on the available decoding procedure. Our DQI separation instead uses a lower-degree filter together with conventional unique decoding, whereas our exact-search construction in \cref{sec:proof-sampled-exact} analyzes the verified perfect-output probability directly through truncated Fourier amplitudes and coherent fiber summation.

\paragraph{Comparison with our oracle distribution.}
For uniformly random $G$, the bits $G_i(z)$ are mutually independent over all pairs $(i,z)$. Consequently, the acceptance sets $T_i^G\coloneq\{z\in\Sigma:G_i(z)=1\}$ have fluctuating cardinalities. Our sampled acceptance sets are independent across blocks, but each is uniform among subsets of cardinality exactly $|\Sigma|/2$. This is the independent-bit distribution conditioned on exact balance in every block, and it introduces dependence between membership bits within each block. The oracle-interface simulations above do not remove this distributional difference. In particular, the exact binomial Fourier-weight law in \eqref{eq:exact-binomial-layers} follows from our balance promise and holds for every promised oracle.

\end{document}